\documentclass[a4paper,11pt]{article}
\usepackage{graphicx}
\usepackage{here}
\usepackage{amsmath,amssymb,amsthm}
\usepackage{mathrsfs}
\usepackage{ascmac}
\usepackage[sort, comma]{natbib}
\usepackage[subrefformat=parens]{subcaption}
\usepackage{url}
\usepackage{tikz}
\usetikzlibrary{intersections,calc}
\usepackage{mathtools}
\usepackage{physics}
\mathtoolsset{showonlyrefs=true}
\usepackage[stable]{footmisc}
\usepackage{hyperref}

\newtheorem{theo}{Theorem}
\newtheorem{lem}{Lemma}

\newcommand{\ol}[1]{\overline{#1}}

\newcommand{\cl}[1]{\mathcal{#1}}

\title{
{\Large A footloose entrepreneur model in a continuous space}
}
\author{Kensuke Ohtake\thanks{Center for General Education, Shinshu University, Matsumoto, Nagano 390-8621, Japan,
E-mail: k\_ohtake@shinshu-u.ac.jp}}
\date{September 24, 2026}

\begin{document}
\maketitle

\begin{abstract}
This paper studies a footloose entrepreneur model in a continuous space. In an appropriate function space, the model is formulated as an initial value problem for an infinite-dimensional ordinary differential equation. A unique global solution is constructed based on the Banach fixed point theorem. The stability of a homogeneous stationary solution is then investigated, and numerical simulations of the asymptotic behavior of global solutions under different parameter values are performed. Numerical solutions starting near the unstable homogeneous stationary solution converge to spike-shaped non-homogeneous stationary solutions, and the number of spikes decreases with declining transport costs. An analysis of the stability of non-homogeneous stationary solutions supports this observation.
\end{abstract} 

\noindent
{\bf Keywords:\hspace{1mm}}
core-periphery model;
differential equations;
dynamical systems;
footloose entrepreneur model;
functional equations;
global solutions;
integral equations;
new economic geography;
self-organization


\section{Introduction}

\citet{Krug91} proposed a general equilibrium model for understanding spatial inequality caused by the agglomeration of households and firms. This model is known as the core-periphery (CP) model, which is a fundamental model in new economic geography (NEG). In the CP model, two industries are assumed: manufacturing, which produces various differentiated goods under monopolistic competition, and agriculture, which produces a single variety of goods under perfect competition. Firms in the manufacturing sector have increasing returns technology, while those in the agricultural sector have constant returns. Mobile workers are utilized as both fixed and marginal inputs in the manufacturing sector, whereas immobile workers are primarily engaged in agricultural production. Workers, as consumers, have a preference for variety in manufactured goods. That is, they desire a greater variety of manufactured goods. Transport costs are incurred for the transportation of manufactured goods, whereas they are not incurred for the transportation of agricultural goods. Based on these assumptions, in the CP model, agglomeration occurs through the following mechanism. First, due to increasing returns and transport costs, firms tend to agglomerate in regions with higher densities of households (i.e., workers). This is advantageous because it allows them to save on production and transport costs. Second, households tend to agglomerate in regions where more firms are concentrated, driven by their preference for variety and lower transport costs. This is because it allows them to enjoy a richer variety at lower prices.

Mathematically, the model comprises nonlinear algebraic and differential equations. In particular, the nonlinear algebraic equations that give market equilibrium (called instantaneous equilibrium) for a given population distribution make the analytical handling of the model difficult.\footnote{There are many studies on the CP model in two, three, or more discrete-regional cases. For example, we can mention \citet{LanaSan}, \citet{MaffTrio}, \citet{Rob}, \citet{CurrKub}, \citet{Moss06}, \citet{LanQua}, \citet{IkeAkaKon}, \citet{IoaIoa}, \citet{BarbZof}, \citet{SaYa2023}, and \citet{KaYa2025}} To circumvent this difficulty, several more tractable models were developed. The first approach is to make the utility function of consumers quasilinear, as in \citet{OttaTabThi} and \citet{Pfl}. This greatly simplifies the instantaneous equilibrium equation. Of course, this approach sacrifices the generality of consumer behavior by not considering income effects.\footnote{However, as their results suggest, the properties of the solutions to these models would not be qualitatively affected by the lack of income effects in two regional cases. Regarding the Pfl{\"u}ger model, studies of the Krugman model and the Pfl{\"u}ger model extended in continuous space (\citet{OhtakeYagi_point} and \citet{Ohtake2023cont}) also support this. Regarding the Ottaviano-Tabuchi-Thisse model, the properties of the solution change significantly in a multi-regional case. In fact, \citet{Ohtake2022agg} shows that the homogeneous stationary solution is always unstable when the number of regions is a multiple of 4.} The second approach, as proposed by \citet{ForsOtta}, is to change the assumption so that production in the manufacturing sector requires not only mobile but also immobile workers. To be more precise, mobile workers are assumed to be fixed inputs, and immobile workers are assumed to be marginal inputs. This allows the instantaneous equilibrium equation to be linear. The Forslid-Ottaviano model is often referred to as the {\it Footloose Entrepreneur {\rm (FE)} model} because mobile workers can be interpreted as high-skilled workers or entrepreneurs needed in the manufacturing sector.\footnote{The name ``Footloose Entrepreneur model'' comes from \citet{Baletal2003}} The FE model is significant in that it makes the model tractable without sacrificing the generality in consumer behavior. In addition, the properties of the solution are not qualitatively different from the Krugman model in the two-regional case as shown in \citet{ForsOtta}.\footnote{\citet{Rob} discusses the isomorphism of various models, including the CP model and the FE model.}

This paper discusses the dynamics of solutions to the FE model in continuous space, going beyond the limited situation of the two-regional case. The model comprises a Fredholm linear integral equation that describes instantaneous equilibrium and a replicator equation that describes the dynamics of population migration. We begin by formulating the model as an initial value problem for an infinite-dimensional ordinary differential equation in an appropriate function space. Application of the Banach fixed point theorem and the Picard–Lindel\"{o}f theorem constructs a unique global solution to the model. When the periodic boundary condition is imposed in a one-dimensional space (i.e., racetrack economy), this equation has a homogeneous stationary solution, in which the population density and all other variables are uniformly distributed throughout the space. When a small perturbation added to the homogeneous stationary solution is decomposed into eigenfunctions, the eigenfunctions can be classified into those whose amplitudes grow and those whose amplitudes decay over time. It is found that the spatial frequency of the former becomes smaller as transport costs decline. Finally, we simulate the time evolution of model's solutions to obtain non-homogeneous numerical stationary solutions. The numerical stationary solutions thus obtained have several spikes, and it is observed that the number of spikes decreases as transport costs decline.

The topics addressed in this study are primarily theoretical, and the author would like to articulate their significance for spatial economics. First, the unique existence of solutions to dynamical systems is important as a foundation for the understanding of the mechanisms underlying phenomena. If solutions were not unique, it would become impossible to determine which solution to choose as the theoretical consequence that should be verified against the real world. In this case, even if a specific solution is chosen arbitrarily, and the solution provides a good description of the real world, we cannot understand the phenomenon because internal mechanisms of the model cannot explain why that solution is chosen. Second, we adopt somewhat unrealistic assumptions. That is, a uniform continuous space, and symmetric continuous transport costs across regions. In this setting, the heterogeneity of real-world space, as well as the asymmetries and the discontinuities in transport costs, are not taken into account. This does not mean that this paper is not concerned with the real world; rather, its purpose is to understand the fundamental mechanisms of economic agglomeration, which are obscured by myriad factors of reality. Thus, this study investigates how heterogeneous spatial structures emerge from homogeneous and symmetric spatial conditions.\footnote{In other words, we focus on the emergence of the ``second nature'' (\citet{Cronon91} and \citet{Krug93fnsn}). This can also be viewed as research into self-organizing economy \citep{Krug}. For this purpose, a continuous space is considered preferable to a discrete space. The reason is that, in a discrete space, there is an asymmetry between locations where economic factors can be located and those where they cannot.} The author expects that a deeper understanding of mechanisms of agglomeration under theoretically pure conditions enhances our understanding of real-world economic agglomeration and open up new possibilities for policy applications.

Let us review closely related works to the present paper. Research on the original CP model in continuous space includes \citet[Chapter 6]{FujiKrugVenab}, \citet{ChiAsh08}, \citet{TabaEshiSakaTaka}, \citet{TabaEshi_existence}, \citet{TabaEshi_explosion}, \citet{TabaEshi23}, \citet{OhtakeYagi_Asym}, and \citet{OhtakeYagi_point}. The FE model in a continuous space, same as the present paper, was considered in \citet[Appendix E]{Fab}. There, it has been demonstrated that the model exhibits spiky stationary solutions, which are consistent with the findings of this study. \citet{IkeMuroAkaKoTa}, \citet{IkeOnTaka2019}, \citet{AiIkeOsa2020}, and \citet{Gasetal2020FE} discussed the FE model in a discretely multi-regional case. The mathematically rigorous discussions of the CP model in a continuous space can be found in \citet{TabaEshiSakaTaka},  \citet{TabaEshi_existence}, and \citet{TabaEshi23}. In particular, the functional analysis techniques used in \citet{TabaEshiSakaTaka} have been of great help in this paper.

The rest of the paper is organized as follows. Section \ref{sec:model} introduces the model. Section \ref{sec:globalsol} formulates the model as an ordinary differential equation in a Banach space and constructs a unique global solution to it. Section \ref{sec:racetrack} sets up the model on a one-dimensional periodic space. Section \ref{sec:instability} investigates the stability of the homogeneous stationary solution. Section \ref{sec:numerical} presents the results of numerical simulations of the asymptotic behavior of solutions to the model. Section \ref{sec:stabilityspikysol} investigates the stability of non-homogeneous spiky stationary solutions. Section \ref{sec:cd} provides conclusions and discussion. Section \ref{sec:appendix} gives some content omitted in the main text.

\section{The model}\label{sec:model}
In this section, we derive the FE model based on the Dixit-Stiglitz framework\footnote{\citet{DS77}.} in continuous space. The main assumptions are as follows:
\begin{itemize}
\item The geographic space is continuous.
\item The economy has two sectors: manufacturing and agriculture. The manufacturing sector produces a large number of varieties of differentiated goods in monopolistic competition. The agricultural sector produces homogeneous goods in a perfectly competitive market.
\item In the manufacturing sector, each firm is engaged in the production of one variety.
\item All firms have the same production technology, in which production requires only labor as an input.
\item Transport costs arise when manufactured goods are transported. The transport costs are assumed to be of the iceberg type, as formulated by \citet{Sam1952}.
\item Transportation of agricultural goods does not incur transport costs.
\item Workers are divided into two types: mobile and immobile workers. Mobile workers (or entrepreneurs) can migrate among regions, while immobile workers cannot.
\item The labor of mobile workers is fixed input in the manufacturing sector.
\item The labor of immobile workers is marginal input in the manufacturing sector and the agricultural sector.
\item The nominal wage of the immobile workers is assumed to be one as num\'{e}raire.
\end{itemize}

\subsection{Consumer behavior \footnote{The discussion in this subsection follows \citet[pp.46-48]{FujiKrugVenab}.}}

There are differentiated varieties of manufactured goods, specified by continuous indices $i\in[0,n]$, and one homogeneous variety of agricultural goods. Let $q(i)$ be the consumption of variety $i$ of manufactured goods and $A$ be the consumption of agricultural goods. Suppose that the preference of a consumer is expressed by the utility function
\begin{equation}\label{utility}
U=M^\mu A^{1-\mu},~\mu\in(0,1),
\end{equation}
where 
\begin{equation}\label{composite}
M = \left[\int_0^nq(i)^{\frac{\sigma-1}{\sigma}}\right]^{\frac{\sigma}{\sigma-1}}.
\end{equation}
Here, $\sigma>1$ stands for the elasticity of substitution between two varieties of manufactured goods. In other words, the closer $\sigma$ is to $1$, the stronger the preference for variety of consumers, and the larger the value of $\sigma$, the weaker the preference for variety of consumers. The budget constraint of a consumer is given by 
\begin{equation}\label{budget}
p_A A + \int_0^n p(i)q(i)di = Y,
\end{equation}
where $p_A$ is the price of the agricultural good, $p(i)$ is the price of the $i$-th variety of manufactured goods, and $Y$ is nominal income of the consumer.

A two-stage budgeting is used to maximize the utility function \eqref{utility} under the budget constraint \eqref{budget} when $p_A$, $\left\{p(i)\right\}_{i\in[0,n]}$, and $Y$ are given. We first have to solve the minimization problem
\begin{align}
&\min_{q} \int_0^n p(i)q(i)di
\end{align}
subject to \eqref{composite}. The first order condition is
\begin{equation}\label{firstorder}
\frac{q(i)^{-\frac{1}{\sigma}}}{q(j)^{-\frac{1}{\sigma}}} = \frac{p(i)}{p(j)}.
\end{equation}
Therefore,
\begin{equation}\label{demandforqi}
q(i) = q(j)\left[\frac{p(j)}{p(i)}\right]^\sigma.
\end{equation}
Substituting \eqref{demandforqi} into \eqref{composite}, we have
\begin{equation}\label{demandforqj}
q(j) = \frac{p(j)^{-\sigma}}{\left[\int_0^n p(i)^{1-\sigma}di\right]^{\frac{\sigma}{\sigma-1}}}M.
\end{equation}
By defining the price index as 
\begin{equation}\label{priceindex}
G = \left[\int_0^n p(i)^{1-\sigma}di\right]^{\frac{1}{1-\sigma}},
\end{equation}
we see that \eqref{demandforqj} becomes
\begin{equation}\label{demandforqjGM}
q(j) = \left[\frac{p(j)}{G}\right]^{-\sigma}M.
\end{equation}
Substituting \eqref{demandforqjGM} into \eqref{budget}, we have
\begin{equation}\label{pqGM}
\int_0^n p(i)q(i)di = GM,
\end{equation}
which states that the price index \eqref{priceindex} can be interpreted as the price of one unit of the composite index of manufactured goods \eqref{composite}. Then, the next problem to be solved is
\begin{align}
&\max~U = M^{\mu}A^{1-\mu}\\
&\hspace{1mm}\text{s.t.}\hspace{2mm} p_AA + GM = Y
\end{align}
This yields that
\begin{align}
&M = \mu\frac{Y}{G},\label{demandforM}\\
&A = (1-\mu)\frac{Y}{p_A}.\label{demandforA}
\end{align}
Substituting \eqref{demandforM} into \eqref{demandforqjGM}, we obtain the compensated demand function for the $j$-th variety
\begin{equation}\label{qjmuYpjmsigG1msig}
q(j) = \mu Y\frac{p(j)^{-\sigma}}{G^{1-\sigma}}
\end{equation}
Substituting \eqref{demandforM} and \eqref{demandforA} into \eqref{utility}, we obtain the indirect utility 
\begin{equation}\label{indirectu}
V = \mu^\mu(1-\mu)^{1-\mu}YG^{-\mu}p_A^{\mu-1}.
\end{equation}

\subsection{Producer behavior \footnote{The discussion in this subsection is based on not only \citet{ForsOtta}, but also \citet[Chapter 3 and 7]{ZenTaka}}}
Let $c_{\rm f}$ and $c_{\rm m}$ represent fixed cost and marginal cost expressed in nominal terms, respectively. This means that the nominal cost to achieve output $q$ of any variety of manufactured goods is given by
\[
c_{\rm f}+c_{\rm m}q.
\]
Then, the profit of a firm that produces the $i$-th variety is given by
\begin{equation}\label{profit}
\Pi(i) = p(i)q(i)-(c_{\rm f}+c_{\rm m}q(i)).
\end{equation}
By differentiating \eqref{qjmuYpjmsigG1msig} by $p(i)$, holding $G$ fixed, we have
\begin{equation}\label{dqdp}
\frac{d}{dp(i)}q(i) = -\sigma\frac{q(i)}{p(i)}.
\end{equation}
By \eqref{profit} and \eqref{dqdp}, maximizing the profit \eqref{profit} yields
\begin{equation}\label{firstorderpi}
p(i) = \frac{\sigma}{\sigma-1}c_{\rm m}.
\end{equation}
Substituting \eqref{firstorderpi} into \eqref{profit}, we obtain the maximized profit
\begin{equation}\label{profitcmsigm1qicf}
\Pi(i) = \frac{c_m}{\sigma-1}q(i)-c_{\rm f}
\end{equation}
and assuming the zero-profit due to free entry and exit, we see that
\begin{equation}\label{cmsigm1qi}
c_{\rm f} = \frac{c_{\rm m}}{\sigma-1}q(i).
\end{equation}
From \eqref{firstorderpi} and \eqref{cmsigm1qi}, we have
\begin{equation}\label{pisigqi}
c_{\rm f} = \frac{p(i)}{\sigma}q(i).
\end{equation}

\subsection{Continuous space modeling}
Now that we have made the above preparations, we derive a continuous space model. Let $\Omega\subset\mathbb{R}^n$ be a bounded closed domain where economic activity takes place. The regions in $\Omega$ are indicated by the space variables $x$, $y$, or $z\in \Omega$. As in \citet{ForsOtta}, it is assumed that production requires $F$ units of mobile workers as fixed inputs and $\frac{\sigma-1}{\sigma}$ units of immobile workers as marginal inputs. The nominal wage of mobile workers in region $x$ is $w(x)$, and the nominal wage of immobile workers is assumed to be one regardless of region. Therefore, $c_f$ and $c_m$ in \eqref{profit} are 
\begin{align}
&c_f = w(x)F,\label{cfwrF}\\
&c_m = \frac{\sigma-1}{\sigma}.\label{cmsigm1sig}
\end{align}
Let $q(x)$ be the demand for a variety of manufactured goods produced by one firm in region $x$. From \eqref{cmsigm1qi} and \eqref{cmsigm1sig}, we see that
\begin{equation}\label{cf1sigqi}
c_f = \frac{1}{\sigma}q(x),
\end{equation}
Let $p(x,y)$ be the price in region $y$ of a variety of manufactured goods produced in region $x$. The iceberg transport technology is expressed by a continuous function $T$ on $\Omega\times\Omega$. That is, for any variety of manufactured goods, $T(x,y)\geq 1$ units of them must be shipped to transport one unit of manufactured goods from $x$ to $y$. Therefore,
\begin{equation}\label{pxypxTxy}
p(x,y) = p(x)T(x,y).
\end{equation}
The transport function $T$ has a minimum and a maximum value in $\Omega\times \Omega$.
\begin{equation}\label{Tbound}
T_{\rm min}\leq T(x,y) \leq T_{\rm max},\hspace{5mm}\forall x,y\in\Omega
\end{equation}
From \eqref{pisigqi} and \eqref{cf1sigqi}, we have 
\begin{equation}\label{pr1}
p(x) \equiv 1.
\end{equation}
It follows immediately from \eqref{pxypxTxy} and \eqref{pr1} that 
\begin{equation}\label{pxyTxy}
p(x,y) = T(x,y).
\end{equation}
From \eqref{cfwrF} and \eqref{cf1sigqi}, we have
\begin{equation}\label{qxsigwxF}
q(x) = \sigma F w(x).
\end{equation}

We now derive model equations for an instantaneous equilibrium. Let $n(x)$ and $\lambda(x)$ be the density of the number of varieties of manufactured goods produced in region $x$ and the density of the mobile population at $x$, respectively. Given a total mobile population of $\Lambda> 0$, the population density satisfies
\begin{equation}\label{totalmobpop}
\int_\Omega\lambda(x)dx=\Lambda.
\end{equation}
Since it is assumed that a firm is engaged in the production of only one variety of manufactured goods, and each firm requires $F$ units of mobile workers, we see that
\begin{equation}\label{nx}
n(x) = \frac{\lambda(x)}{F}.
\end{equation}
By using \eqref{priceindex} and \eqref{pxyTxy}, we have
\begin{align}
G(x) &= \left[\int_\Omega n(y)p(x,y)^{1-\sigma}dy\right]^{\frac{1}{1-\sigma}}\\
&= \left[\int_\Omega n(y)T(x,y)^{1-\sigma}dy\right]^{\frac{1}{1-\sigma}}.\label{GxsumnyTxy1msig11msig}
\end{align}
Substitutiong \eqref{nx} into \eqref{GxsumnyTxy1msig11msig}, we obtain the price index equation
\begin{equation}\label{priceindexeq}
G(x) = \left[\frac{1}{F}\int_\Omega \lambda(y)T(x,y)^{1-\sigma}dy\right]^{\frac{1}{1-\sigma}}.
\end{equation}
Let $Y(x)$ be the density of total income in region $x$, which is given by
\begin{equation}\label{incomex}
Y(x) = w(x)\lambda(x) + \phi(x),
\end{equation}
where $\phi$ is the density of immobile workers whose nominal wage is assumed to be one. Given a total immobile population of $\Phi> 0$, the deisnity satisfies that
\begin{equation}\label{totalimpop}
\int_\Omega\phi(x)dx=\Phi.
\end{equation}
From \eqref{qjmuYpjmsigG1msig} and \eqref{incomex}, the demand from region $y$ for a variety of manufactured goods produced in region $x$ is given by
\begin{equation}\label{demandfromytox}
q(x,y) = \mu p(x,y)^{-\sigma}Y(y)G(y)^{\sigma-1}.
\end{equation}
To cover the transport costs, a firm at $x$ must produce $T(x,y)$ times as much as \eqref{demandfromytox}. Therefore, with \eqref{pxyTxy}, the total demand for a variety of manufactured goods produced in region $x$ is given by
\begin{align}
q(x) &= \int_\Omega q(x,y)T(x,y)dy\\
&= \mu\int_\Omega Y(y)G(y)^{\sigma-1}T(x,y)^{1-\sigma}dy.\label{totaldemandforx}
\end{align}
By \eqref{qxsigwxF} and \eqref{totaldemandforx}, we obtain the nominal wage equation
\begin{equation}\label{nominalwageeq}
w(x) = \frac{\mu}{\sigma F}\int_\Omega Y(y)G(y)^{\sigma-1}T(x,y)^{1-\sigma}dy.
\end{equation}
The nominal wage and the price index of instantaneous equilibrium are given by solutions of \eqref{priceindexeq} and \eqref{nominalwageeq}.

We assume that the migration dynamics of mobile workers are driven by spatial disparities in real wages, as in \citet[p.62]{FujiKrugVenab}. The price $p_A$ of the agricultural good is proved to be one.\footnote {Under the assumptions of perfect competition and constant-return technology, the wage of immobile workers ($\equiv 1$) and the price $p_A$ of the agricultural good are equal in equilibrium. See \citet[pp.~192-193,~pp.~217-218]{HayashiMicro2021} for details.} Then, from \eqref{indirectu}, it is natural to define the real wage in $x$ of mobile workers as 
\begin{equation}\label{realwage}
\omega(x) = w(x)G(x)^{-\mu}.
\end{equation}
Then, the dynamics of the mobile population are 
\begin{equation}\label{dynamics}
\frac{\partial}{\partial t}\lambda(t,x) = v\left[\omega(t,x) - \frac{1}{\Lambda}\int_{\Omega} \omega(t,y)\lambda(t,y)dy\right]\lambda(t,x),
\end{equation}
where $v>0$. In the right-hand side, $\frac{1}{\Lambda}\int_{\Omega} \omega(t,y)\lambda(t,y)dy$ stands for the spatial average of real wages. Taking into account that $Y$, $G$, $w$, $\omega$ and $\lambda$ are also functions of time $t\geq 0$, the equations \eqref{priceindexeq}, \eqref{incomex}, \eqref{nominalwageeq}, \eqref{realwage}, and \eqref{dynamics} can be reorganized as the system
\begin{equation}\label{1}
\left\{
\begin{aligned}
&Y(t,x) = w(t,x)\lambda(t,x) + \phi(x),\\
&G(t,x) = \left[\frac{1}{F}\int_\Omega \lambda(t,y)T(x,y)^{1-\sigma}dy\right]^{\frac{1}{1-\sigma}},\\ 
&w(t,x) = \frac{\mu}{\sigma F}\int_\Omega Y(t,y)G(t,y)^{\sigma-1}T(x,y)^{1-\sigma}dy,\\
&\omega(t,x) = w(t,x)G(t,x)^{-\mu},\\
&\frac{\partial}{\partial t}\lambda(t,x) = v\left[\omega(t,x) - \frac{1}{\Lambda}\int_{\Omega} \omega(t,y)\lambda(t,y)dy\right]\lambda(t,x)
\end{aligned}\right.
\end{equation}
for $(t,x) \in [0,\infty)\times \Omega$ with the initial condition
\begin{equation}\label{inival}
\lambda(0, x)=\lambda_0(x).
\end{equation}

\section{Constructing a unique global solution}\label{sec:globalsol}
In this section, we construct a unique global solution to the model \eqref{1} with the initial condition \eqref{inival}.
\subsection{Formulation as an infinite-dimensional ODE}
We formulate the equations \eqref{1} and \eqref{inival} as an initial value problem for an ordinary differential equation (ODE) in a Banach space. As basic function spaces, we adopt the space $L^1(\Omega)$ of integrable functions on $\Omega$ with the norm
\begin{equation}
\left\|f\right\|_{L^1}:=\int_\Omega|f(x)|dx
\end{equation}
and the space $C(\Omega)$ of continuous functions on $\Omega$ with the norm
\begin{equation}
\left\|f\right\|_{\infty}:=\max_{x\in\Omega}|f(x)|.
\end{equation}
In addition, we introduce the subsets
\[
L^1_\Lambda(\Omega) := \left\{f\in L^1(\Omega)\middle|f\geq 0~\text{a.e.}~\Omega,~\int_\Omega f(x)dx=\Lambda\right\}
\]
and
\[
C_{0+}(\Omega) := \left\{f\in C(\Omega)\middle|f\geq 0~\text{on}~\Omega\right\}.
\]

Let us introduce several operators. We begin with defining an operator $G:\lambda\in L^1(\Omega)\mapsto G[\lambda]\in C_{0+}(\Omega)$ such that
\begin{equation}\label{opG}
G[\lambda](x) := \left[\frac{1}{F}\int_\Omega\left|\lambda(y)\right|T(x,y)^{1-\sigma}dy\right]^{\frac{1}{1-\sigma}},\hspace{3mm}x\in\Omega,
\end{equation}
which corresponds to \eqref{priceindexeq}. By using \eqref{opG}, we introduce the fixed point problem which corresponds to \eqref{nominalwageeq} with \eqref{priceindexeq}
\begin{equation}\label{fp}
W(x) = \frac{\mu}{\sigma F}\int_\Omega\left[W(y)\left|\lambda(y)\right|+\phi(y)\right]G[\lambda](y)^{\sigma-1}T(x,y)^{1-\sigma}dy,\hspace{1mm}x\in\Omega,
\end{equation}
as for $W\in C(\Omega)$ for given functions $\phi\in L^1_\Phi(\Omega)$ and $\lambda\in L^1(\Omega)$.\footnote{The integral on the right-hand side of \eqref{fp} is well defined.} As we will see later in Theorem \ref{th:fp}, this fixed point problem \eqref{fp} has a unique solution $W\in C_{0+}(\Omega)$. Then, we can define the operator $w:L^1(\Omega)\to C_{0+}(\Omega)$ as 
\begin{equation}\label{opw}
w[\lambda] := \text{the unique fixed point $W\in C_{0+}(\Omega)$ of \eqref{fp} for $\lambda\in L^1(\Omega)$}.
\end{equation}
By using \eqref{opG} and \eqref{opw}, we define $\omega:L^1(\Omega)\to C_{0+}(\Omega)$ as
\begin{equation}\label{opomega}
\omega[\lambda](x) := w[\lambda](x)G[\lambda](x)^{-\mu},\hspace{3mm}x\in\Omega,
\end{equation}
which correspons to \eqref{realwage}. By using \eqref{opomega}, we can define $\tilde{\omega}:L^1(\Omega)\to \mathbb{R}$ as
\begin{equation}\label{opaverageomega}
\tilde{\omega}[\lambda]:=\frac{1}{\Lambda}\int_\Omega\omega[\lambda](y)\lambda(y) dy
\end{equation}
and $\Psi:L^1(\Omega)\to L^1(\Omega)$ as
\begin{equation}\label{opPsi}
\Psi[\lambda](x) := v\left[\omega[\lambda](x)-\tilde{\omega}[\lambda]\right]\lambda(x),\hspace{3mm}x\in\Omega,
\end{equation}
which corresponds to the right-hand side of \eqref{dynamics}. At last, by using \eqref{opPsi}, the dynamics of the system \eqref{1} can be formulated as an initial value problem of an ODE in $L^1(\Omega)$ such that
\begin{equation}\label{ode}
\left\{
\begin{aligned}
&\frac{d\lambda}{dt} (t) = \Psi[\lambda(t)],\\
&\lambda(0)=\lambda_0\in L^1_\Lambda(\Omega).
\end{aligned}\right.
\end{equation}

\subsection{A unique local solution}
In this subsection, we construct a unique local solution to \eqref{ode} on a time interval $[0,c]$ with $c>0$. Let us denote the Banach space of all $L^1(\Omega)$-valued continuous functions on the interval $[0,c]$ by 
\begin{equation}\label{X}
X:=C([0,c];L^1(\Omega))
\end{equation}
with the norm
\[
\left\|\lambda\right\|_X := \max_{t\in[0,c]}\left\|\lambda(t)\right\|_{L^1}.
\]
for any $\lambda\in X$. Similar to \citet[p.78]{ZeidlerFixedPoint}, let $Q_{b}$ denote the set
\begin{equation}\label{Q}
Q_{b}:=\left\{(t,\lambda)\in \mathbb{R}\times L^1(\Omega)\middle| t\in[0,a], \left\|\lambda-\lambda_0\right\|_{L^1}\leq b\right\}.
\end{equation}
Here, $b>0$ must be small enough to satisfy at least $\Lambda-b>0$. We must first verify that the operator $w$ introduced in \eqref{opw} is well-defined. In other words, we must be sure that $W\in C_{0+}(\Omega)$ exists uniquely for each $\lambda\in L^1(\Omega)$.
\begin{theo}\label{th:fp}
For any $\lambda\in L^1(\Omega)$ such that 
\begin{equation}\label{Lam1lambdaLam2}
0<\Lambda_1\leq \left\|\lambda\right\|_{L^1}\leq \Lambda_2,
\end{equation}
if the inequality
\begin{equation}\label{sufficientcondition}
\frac{\mu}{\sigma}\left(\frac{T_{\max}}{T_{\min}}\right)^{\sigma-1}\frac{\Lambda_2}{\Lambda_1} < 1
\end{equation}
holds, then \eqref{fp} has a uniuqe fixed point $W\in C_{0+}(\Omega)$.
\end{theo}
\noindent
{\bf Note}: If $\frac{\Lambda_2}{\Lambda_1}\to 1$ and $\frac{T_{\max}}{T_{\min}}\to 1$, the left side of \eqref{sufficientcondition} converges to $\frac{\mu}{\sigma}<1$. As we will see later, since $\Lambda_1=\Lambda_2=\Lambda$ in practice, the assumption \eqref{sufficientcondition} requires that a disparity in transport costs $\frac{T_{\max}}{T_{\min}}$ is sufficiently close to $1$.

\vspace{1mm}
\noindent
{\bf Note}: It may be possible to offer the following intuitive interpretation of why a moderate disparity in transport costs $\frac{T_{\max}}{T_{\min}}$ is a sufficient condition for the unique existence of instantaneous equilibrium. Suppose that the disparity in transport costs is too great. Then, each region has to operate an autarkic economy. Consequently, in that case, there is no unified equilibrium price system for the economy as a whole.

\vspace{3mm}
The following two lemmas are essential for constructing a unique local solution to \eqref{ode}. See Subsection \ref{subsec:prooflembounded} and \ref{subsec:prooflemLip} for their proofs.
\begin{lem}\label{lem:bounded}
If
\begin{equation}\label{sufficientconditionLb}
\frac{\mu}{\sigma}\left(\frac{T_{\max}}{T_{\min}}\right)^{\sigma-1}\frac{\Lambda+b}{\Lambda-b} < 1
\end{equation}
holds, there exists $K>0$ which does not depend on $\lambda_0$ and satisfy
\begin{equation}\label{bounded}
\left\|\Psi[\lambda]\right\|_{L^1} \leq K
\end{equation}
for any $(t,\lambda)\in Q_{b}$.
\end{lem}
\begin{lem}\label{lem:lip}
If \eqref{sufficientconditionLb} holds, there exists $L>0$ which does not depend on $\lambda_0$ and satisfy
\begin{equation}\label{Lip}
\left\|\Psi[\lambda_1]-\Psi[\lambda_2]\right\|_{L^1} \leq L\left\|\lambda_1-\lambda_2\right\|_{L^1}
\end{equation}
for any $(t,\lambda_1)\in Q_{b}$ and $(t,\lambda_2)\in Q_{b}$.
\end{lem}

\noindent
{\bf Note}: The condition \eqref{sufficientconditionLb} is equivalent to placing $\Lambda_1=\Lambda-b$ and $\Lambda_2=\Lambda+b$ in the condition \eqref{sufficientcondition}.

\vspace{5mm}
\noindent
By using these lemmas, we can construct a unique local solution to \eqref{ode}.
\begin{theo}\label{th:localsol}
Suppose that $c>0$ satisfies
\[
c < \min\left\{a,\frac{b}{K},\frac{1}{L}\right\},
\]
then there exists a unique local solution $\lambda\in C^1([0,c];L^1_\Lambda(\Omega))\subset X$ to \eqref{ode} for any $\lambda_0\in L^1_\Lambda(\Omega)$. The constant $c>0$ does not depend on $\lambda_0$.
\end{theo}
\begin{proof}
From the Picard–Lindel\"{o}f theorem,\footnote{See \citet[pp.78-79]{ZeidlerFixedPoint} or \citet[p.147]{oishi1997} for example.} we obtain a unique local solution 
\[
\lambda\in C^1([0,c];L^1(\Omega)).
\]
The task is now to show that $\lambda\geq 0$ a.e. $x\in \Omega$ and that $\int_\Omega\lambda(x)dx=\Lambda$. To show the former, we only have to see that the solution $\lambda\in C^1([0,c];L^1(\Omega))$ to \eqref{ode} can be expressed by
\begin{equation}\label{lambdaexp}
\lambda(t) = \lambda_0\exp{\int_0^t \left(\omega[\lambda(s)]-\tilde{\omega}[\lambda]\right)ds}
\end{equation}
a.e. $x\in \Omega$.\footnote{Differentiating both sides of \eqref{lambdaexp} by $t$ yields \eqref{ode}.} Hence, $\lambda(t)\geq 0$ a.e. $x\in\Omega$ for all $t\in[0,c]$. To show the latter, we integrate both sides of the first equation of \eqref{ode} with the variable $t$ and using \eqref{opPsi} and the second equation of \eqref{ode}. Then, we obtain
\begin{equation}\label{solintegralform}
\lambda(t)=\lambda_0+v\int_0^t\left(\omega[\lambda(s)]-\tilde{\omega}[\lambda(s)]\right)\lambda(s)ds.
\end{equation}
Since $\lambda(t)\geq 0$ a.e. $x\in\Omega$ for all $t\in[0,c]$, it follows from \eqref{solintegralform} that
\[
\begin{aligned}
\left\|\lambda(t)\right\|_{L^1}
&= \int_\Omega\lambda(t,x)dx \\
&= \int_\Omega\lambda_0(x)dx
+ \int_\Omega \int_0^t\Psi[\lambda(s)]dsdx \\
&=\Lambda + v\int_0^t\int_\Omega\left(\omega[\lambda(s)]-\tilde{\omega}[\lambda(s)]\right)\lambda(s)dxds \\
&= \Lambda
\end{aligned}
\]
for all $t\in[0, c]$.
\end{proof}

\noindent
{\bf Note}: The local solution is in $C^1$ class. This is immediate from \eqref{ode} and the fact that $\Psi[\lambda(t)]$ is continuous with respect to $t$.

\subsection{A unique global solution}

A unique global solution can be constructed by applying Theorem \ref{th:localsol} repeatedly. That is, since a time interval $[0, c]$ over which a unique local solution exists in Theorem \ref{th:localsol} does not depend on an initial value $\lambda_0$, it is easy to extend the local solution by re-taking an initial value as $\lambda_0=\lambda(c)$ and re-applying Theorem \ref{th:localsol}. Thus, we now have a local solution on $[0, 2c]$. All that remains is to repeat this process. Finally, we have the following theorem.
\begin{theo}
For any $\lambda_0\in L^1_\Lambda(\Omega)$, there exists a unique global solution $\lambda\in C^1([0,\infty);L^1_\Lambda(\Omega))$ to \eqref{ode}.
\end{theo}

\section{Racetrack economy}\label{sec:racetrack}
In Sections \ref{sec:racetrack}, \ref{sec:instability}, \ref{sec:numerical}, and \ref{sec:stabilityspikysol}, we consider the FE model \eqref 
{1} on a circle $S$ of which radius is $\rho>0$. This is a setting often referred to as a racetrack economy. A function $h$ on $S$ is identified with a periodic function $\tilde{h}$ on a closed interval $[-\pi,\pi]$. We can put $x=\rho\theta$ for $x\in S$ and $\theta\in[-\pi,\pi]$, thus we have
\begin{equation}\label{hhtilde}
h(x)=h(\rho\theta)=:\tilde{h}(\theta).
\end{equation}
Then, the integral on $S$ is computed as 
\begin{equation}\label{intonmpipi}
\int_{S}h(x)dx = \int_{-\pi}^\pi \tilde{h}(\theta)\rho d\theta.
\end{equation}
To avoid complicated symbols, $\tilde{h}$ is written as $h$ again from here on when no confusion can arise.

The distance between two points $x=\rho\theta$ and $y=\rho\theta^\prime$ on $S$ is defined as
\begin{equation}\label{distonmpipi}
d(x,y):=\rho\min\left\{|\theta-\theta^\prime|,2\pi-|\theta-\theta^\prime|\right\}
\end{equation}
It is obvious that $d(x,y)=d(y,x)$. In the following sections \ref{sec:instability} and \ref{sec:numerical}, we use
\begin{equation}\label{}
T(x,y) = e^{\tau d(x,y)}.
\end{equation}
as $T(x,y)$ in \eqref{1}. Then, we see that
\begin{equation}\label{Tmin}
T_{\min} := \min_{x,y\in S} T(x,y) = 1
\end{equation}
and
\begin{equation}\label{Tmax}
T_{\max} := \max_{x,y\in S} T(x,y) = e^{\tau\rho\pi}.
\end{equation}
It is convenient to define 
\begin{equation}\label{defalpha}
\alpha:=\tau(\sigma-1)
\end{equation}
since $\tau(\sigma-1)$ often appears in the integral equations of \eqref{1}.

\section{Stability of a homogeneous stationary solution}\label{sec:instability}
In this Section, the existence of a homogeneous stationary solution is shown, and its stability is investigated by using Fourier analysis.

\subsection{A homogeneous stationary solution}
Suppose that the immobile population distribution is spatially uniform 
\begin{equation}\label{stphi}
\phi=\ol{\phi}\equiv \frac{\Phi}{2\pi\rho},
\end{equation}
which satisfies \eqref{totalimpop}. In this case, if the mobile population distribution is spatially uniform, it can be shown that there exists a homogeneous stationary solution in which all unknown functions are spatially uniform. In fact, substituting \eqref{stphi} and 
\begin{equation}\label{stlambda}
\lambda=\ol{\lambda}\equiv \frac{\Lambda}{2\pi\rho},
\end{equation}
which satisfies \eqref{totalmobpop}, into \eqref{1}, and applying
\begin{equation}
\int_C e^{-\alpha d(x,y)}dy
= \frac{2(1-e^{-\alpha\rho\pi})}{\alpha},
\end{equation}
to the second and the third equations in \eqref{1}, we have the homogeneous states
\begin{align}
&Y=\ol{Y} \equiv \ol{w}\ol{\lambda}+\ol{\phi}\geq 0,\label{stY}\\
&w=\ol{w} \equiv \frac{\frac{\mu\ol{\phi}}{\sigma\ol{\lambda}}}{1-\frac{\mu}{\sigma}}\geq 0,\label{stw}\\
&G=\ol{G} \equiv \left[\frac{2\ol{\lambda}(1-e^{-\alpha\rho\pi})}{F\alpha}\right]^{\frac{1}{1-\sigma}}\geq 0,\label{stG}\\
&\omega=\ol{\omega} \equiv \ol{w}\ol{G}^{-\mu}\geq 0.\label{stomega}
\end{align}

\subsection{Stability analysis}\label{subsec:contstabilityanalysis}
Let $\varDelta\lambda$, $\varDelta Y$, $\varDelta G$, $\varDelta w$, and $\varDelta \omega$ be small perturbations added to the homogeneous states \eqref{stlambda}, \eqref{stY}, \eqref{stw}, \eqref{stG}, and \eqref{stomega}, respectively. From \eqref{totalmobpop},
\begin{equation}\label{totaldel0}
\int_C \varDelta\lambda(t,x)dx=0,~\forall t\in[0,\infty).
\end{equation}
Substituting $\lambda=\ol{\lambda}+\varDelta\lambda$, $Y=\ol{Y}+\varDelta Y$, $G=\ol{G}+\varDelta G$, $w=\ol{w}+\varDelta w$, and $\omega=\ol{\omega}+\varDelta \omega$ into \eqref{1} and ignoring higher order terms of the perturbations, we obtain the following linearized system
\begin{equation}\label{l}
\left\{
\begin{aligned}
&\varDelta Y(t,\theta) = \ol{\lambda}\varDelta w(t,\theta) + \ol{w}\varDelta \lambda(t,\theta),\\
&\varDelta G(t,\theta) = \frac{\ol{G}^\sigma}{F(1-\sigma)}\int_{-\pi}^\pi\varDelta \lambda(t,\theta^\prime)e^{-\alpha d(\rho\theta,\rho\theta^\prime)}\rho d\theta^\prime,\\ 
&\varDelta w(t,\theta) = \frac{\mu}{\sigma F}(\sigma-1)\ol{Y}\ol{G}^{\sigma-2}\int_{-\pi}^\pi\varDelta G(t,\theta^\prime)e^{-\alpha d(\rho\theta,\rho\theta^\prime)}\rho d\theta^\prime\\
&\hspace{30mm}+\frac{\mu}{\sigma F}\ol{G}^{\sigma-1}\int_{-\pi}^\pi\varDelta Y(t,\theta^\prime)e^{-\alpha d(\rho\theta,\rho\theta^\prime)}\rho d\theta^\prime,\\
&\varDelta \omega(t,\theta) = -\mu\ol{w}\ol{G}^{-\mu-1}\varDelta G(t,\theta) + \ol{G}^{-\mu}\varDelta w(t,\theta),\\
&\frac{\partial}{\partial t}\varDelta \lambda(t,\theta) = v\ol{\lambda}\varDelta\omega(t,\theta).
\end{aligned}\right.
\end{equation}
We expand all the perturbations in \eqref{l} into a Fourier series. The expansion of a periodic function $\varDelta h$ on $[-\pi,\pi]$ is defined as
\begin{equation}\label{Fseries}
\varDelta h(\theta)=\frac{1}{2\pi}\sum_{k=0,\pm1,\pm2,\cdots}\hat{h}_ke^{ik\theta}
\end{equation}
where $\hat{h}_k$ is a Fourier coefficient defined by
\begin{equation}\label{Fc}
\hat{h}_k := \int_{-\pi}^{\pi}\varDelta h(\psi)e^{ik\psi}d\psi 
\end{equation}
As a result of the Fourier expansion, we have to calculate the convolution\footnote{The same convolution appears in the analysis of NEG models of the racetrack economy. See \citet{OhtakeYagi_Asym}, \citet{Ohtake2023city}, \citet{Ohtake2023cont}, and \citet{Ohtake2025agriculture} for example.}
\begin{align}
\int_{-\pi}^{\pi}e^{ik\theta^\prime}e^{-\alpha d(\rho\theta,\rho\theta^\prime)}\rho d\theta^\prime.
\end{align}
By a careful calculation, we obtain
\begin{equation}
\int_{-\pi}^{\pi}e^{ik\theta^\prime}e^{-\alpha d(\rho\theta,\rho\theta^\prime)}\rho d\theta^\prime = E_ke^{ik\theta},
\end{equation}
where
\begin{equation}\label{Ek}
E_k = \frac{2\alpha\rho^2\left(1-(-1)^{|k|}e^{-\alpha\rho\pi}\right)}{k^2+\alpha^2\rho^2}.
\end{equation}
 Then, by defining $Z_k$ as
\begin{equation}\label{Zk}
Z_k:= \frac{\alpha^2\rho^2\left(1-(-1)^{|k|}e^{-\alpha\rho\pi}\right)}{(k^2+\alpha^2\rho^2)(1-e^{-\alpha\rho\pi})},
\end{equation}
we see that
\begin{equation}\label{stGsigm1EkFlamZk}
\ol{G}^{\sigma-1}E_k = \frac{F}{\ol{\lambda}}Z_k.
\end{equation}
It is known that \footnote{The variable $Z_k$ corresponds to $Z$ in \citet[p.73, p.90]{FujiKrugVenab}. This variable commonly appears in the stability analysis of NEG models of the racetrack economy. See \citet{Ohtake2023city} or \citet{Ohtake2025agriculture} for the proof of \eqref{ddalpaZk>0}, \eqref{limZk0}, and \eqref{limZk1}.}
\begin{align}
&\frac{d}{d\alpha}Z_k \geq 0,\label{ddalpaZk>0}\\
&\lim_{\alpha\rho\to 0}Z_k = 0,\label{limZk0}\\
&\lim_{\alpha\rho\to\infty}Z_k = 1.\label{limZk1}
\end{align} 
Then, from the linearized system \eqref{l}, we have the following equations for the Fourier coefficients
\begin{align}
&\hat{Y}_k = \ol{\lambda}\hat{w}_k+\ol{w}\hat{\lambda}_k,\label{eqYhatk}\\
&\hat{G}_k = \frac{\ol{G}}{(1-\sigma)\ol{\lambda}}Z_k\hat{\lambda}_k,\label{eqGhatk}\\
&\left(1-\frac{\mu}{\sigma}Z_k\right)\hat{w}_k = \left(-\frac{\mu\ol{Y}}{\sigma\ol{\lambda}^2}Z_k^2+\frac{\mu\ol{w}}{\sigma\ol{\lambda}}Z_k\right)\hat{\lambda}_k,\label{eqwhatk}\\
&\hat{\omega}_k = -\mu\ol{w}\ol{G}^{-\mu-1}\hat{G}_k + \ol{G}^{-\mu}\hat{w}_k,\label{eqomegahatk}\\
&\frac{d}{dt}\hat{\lambda}_k = v\ol{\lambda}\hat{\omega}_k.\label{eqddtlambdahatk}
\end{align}
Because of \eqref{totaldel0}, $\hat{\lambda}_0=0$. Therefore, we only have to consider the case $k\neq 0$. Solving \eqref{eqwhatk} for $\hat{w}_k$ yields 
\begin{equation}\label{whatk}
\hat{w}_k = \frac{-\frac{\mu\ol{Y}}{\sigma\ol{\lambda}^2}Z_k^2+\frac{\mu\ol{w}}{\sigma\ol{\lambda}}Z_k}{1-\frac{\mu}{\sigma}Z_k}\hat{\lambda}_k.
\end{equation}
Substituting \eqref{eqGhatk} and \eqref{whatk} into \eqref{eqomegahatk} yields
\begin{equation}\label{omegahatk}
\hat{\omega}_k = -\frac{\mu\ol{G}^{-\mu}}{\ol{\lambda}}\left(\frac{\ol{w}}{1-\sigma}Z_k + \frac{\frac{\ol{Y}}{\ol{\lambda}}Z_k^2-\ol{w}Z_k}{\sigma-\mu Z_k}\right)\hat{\lambda}_k.
\end{equation}

From \eqref{eqddtlambdahatk} and \eqref{omegahatk}, we have
\begin{equation}\label{dynlamk}
\frac{d}{dt}\hat{\lambda}_k = \Gamma_k\hat{\lambda}_k
\end{equation}
where
\begin{equation}\label{Gammak}
\Gamma_k := -v\mu\ol{G}^{-\mu}\left(\frac{\ol{w}}{1-\sigma} + \frac{\frac{\ol{Y}}{\ol{\lambda}}Z_k-\ol{w}}{\sigma-\mu Z_k}\right)Z_k
\end{equation}
which is called an {\it eigenvalue} of the system \eqref{l}. Note that
\begin{equation}
\Gamma_k = \Gamma_{-k}
\end{equation}
for any $k\neq 0$. It is easy to solve \eqref{dynlamk} as
\begin{equation}
\hat{\lambda}_k(t) = \hat{\lambda}(0)e^{t\Gamma_k},~t\geq 0
\end{equation}
Therefore, if $\Gamma_k<0$ (resp.~$\Gamma_k>0$), the k-th mode $e^{ik\theta}$ decays (resp.~grows) over time. In this sense, the homogeneous stationary solution is stable (resp.~unstable) when $\Gamma_k<0$ (resp.~$\Gamma_k>0$). For convenience, we also use the expression that the $k$-th mode is stable (resp. unstable) when $\Gamma_k<0$ (resp.~$\Gamma_k>0$).

Note that $\sigma-\mu Z_k > 0$ because $\sigma>1$, $\mu\in(0,1)$, and $Z_k\in(0,1)$. Therefore, from \eqref{Gammak}, we can show that there exists
\begin{equation}\label{Z*}
Z^* = \frac{\ol{w}+\frac{\ol{w}\sigma}{\sigma-1}}{\frac{\mu\ol{w}}{\sigma-1}+\frac{\ol{Y}}{\ol{\lambda}}} > 0
\end{equation} 
such that the $k$-th mode is stable (resp.~unstable) when $Z_k>Z^*$ (resp.~$Z_k<Z^*$). See Fig.~\ref{fig:signSk} for the sketch of $(\sigma-\mu Z_k)\Gamma_k$ whose sign is the same as that of $\Gamma_k$. 
\begin{figure}[H]
\centering
\begin{tikzpicture}
 \draw[name path=xaxis,->,thin] (0,0)--(5.0,0)node[right]{$Z_k$};
 \draw[->,thin] (0,-1.5)--(0,1.5)node[right]{$$};
 \draw (0,0)node[below left]{O};
 \draw[name path=Omega1,blue,thin,domain=0:4.5] plot(\x,{-0.30*pow(\x,2)+\x})node[above right]{$(\sigma-\mu Z_k)\Gamma_k$};
 \path[name intersections={of=Omega1 and xaxis}];
 \fill[black] (intersection-2) circle (0.07) node[above right]{$Z^*$};
\end{tikzpicture}
\caption{Sketch of $(\sigma-\mu Z_k)\Gamma_k$}
\label{fig:signSk}
\end{figure}
\noindent
{\bf Note}: Since $Z_k<1$, in order for the $k$-th mode to actually stabilize, $Z^*<1$ must be satisfied. This is a so-called assumption of no black holes,\footnote{\citet[p.~59]{FujiKrugVenab}} which provides a necessary and sufficient condition for any mode to be stable if $Z$ is sufficiently close to $1$. Putting $Z^*<1$ in \eqref{Z*} and using \eqref{stY} and \eqref{stw}, we obtain the assumption of no black holes
\begin{equation}\label{nbh}
\sigma > 1+\mu,
\end{equation}
which requires a relatively low preference for variety. This assumption is also advantageous for \eqref{sufficientcondition} in Theorem \ref{th:fp} to hold.

\vspace{5mm}
Suppose that the assumption of no black holes \eqref{nbh} holds. As in \eqref{ddalpaZk>0}, since $Z_k$ is monotonically non-decreasing with respect to $\tau$, any $k$-th mode is stable (resp. unstable) if $\tau$ is larger (resp. less) than a {\it critical point} $\tau_k^*$ on which $\Gamma_k=0$. Fig.~\ref{fig:eigen} shows the eigenvalue $\Gamma_k$ for several values of $k$.\footnote{The parameters are fixed to $\Phi=1.0$, $\Lambda=1.0$, $\mu=0.6$, $F=1.0$, $\sigma=3.0$, and $\rho=1.0$. The following figures \ref{fig:eigen}, \ref{fig:hmaps}, and \ref{fig:ccurves} are drawn using Numpy \citep{numpy} and Matplotlib \citep{matplotlib}. The source codes are available on \url{https://github.com/k-ohtake/contspace-femodel}.} From this, we can see that the smaller the value of $k$, the smaller the value of $\tau_k^*$.\footnote{Analytically, we can show that $\tau_{|k|}^*<\tau_{|k|+2}^*$ and $\tau_{-|k|}^*<\tau_{-|k|-2}^*$. See Theorem \ref{th:tau*k} in Subsection \ref{subsec:deptau*k} for details.} This means that the absolute value of frequencies of unstable modes decreases as transport costs decrease.\footnote{Note that $\Gamma_k = \Gamma_{-k}$ for $k\neq 0$.}This is a common property shared by many basic NEG models of the racetrack economy, both in discrete and continuous spaces.\footnote{See \citet{ChiAsh08}, \citet{AkaTakaIke}, \citet{OhtakeYagi_point}, \citet{OhtakeYagi_Asym}, and \citet{Ohtake2023cont} for example.}

\begin{figure}[H]
\centering
\includegraphics[width=0.6\columnwidth]{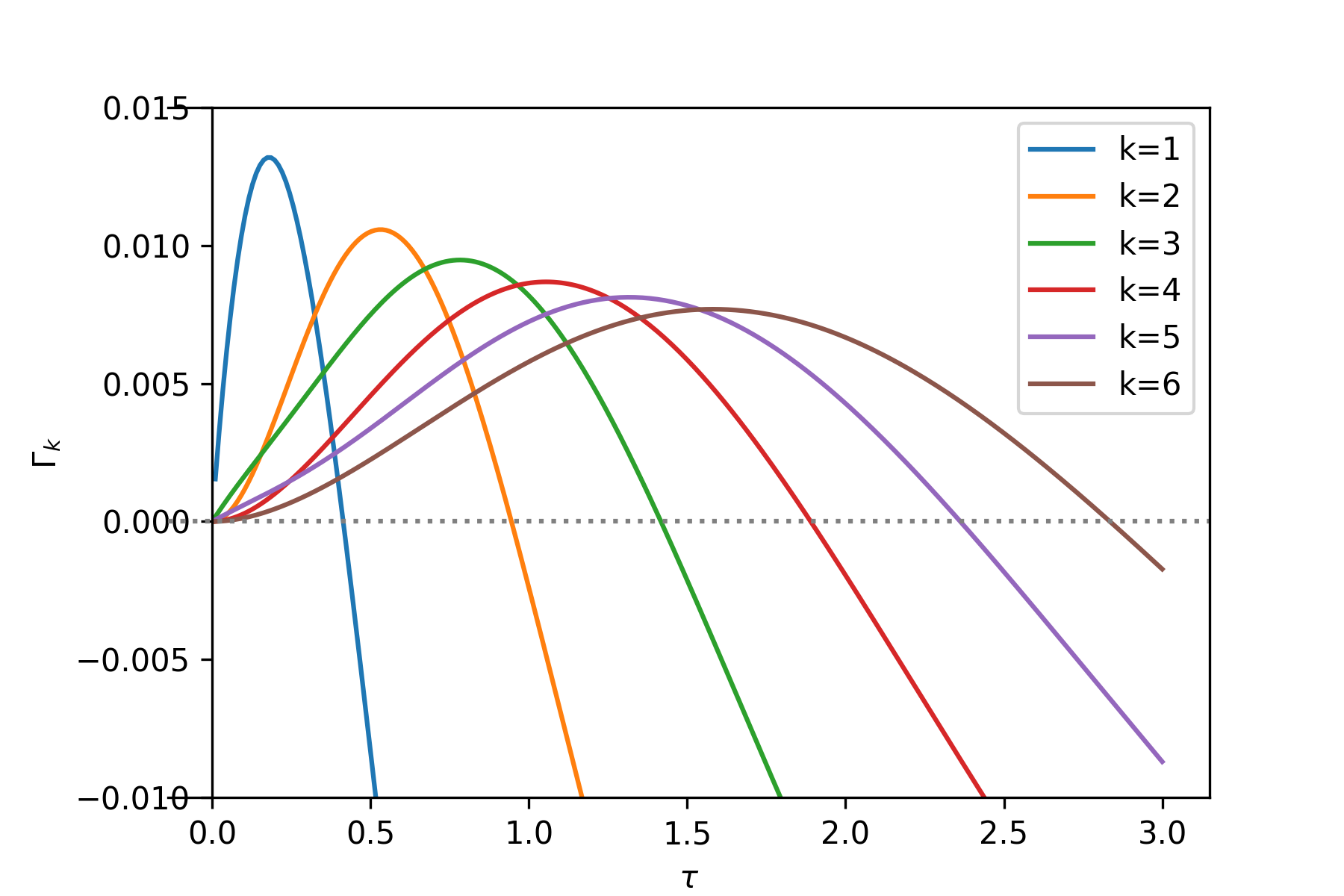}
\caption{Eigenvalues for various frequencies}
\label{fig:eigen}
\end{figure}

For a more detailed view, Fig.~\ref{fig:hmaps} shows the eigenvalues plotted as  heatmaps on the $\tau$-$\sigma$ plane.\footnote{The parameters are fixed to $\Phi=1.0$, $\Lambda=1.0$, $\mu=0.6$, $F=1.0$, and $\rho=1.0$.} The thick black curves represent $\Gamma_k=0$. Let us refer to the curve as the {\it critical curve}. Within the area bounded by the critical curve (lower left area in each heatmap), the eigenvalues are positive, whereas outside (upper right area in each heatmap), the eigenvalues are negative. From these figures, we can see the $\sigma$-dependence of critical points. That is, the smaller $\sigma$ is, the larger the critical points become. In other words, the stronger the preference for variety, the more agglomeration can occur even with larger transport costs. Furthermore, Fig.~\ref{fig:ccurves} shows only the critical curves for each frequency. We can see that the critical curve shifts outwards as the absolute value of frequencies increases. That is, for any given $\sigma$, modes with large frequencies destabilize earlier as transport costs decrease. \citet{Ohtake2025pattern} confirms that the shape and frequency dependence of the critical curve is similar to other representative NEG models.\footnote{Specifically, the Krugman-type model originated from \citet{Krug91}, in which the utility function is given in the Cobb-Douglas form, and the Pfl{\"u}ger-type model originated from \citet{Pfl}, in which the utility function is given in a quasilinear-log form.}

\begin{figure}[H]
\centering
\includegraphics[width=\columnwidth]{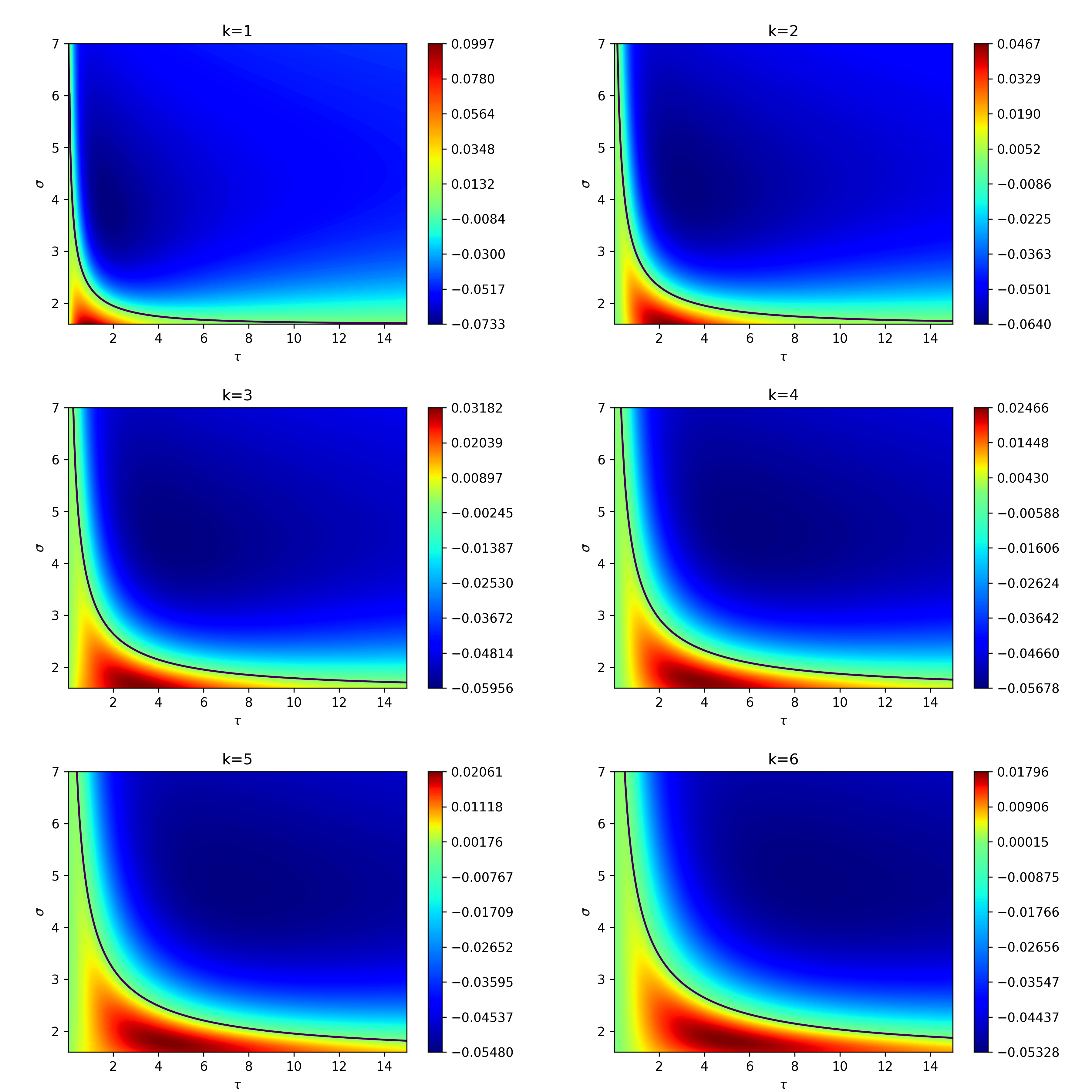}
\caption{Heatmaps of eigenvalues}
\label{fig:hmaps}
\end{figure}

\begin{figure}[H]
\centering
\includegraphics[width=0.6\columnwidth]{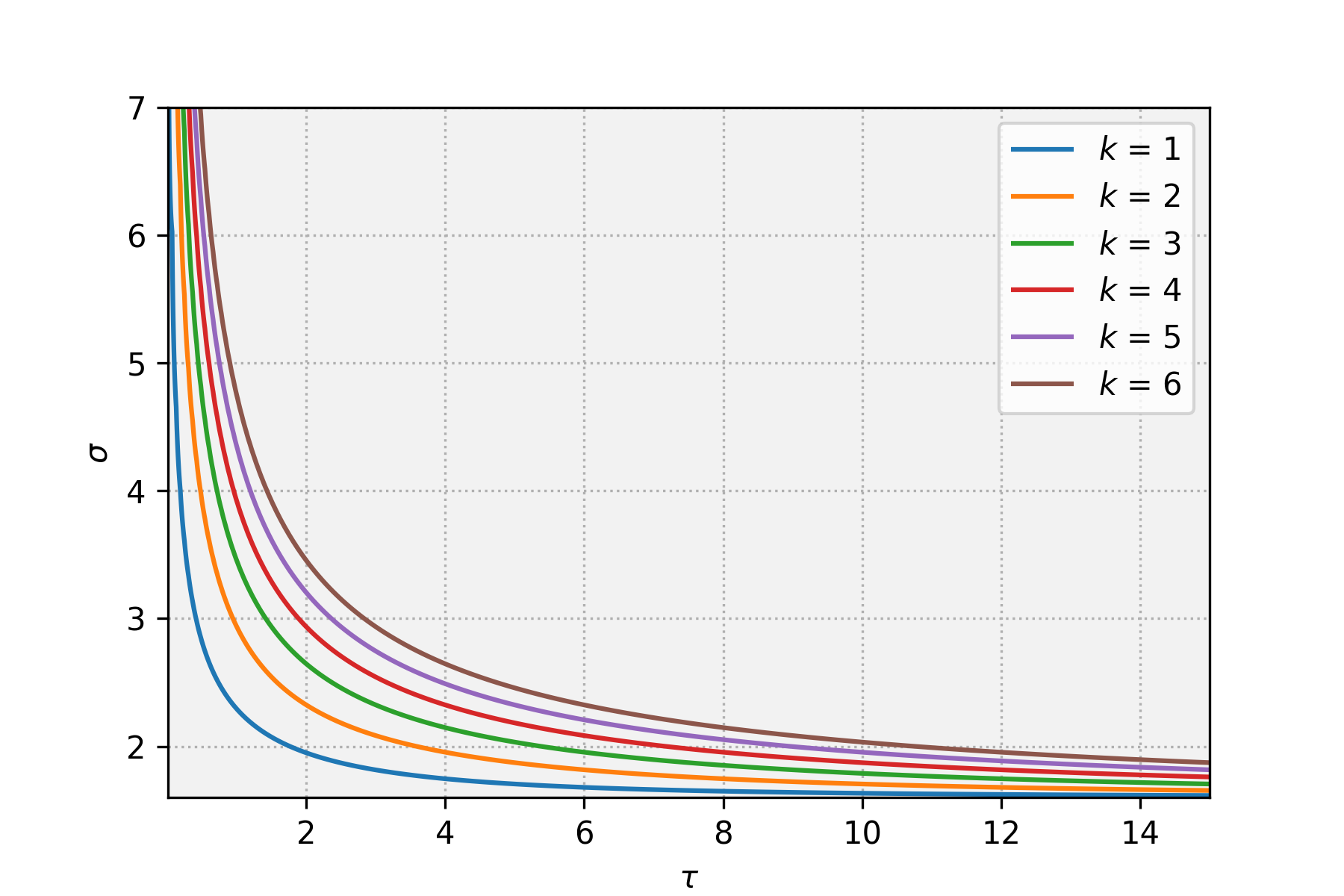}
\caption{Critical curves for various frequencies}
\label{fig:ccurves}
\end{figure}

\section{Numerical simulations for time evolution}\label{sec:numerical}
It is difficult to analytically calculate final states of solutions that start near the unstable homogeneous stationary solution. In this section, we observe by numerical computation that the solutions asymptotically approach non-homogeneous stationary solutions.

\subsection{Numerical scheme}
To solve equation (1) numerically, continuous time and space must be discretized. With $\Delta t=0.01$, the time variable $t$ is discretized into the series of $t_k=k\Delta t$ for $k=0,1,2,\cdots$. With $\Delta \theta=2\pi/N$, the space variable $\theta\in[-\pi,\pi]$ is discretized into $N$ nodes $\theta_i=-\pi+(i-1)\Delta\theta$ for $i=1,2,\cdots,N$.\footnote{In the simulation of this paper, $N=256$.} Then, the functions in the model \eqref{1} are approximately expressed as $N$-dimensional vectors. For example, $\lambda(t_k,\cdot)$ is approximated as $\lambda^k:=[\lambda_1^k,\lambda_2^k,\cdots,\lambda_{N}^k]\in\mathbb{R}^N$. Then, the integral of a function $h$ on $\mathbb{R}\times [-\pi,\pi]$ is simply approximated by the Riemann sum
\begin{equation}
\int_{-\pi}^\pi h(t_k,\theta)\rho d\theta \simeq \sum_{i=1}^{N} h^k_i\rho\Delta\theta,~k=0,1,2,\cdots.
\end{equation}
For a given $\lambda^k\in\mathbb{R}^N$ at each time step $t_k$, by substituting \eqref{priceindexeq} and \eqref{incomex} into \eqref{nominalwageeq}, we obtain the instantaneous equilibrium equation
\begin{equation}\label{insteq}
w(x) = \frac{\mu}{\sigma}\int_\Omega\frac{w(y)\lambda(y)+\phi(y)}{\int_\Omega\lambda(t,z)T(y,z)^{1-\sigma}dz}T(x,y)^{1-\sigma}dy,\hspace{1mm}x\in\Omega.
\end{equation}
We approximate \eqref{insteq} as
\begin{equation}
w_i = \frac{\mu}{\sigma}\sum_{j=1}^N\frac{w_j\lambda_j+\phi_j}{\sum_{k=1}^N\lambda_kT(y_j,z_k)^{1-\sigma}\Delta z}T(x_i,y_j)^{1-\sigma}\Delta y,\hspace{1mm}i=1,2,\cdots,N.
\end{equation}
and have to solve these nonlinear algebraic equations. In most cases, a simple iterative method works well.\footnote{In fact, the iterations converge even for parameter values that far exceed the sufficient condition \eqref{sufficientcondition} for the operator $E$ defined in \eqref{mapE} to be a contraction. This suggests that the sufficient condition \eqref{sufficientcondition} is too strict.} The differential equation for $\lambda$ of \eqref{1} is solved by the explicit Euler method. Thus, we start with an initial value $\lambda^0\in\mathbb{R}^N$ and obtain a sequence $\lambda^1,\lambda^2,\cdots$ as an approximation of the time evolution of the unknown function $\lambda$. When $|\lambda_i^{n+1}-\lambda_i^{n}|<10^{-10}~(\forall i=1,2,\cdots,N)$, a numerical solution $\lambda^*=\lambda^{n+1}$ is considered to be an approximated stationary solution and the computation for time evolution is stopped. The simulation code in Julia\footnote{See \citet{be2017julia} for details on the Julia language.} is available at \url{https://github.com/k-ohtake/contspace-femodel}.

\subsection{Numerical results}

The following figures show numerical stationary solutions $\lambda^*$ obtained in this way and the corresponding distributions $\omega^*$ of real wages on the interval $[-\pi,\pi]$.\footnote{These figures are drawn using Matplotlib \citep{matplotlib}.} In this simulation, fixed parameters were set to $\Phi=1$, $\Lambda=1$, $\mu=0.6$, $F=1$, $v=1$, and $\rho=1$. Control parameters are $\tau>0$ and $\sigma>1$. 

For each set of parameters, five simulations were performed. Since initial values were randomly generated for each numerical experiment, they generally differ from one another. As we will see below, the resulting mobile population distribution shapes several spikes. Although the number and position of the spikes vary depending on initial values, the number of spikes exhibits a considerable degree of robustness, barring minor fluctuations. The figure below shows the results where the number of spikes was largest among the five experiments for each parameter set.

Figs.~\ref{fig:t0p2}-\ref{fig:t1p6} show the results when $\sigma$ is fixed to $3.0$ and $\tau$ is varied. For each value of $\tau$, the mobile population density $\lambda$ at the stationary solution forms several spikes. Additionally, in the spike regions, higher real wages are achieved compared to other regions. As we increase $\tau=0.2$, $0.5$, $0.9$, $1.1$, $1.3$, and $1.6$, we see that the largest number of spikes increases from $1$, $2$, $3$, $4$, $5$, and $6$. In other words, lower transport costs promote agglomeration in a smaller number of regions. This qualitative property is similar to that of other models of the racetrack economy.\footnote{In the case of discrete space, for example, \citet{AkaTakaIke} and \citet{IkeAkaKon} must be cited. In the case of continuous space, see, for example, \citet{OhtakeYagi_point}, \citet{Ohtake2023city}, and \citet{Ohtake2023cont}.}

\begin{figure}[H]
 \begin{subfigure}{0.5\columnwidth}
  \centering
  \includegraphics[width=\columnwidth]{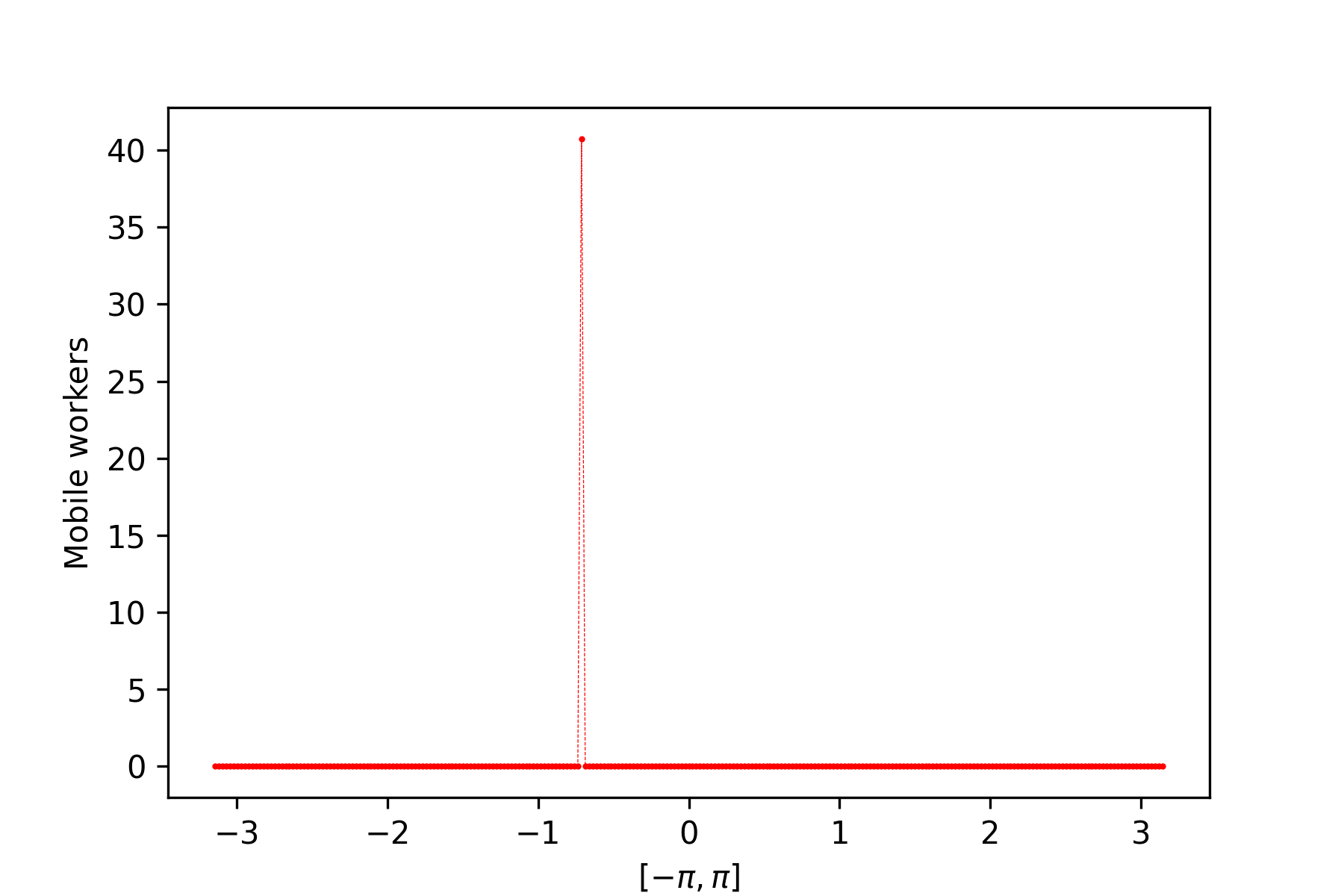}
  \caption{Mobile population}
 \end{subfigure}
 \begin{subfigure}{0.5\columnwidth}
  \centering
  \includegraphics[width=\columnwidth]{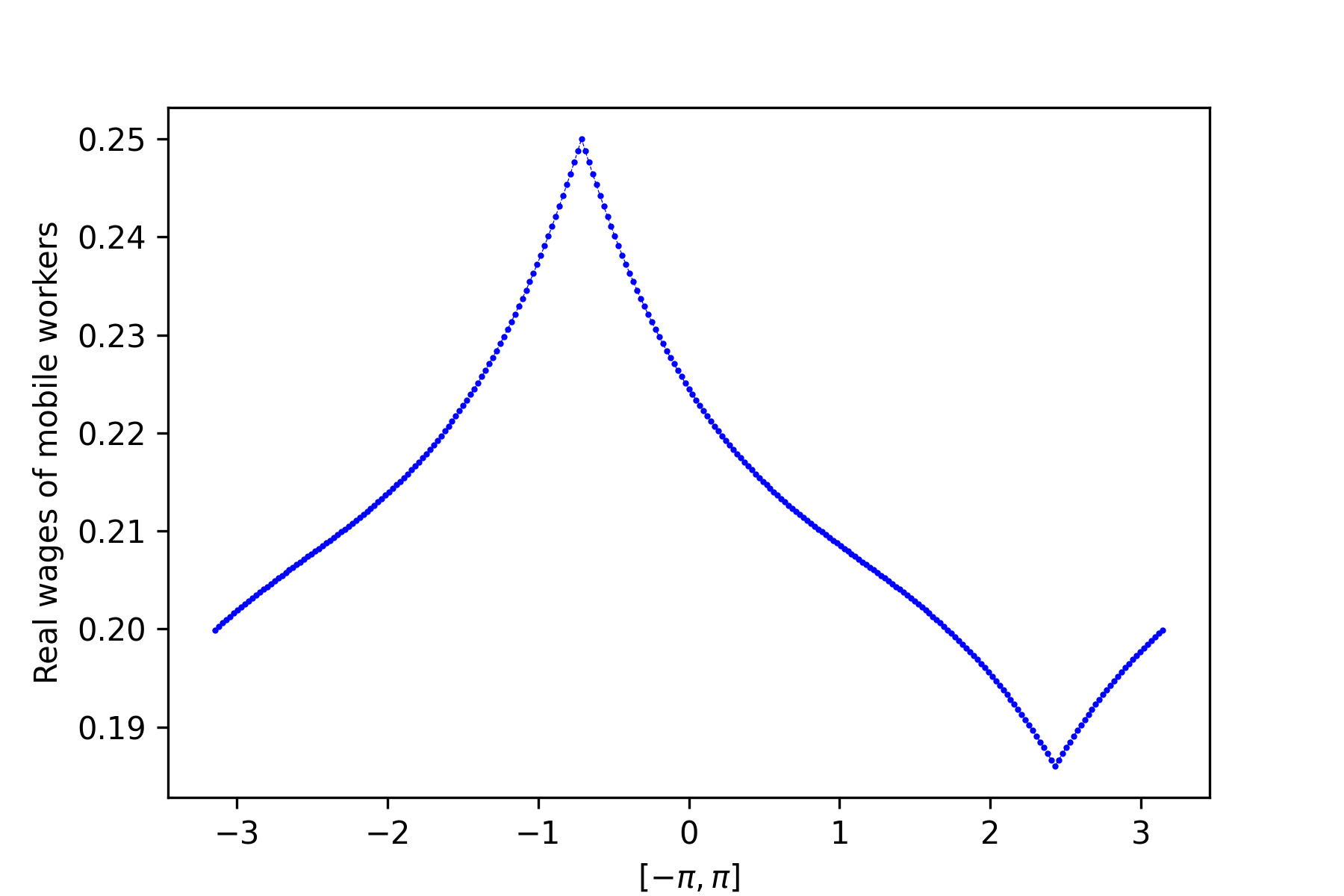}
  \caption{Real wage}
 \end{subfigure}\\
 \caption{Stationary solution for $(\sigma,\tau)=(3.0, 0.2)$}
 \label{fig:t0p2}
\end{figure}

\begin{figure}[H]
 \begin{subfigure}{0.5\columnwidth}
  \centering
  \includegraphics[width=\columnwidth]{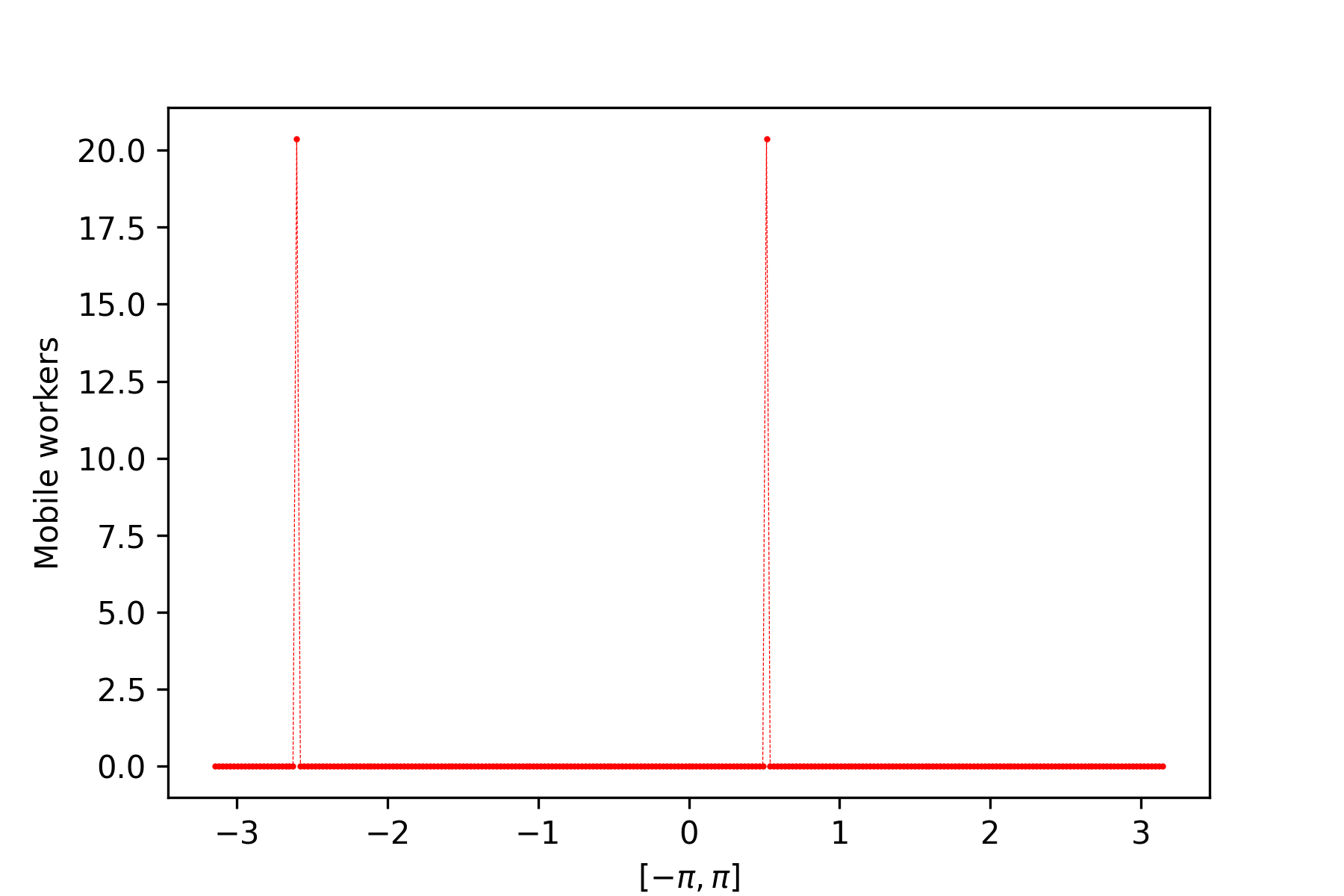}
  \caption{Mobile population}
 \end{subfigure}
 \begin{subfigure}{0.5\columnwidth}
  \centering
  \includegraphics[width=\columnwidth]{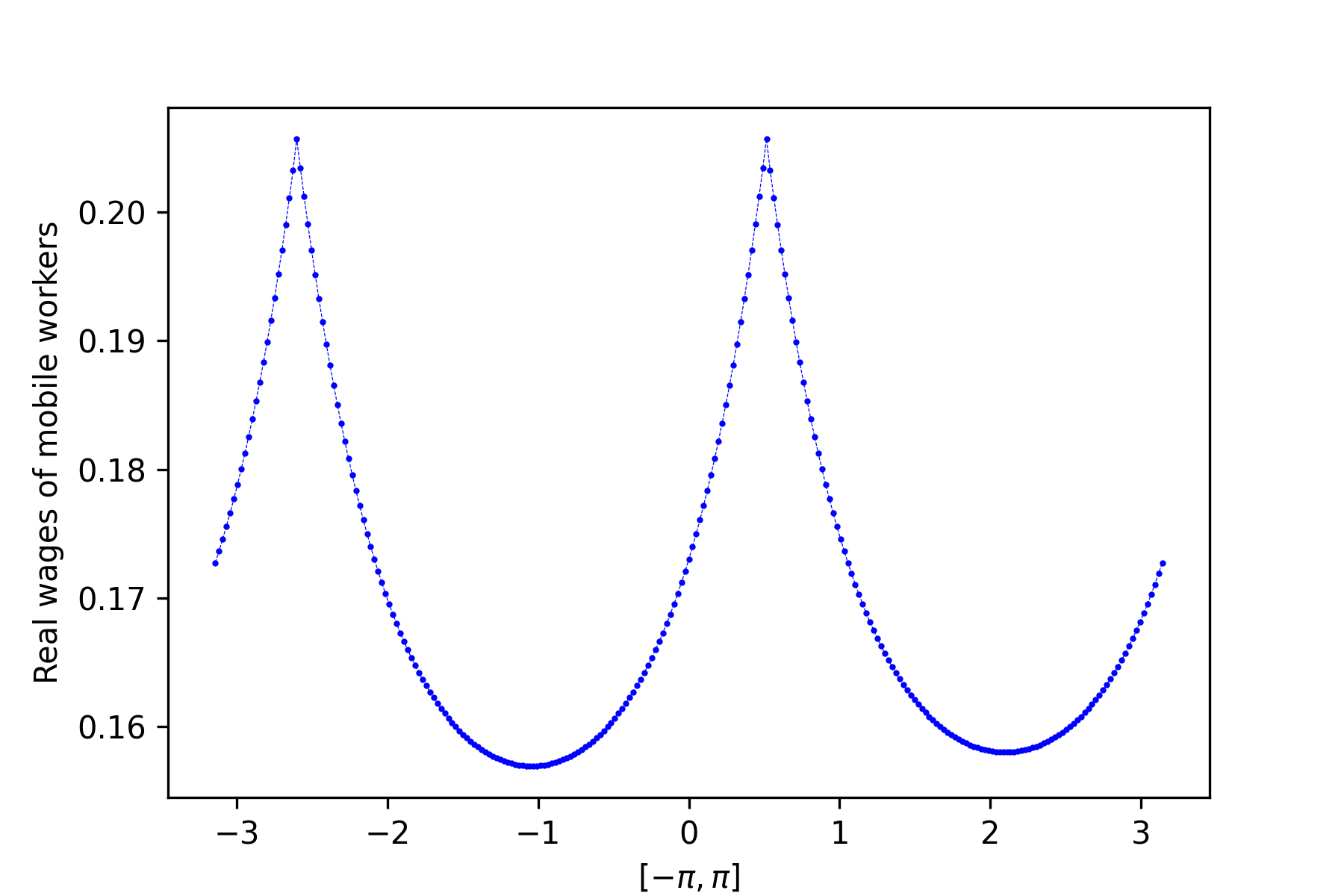}
  \caption{Real wage}
 \end{subfigure}\\
 \caption{Stationary solution for $(\sigma,\tau)=(3.0, 0.5)$}
 \label{fig:t0p5}
\end{figure}

\begin{figure}[H]
 \begin{subfigure}{0.5\columnwidth}
  \centering
  \includegraphics[width=\columnwidth]{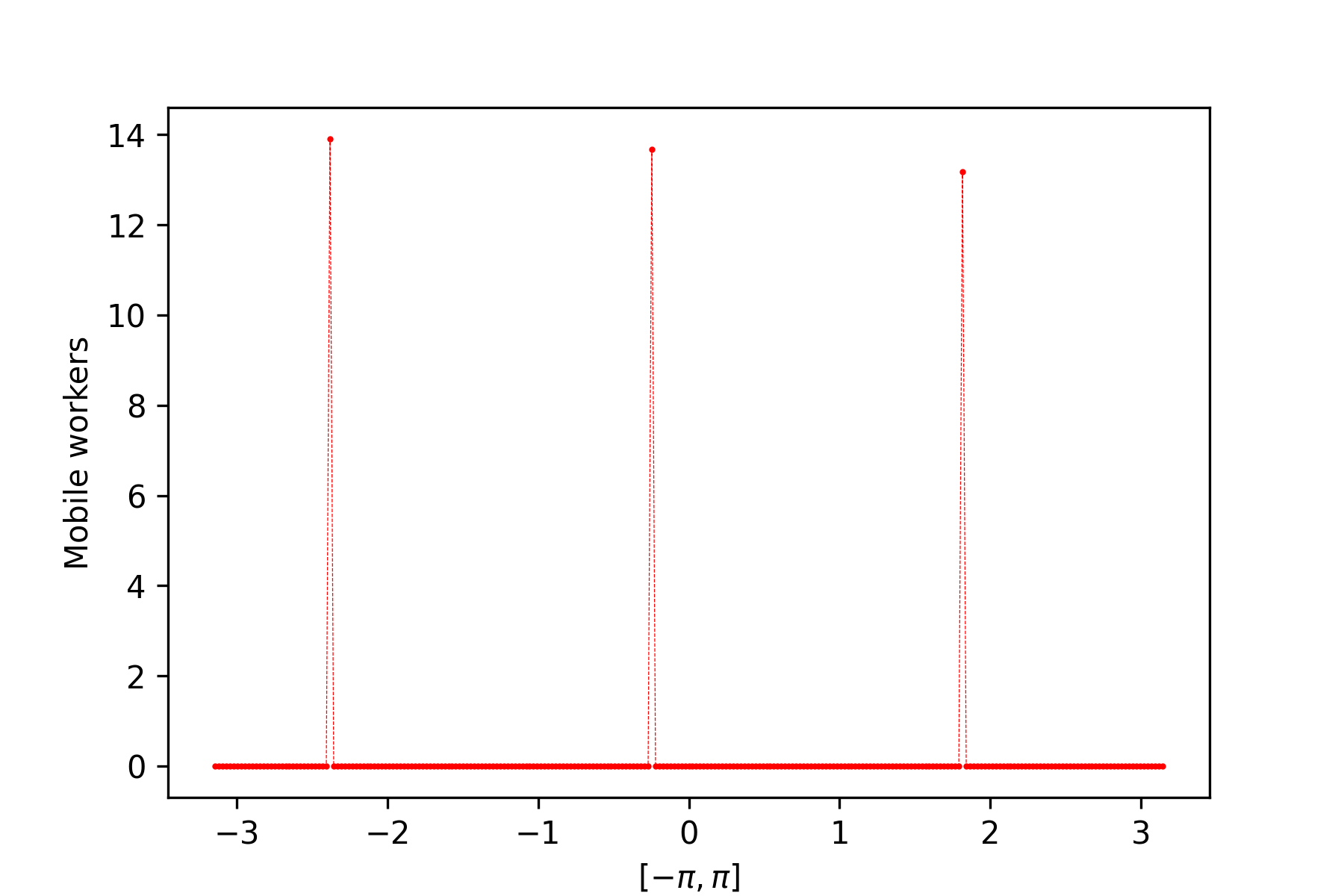}
  \caption{Mobile population}
 \end{subfigure}
 \begin{subfigure}{0.5\columnwidth}
  \centering
  \includegraphics[width=\columnwidth]{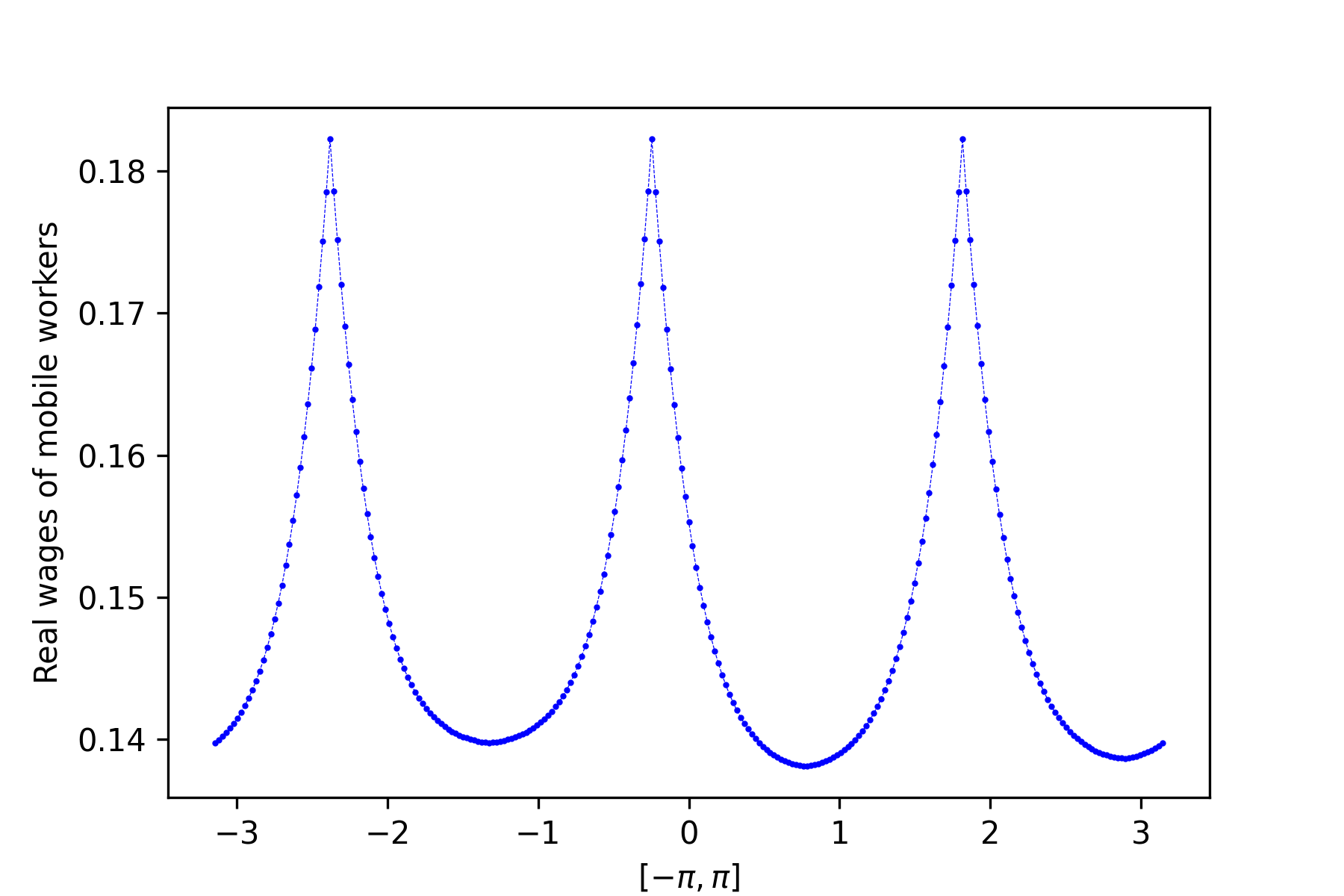}
  \caption{Real wage}
 \end{subfigure}\\
 \caption{Stationary solution for $(\sigma,\tau)=(3.0, 0.9)$}
 \label{fig:t0p9}
\end{figure}

\begin{figure}[H]
 \begin{subfigure}{0.5\columnwidth}
  \centering
  \includegraphics[width=\columnwidth]{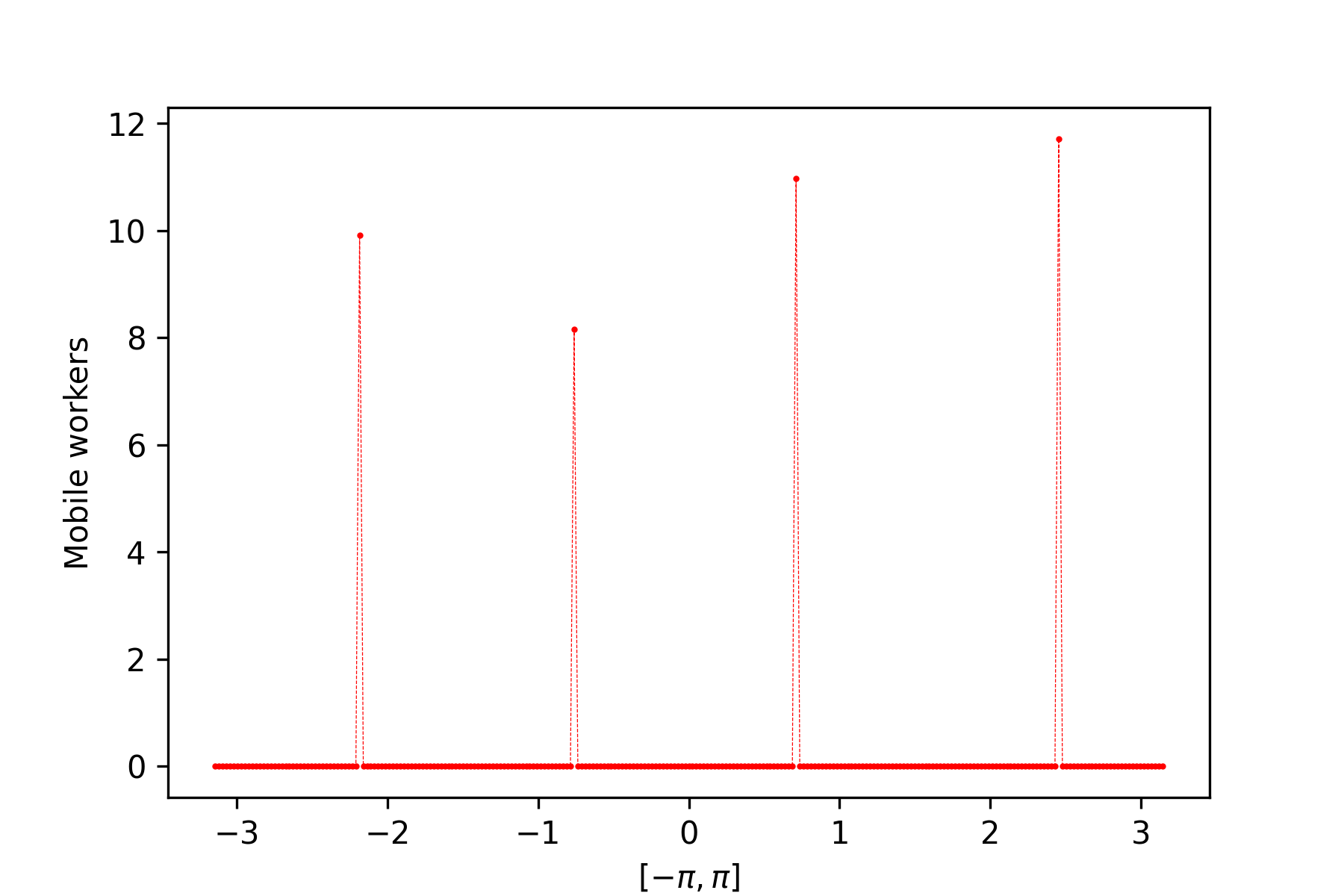}
  \caption{Mobile population}
 \end{subfigure}
 \begin{subfigure}{0.5\columnwidth}
  \centering
  \includegraphics[width=\columnwidth]{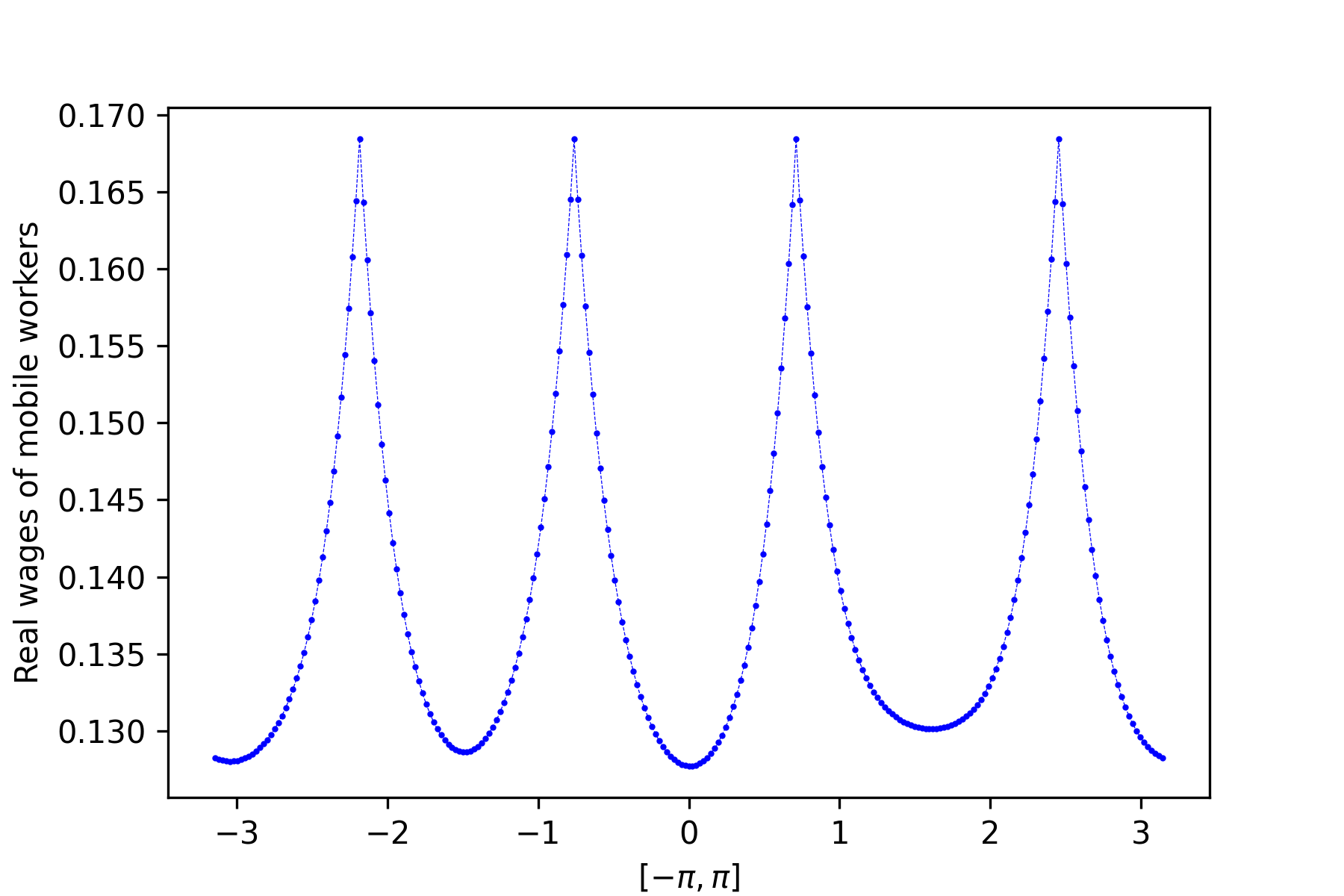}
  \caption{Real wage}
 \end{subfigure}\\
 \caption{Stationary solution for $(\sigma,\tau)=(3.0, 1.1)$}
 \label{fig:t1p1}
\end{figure}

\begin{figure}[H]
 \begin{subfigure}{0.5\columnwidth}
  \centering
  \includegraphics[width=\columnwidth]{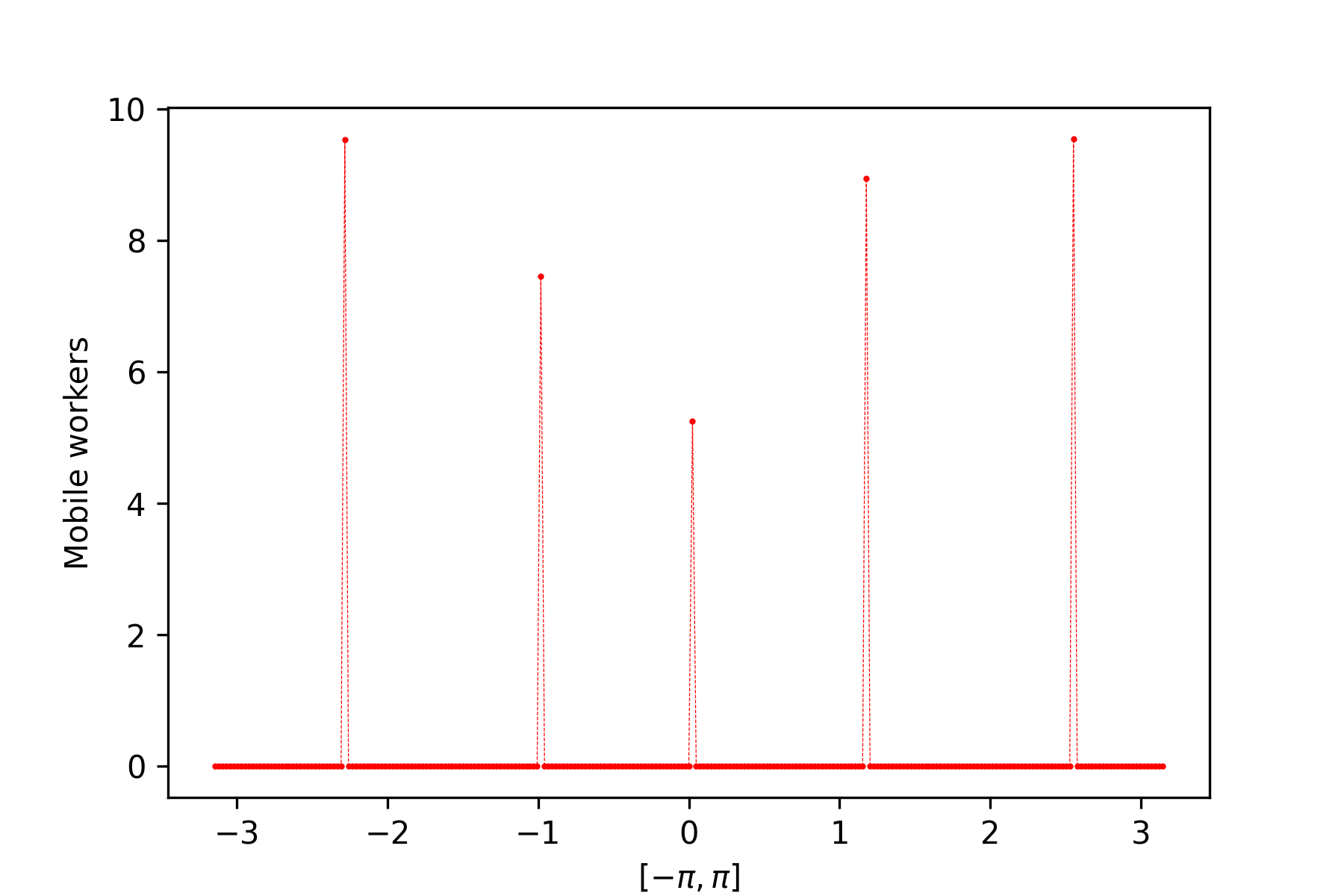}
  \caption{Mobile population}
 \end{subfigure}
 \begin{subfigure}{0.5\columnwidth}
  \centering
  \includegraphics[width=\columnwidth]{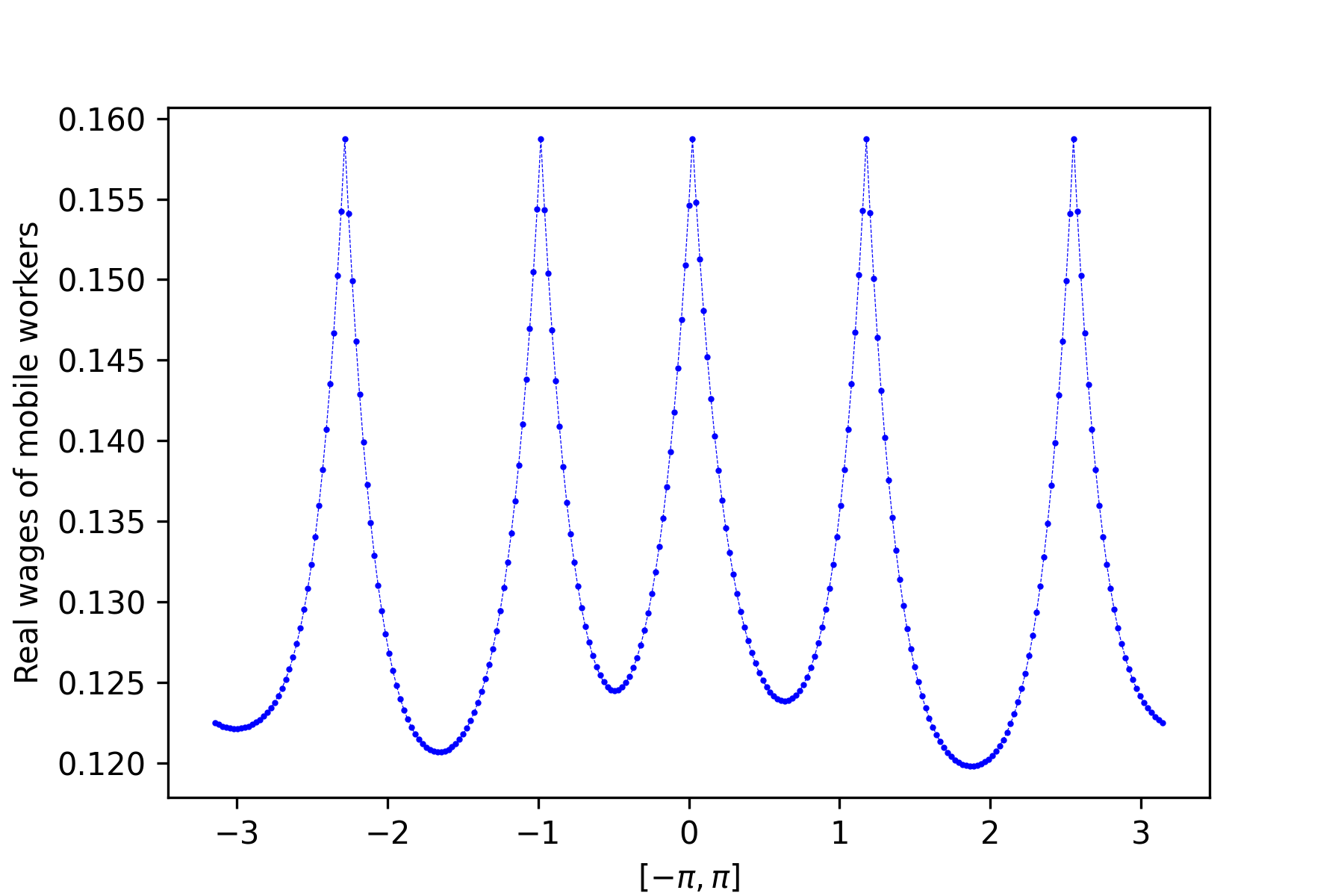}
  \caption{Real wage}
 \end{subfigure}\\
 \caption{Stationary solution for $(\sigma,\tau)=(3.0, 1.3)$}
 \label{fig:t1p3}
\end{figure}

\begin{figure}[H]
 \begin{subfigure}{0.5\columnwidth}
  \centering
  \includegraphics[width=\columnwidth]{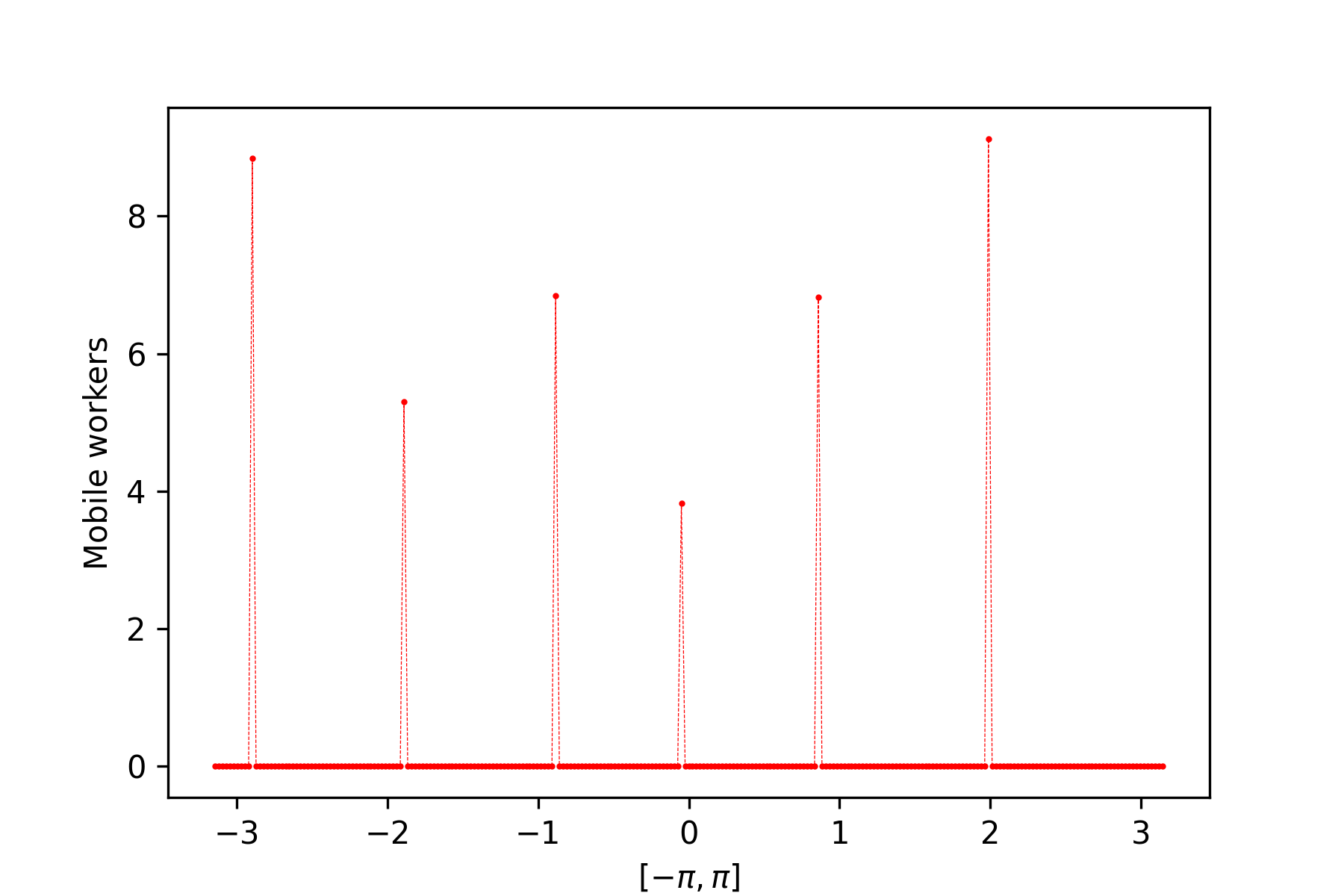}
  \caption{Mobile population}
 \end{subfigure}
 \begin{subfigure}{0.5\columnwidth}
  \centering
  \includegraphics[width=\columnwidth]{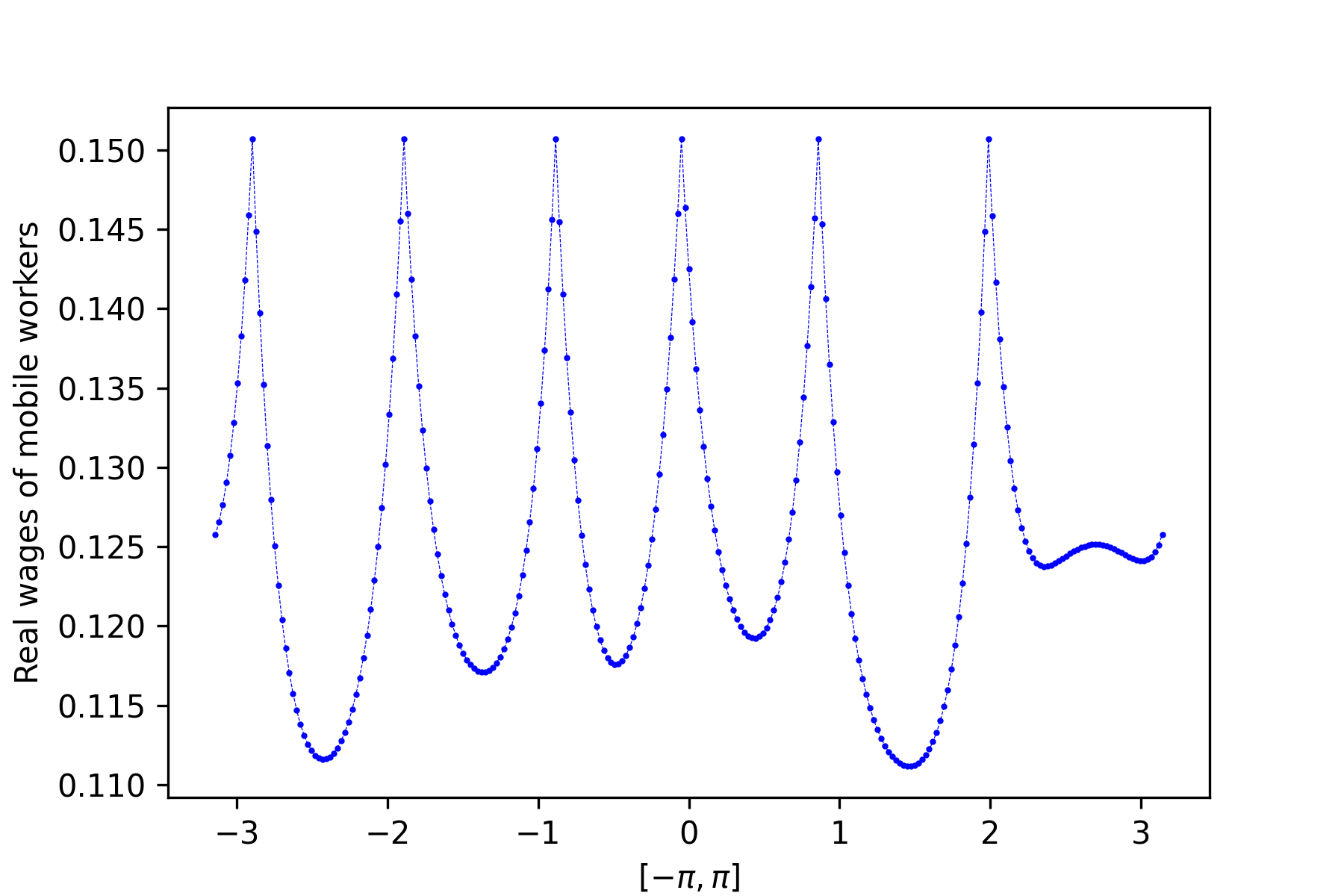}
  \caption{Real wage}
 \end{subfigure}\\
 \caption{Stationary solution for $(\sigma,\tau)=(3.0, 1.6)$}
 \label{fig:t1p6}
\end{figure}

\vspace{2mm}
\noindent
{\bf Note}: The location and size of spikes can vary considerably depending on initial values. This can be interpreted as indicating the historical dependence of urban formation.

\section{Stability of spiky stationary solutions}\label{sec:stabilityspikysol}
As seen in the previous section, numerical experiments suggest that a golobal solution asymptotically converges to a non-homogeneous stationary solution having some spikes (spiky stationary solution). In this section, we examine its stability. We continue to assume that the immobole population distribution is homogeneous $\phi=\bar{\phi}$ as in \eqref{stphi}.

\subsection{Spiky stationary solutions}
Mathematically, a spiky stationary solution $\bar{\lambda}$ having $R$ ($\in \mathbb{N}$) spikes at each centers $x_i\in S$ ($i=1,2,\cdots, R$) is expressed as a convex combination of the Dirac delta functions $\delta(x-x_i)$ as
\begin{equation}\label{lambdadeltafunctions}
\bar{\lambda}(x) = \sum_{i=1}^R \bar{\lambda}_i \delta(x-x_i),~x\in S.
\end{equation}
where $\bar{\lambda}_i\in[0,1]$ for $i=1,2,\cdots,R$. The Dirac delta function satisfies that
\begin{equation}\label{deltaprop1}
\int_S f(x)\delta(x-x_i)dx = f(x_i)
\end{equation}
for any contionus function $f$ on $S$ and that
\begin{equation}\label{deltaprop2}
\int_S \delta(x-x_i)dx = 1.
\end{equation}
From \eqref{totalmobpop} and \eqref{deltaprop2}, integrating \eqref{lambdadeltafunctions} on $S$ yields
\[
\sum_{i=1}^R\bar{\lambda}_i=\Lambda.
\]

\subsubsection{Uniform spiky solution}\label{subsubsec:uniformspikysol}
Since it is difficult to analyze general spiky stationary solutions analytically, wel focus our discussion below on {\it uniform spiky solutions}. That is, spikes are equidistantly distributed along $S$
\[
i\neq j \Leftrightarrow d(x_i, x_j) = \frac{2\pi\rho}{R},
\]
and the weights $\bar{\lambda}_i$ at each centers are uniform
\begin{equation}\label{spikyuniformlambda}
\bar{\lambda}_i = \bar{\lambda} = \frac{\Lambda}{R},~\forall i=1,2,\cdots,R.
\end{equation}
In this case, it can be shown that the price index, the nominal wage, and the real wage are uniform at each center. In fact, substituting the uniform spiky solution \eqref{lambdadeltafunctions} into the price index equation \eqref{priceindexeq}, we have the stationary price index
\begin{equation}\label{spikyhomGx}
\bar{G}(x)= \left[\frac{\bar{\lambda}}{F}\sum_{j=1}^R e^{-\alpha d(x,x_j)}\right]^{\frac{1}{1-\sigma}},~\forall x\in S.
\end{equation}
Therefore, at each centers $x=x_i$, it becomes
\begin{equation}\label{spikyuniformG}
\bar{G}_i \equiv \bar{G} = \left[\frac{\bar{\lambda}}{F}\sum_{j=1}^R e^{-\alpha d(x_i,x_j)}\right]^{\frac{1}{1-\sigma}},~\forall i=1,2,\cdots,R,
\end{equation}
which does not depend on any centers $x_i$. Similarly, substituting \eqref{lambdadeltafunctions}, \eqref{spikyhomGx}, and \eqref{spikyuniformG} into the nominal wage equation \eqref{nominalwageeq} and solving for $\bar{w}_i$ which stands for the stationary nominal wage at $x_i$, we have
\begin{equation}\label{spikyuniformw}
\bar{w}_i \equiv \bar{w}= \frac{\mu\bar{\phi}}{(\sigma-\mu)\ol{\lambda}}\int_S\frac{e^{-\alpha d(x_i,y)}}{\sum_{j=1}^R e^{-\alpha d(y, x_j)}}dy, ~\forall i=1,2,\cdots,R.
\end{equation}
Since the integral term on the right-hand side does not depend on any centers $x_i$, the nominal wage at any centers is also a constant which does not depend on $x_i$. As a result, the realwage at each centers $x_i$ is a constant 
\begin{equation}\label{spikyuniformomega}
\bar{\omega}_i \equiv \bar{\omega} = \bar{w}\bar{G}^{-\mu}, ~\forall i=1,2,\cdots,R.
\end{equation}

\subsection{Stability analysis}
Even when considering only uniform spiky solutions, it is difficult to investigate their stability analytically. We therefore discuss linear stability within a limited scope. Specifically, a small perturbation is added only to each centers of the uniform spiky stationary solution, not to the entire $S$.\footnote{The same method is also used in \citet{OhtakeYagi_Stab}.} Then, a perturbation $\Delta\lambda$ is given by 
\[
\Delta \lambda(t, x) = \sum_{i=1}^R \Delta\lambda_i(t)\delta(x-x_i),~\forall (t,x)\in [0,\infty)\times S.
\]
This perturbation gives rise to perturbations in the price index, the nominal wage, and the real wage, denoted by $\Delta G$, $\Delta w$, and $\Delta\omega$, respectively. Similar to Subsection \ref{subsec:contstabilityanalysis}, we can obtain the following linear system that governs the time evolution of these perturbations.
\begin{equation}\label{spikemodel}
\left\{
\begin{aligned}
&\Delta G_i(t) = \frac{\bar{G}^{\sigma}}{F(1-\sigma)}\sum_{j=1}^R \Delta\lambda_j(t) e^{-\alpha d(x_i,x_j)},\\
&\Delta w_i(t) = \frac{\mu}{\sigma F}\bar{\lambda}\bar{G}^{\sigma-1}\sum_{j=1}^R\Delta w_j(t)e^{-\alpha d(x_i, x_j)}\\
&\hspace{10mm}+ \frac{\mu}{\sigma F}\bar{w}\bar{G}^{\sigma-1}\sum_{j=1}^R\Delta\lambda_j(t) e^{-\alpha d(x_i,x_j)}\\
&\hspace{13mm}-\frac{\mu}{\sigma F^2}\bar{\lambda}\bar{w}\bar{G}^{2\sigma-2}\sum_{j=1}^R\Delta\lambda_j(t)\sum_{k=1}^Re^{-\alpha d(x_i, x_k)}e^{-\alpha d(x_k, x_j)}\\
&\hspace{16mm}-\frac{\mu\bar{\phi}}{\sigma F^2}\sum_{j=1}^R\Delta\lambda_j(t)\int_S\bar{G}(y)^{2\sigma-2}e^{-\alpha d(x_i,y)}e^{-\alpha d(x_j,y)}dy,\\
&\Delta\omega_i(t) =\bar{G}^{-\mu}\left(\Delta w_i(t)-\mu\frac{\bar{w}}{\bar{G}}\Delta G_i(t)\right),\\
&\frac{d}{dt}\Delta\lambda_i(t) = v\bar{\lambda}\Delta\omega_i(t),
\end{aligned}
\right.
\end{equation}
for $t\in[0,\infty),~i=1,2,\cdots,R$. Here, all the letters with bars represent stationary states of the uniform spiky solution given by \eqref{stphi}, \eqref{spikyuniformlambda}, \eqref{spikyhomGx}, \eqref{spikyuniformG}, and \eqref{spikyuniformw}. Let us define a matrix $E_{\rm X}\in\mathbb{R}^{R\times R}$ as 
\begin{equation}\label{matEij}
{E_{\rm X}}_{ij} = e^{-\alpha d(x_i,x_j)}
\end{equation}
and introduce the following matrices $A$, $B$, and $C$:
\begin{align}
&A_{ij} = \frac{\mu}{\sigma F}\bar{\lambda}\bar{G}^{\sigma-1}{E_{\rm X}}_{ij},\label{matA}\\
&B_{ij} = \frac{\mu}{\sigma F}\bar{w}\bar{G}^{\sigma-1}{E_{\rm X}}_{ij}\\
&\hspace{10mm}-\frac{\mu}{\sigma F^2}\bar{\lambda}\bar{w}\bar{G}^{2\sigma-2}\sum_ke^{-\alpha d(x_i, x_k)}e^{-\alpha d(x_k, x_j)}\\
&\hspace{13mm}-\frac{\mu\bar{\phi}}{\sigma \bar{\lambda}^2}\int_S\frac{e^{-\alpha d(y,x_i)}e^{-\alpha d(y,x_j)}}{\left[\sum_{k=1}^N e^{-\alpha d(y,x_k)}\right]^2}dy,\label{matB}\\
&C_{ij} = \frac{\mu}{F(\sigma-1)}\bar{w}\bar{G}^{\sigma-1}{E_{\rm X}}_{ij}.\label{matC}
\end{align}
By using these matrices, we can summarize \eqref{spikemodel} into 
\begin{equation}\label{lindynspike}
\frac{d}{dt}\Delta\lambda = \Gamma\Delta\lambda,
\end{equation}
where
\begin{equation}\label{matGammaspike}
\Gamma = v\bar{\lambda}\bar{G}^{-\mu}\left[\left(I-A\right)^{-1}B+C\right].
\end{equation}
The matrix $\Gamma$ can be proved to be symmetric, all of its eigenvalues are real numbers. Then, the stability of the uniform spiky solution can be determined by the maximum eigenvalue of $\Gamma$. That is, if the maximum eigenvalue of $\Gamma$ is negative, the uniform spiky solution is stable; otherwise, it is unstable.\footnote{This follows from the well known results of the theory of dynamical systems.\citep[Chapter 5, p.92]{HirSm74}}

It becomes difficult to find the integral in \eqref{spikyuniformw} and the integral in \eqref{matC} analytically as $R$ increases. Below, we discuss the cases where $R = 2$ and $3$.

\subsubsection{In the case of $R=2$}\label{subsubsec:heatmapR2}
The inregral in \eqref{spikyuniformw} becomes
\[
\int_S\frac{e^{-\alpha d(x_i,y)}}{\sum_{j=1}^2 e^{-\alpha d(y, x_j)}}dy = \rho\pi.
\]
A careful calculation shows that the integral
\[
\int_S\frac{e^{-\alpha d(y,x_i)}e^{-\alpha d(y,x_j)}}{\left[\sum_{k=1}^2 e^{-\alpha d(y,x_k)}\right]^2}dy
\]
in \eqref{matB} for $(i,j)=(1,1),~(2,2)$ bocomes
\[
\frac{1}{\alpha}\left[\ln\frac{1+e^{\alpha\rho\pi}}{1+e^{-\alpha\rho\pi}}+\frac{e^{-\alpha\rho\pi}-e^{\alpha\rho\pi}}{(1+e^{-\alpha\rho\pi})(1+e^{\alpha\rho\pi})}\right],
\]
and for $(i,j)=(1,2),~(2,1)$, it becomes
\[
-\frac{1}{\alpha}\left[\frac{e^{-\alpha\rho\pi}-e^{\alpha\rho\pi}}{(1+e^{-\alpha\rho\pi})(1+e^{\alpha\rho\pi})}\right].
\]
These analytical results allow us to compute the matrix $\Gamma$. Nevertheless, it is difficult to obtain analytical expressions for the eigenvalues of $\Gamma$. Therefore, the detailed calculations must be performed using numerical computation as follows.

Fig.~\ref{fig:2spikes-eigenvalues} illustrates the maximum eigenvalues of the matrix $\Gamma$ under $\mu=0.6$. Fig.~\ref{subfig:2maxeigenvalue} plots the maximum eigenvalue as a function of $\tau\geq 0$ for $\sigma = 3.0$. Fig.~\ref{subfig:2heatmap} shows the the maximum eigenvalues plotted as a heatmap on the $(\tau, \sigma)$-plane. The thick black curve indicates the points where the maximum eigenvalues are zero. For each value of $\sigma$ that satisfies the assumption of no black holes, we can see that the uniform spiky solution becomes unstable when the transport costs are sufficiently low, and stable when they are sufficiently high.\footnote{For the spike model \eqref{spikemodel}, the assumption of no black holes is given by
\begin{equation}\label{nbh:R23}
\sigma>1+\mu,
\end{equation}
which is same as \eqref{nbh}.}

\begin{figure}[H]
 \begin{subfigure}{0.5\columnwidth}
  \centering
  \includegraphics[width=\columnwidth]{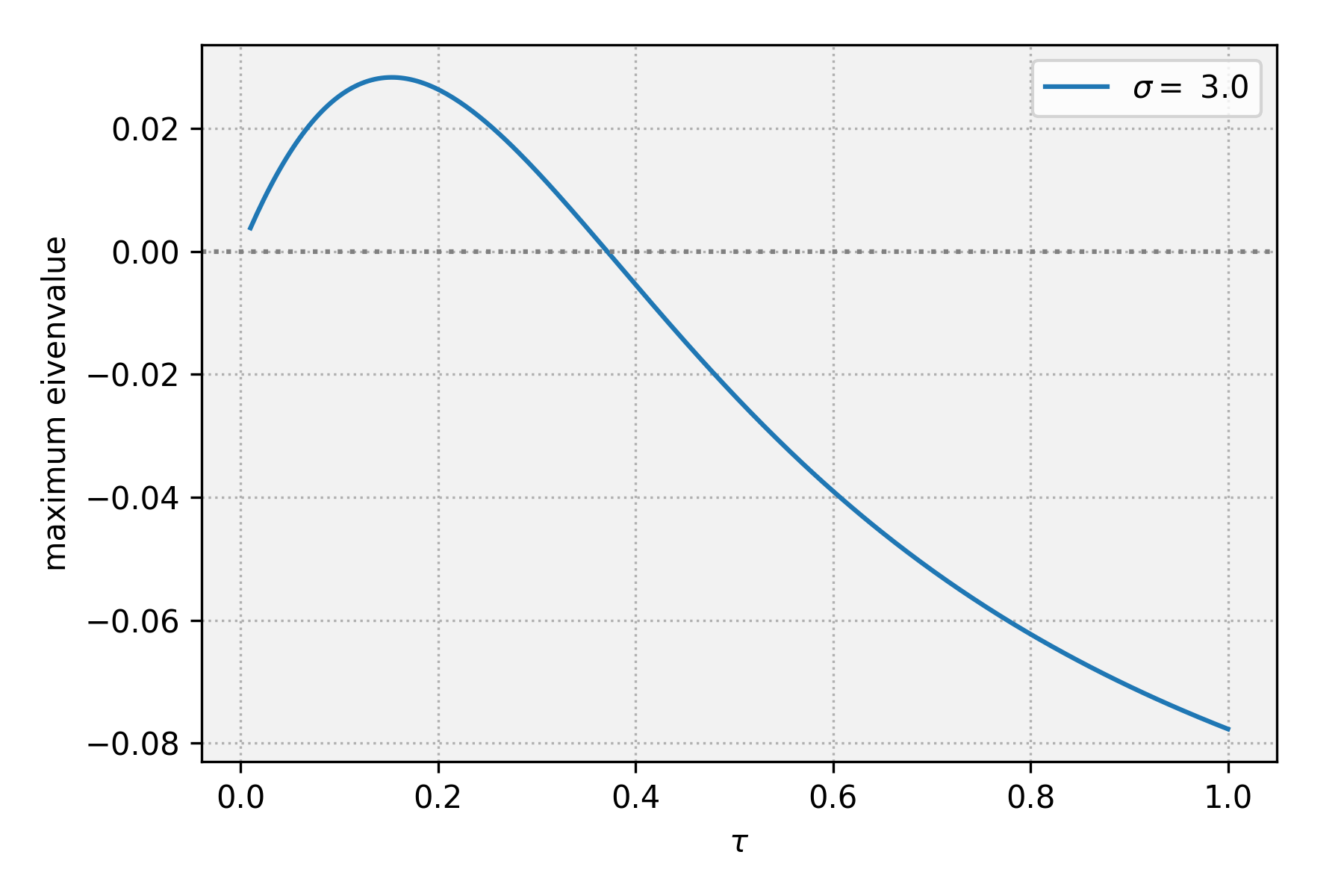}
  \caption{Maximum eigenvalue}\label{subfig:2maxeigenvalue}
 \end{subfigure}
 \begin{subfigure}{0.5\columnwidth}
  \centering
  \includegraphics[width=1.09\columnwidth]{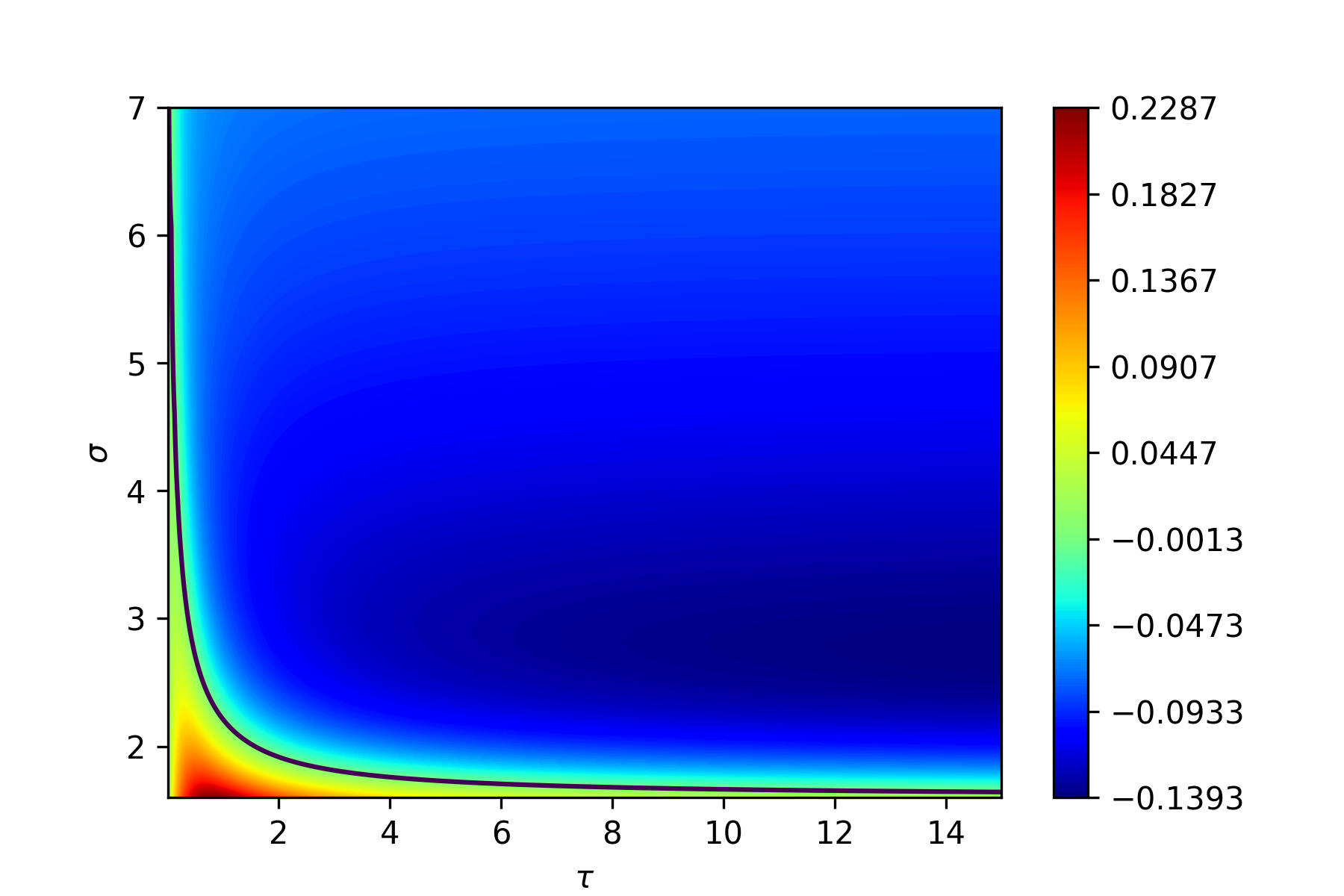}
  \caption{Heatmap of eigenvalues}\label{subfig:2heatmap}
 \end{subfigure}\\
 \caption{Eigenvalue plot of the matrix $\Gamma$ when $R=2$}
 \label{fig:2spikes-eigenvalues}
\end{figure}

\subsubsection{In the case of $R=3$}\label{subsubsec:heatmapR3}
The inregral in \eqref{spikyuniformw} becomes
\[
\int_S\frac{e^{-\alpha d(x_i,y)}}{\sum_{j=1}^2 e^{-\alpha d(y, x_j)}}dy = \frac{2}{3}\rho\pi.
\]
A careful calculaion shows that the integral
\[
\int_S\frac{e^{-\alpha d(y,x_i)}e^{-\alpha d(y,x_j)}}{\left[\sum_{k=1}^3 e^{-\alpha d(y,x_k)}\right]^2}dy
\]
in \eqref{matB} for $(i,j)=(1,1),~(2,2),~(3,3)$ becomes
\[
\begin{aligned}
\frac{2}{3}\pi\rho &+ \frac{1}{\alpha}\left[\frac{1}{\left(1+e^{-\alpha\rho\frac{2\pi}{3}}\right)^2} + \frac{1}{\left(1+e^{\alpha\rho\frac{2\pi}{3}}\right)^2} - 1\right]\ln\frac{2+e^{\alpha\rho\frac{2\pi}{3}}}{2+e^{-\alpha\rho\frac{2\pi}{3}}}\\
&+ \frac{1}{\alpha}\left[\frac{1}{\left(1+e^{-\alpha\rho\frac{2\pi}{3}}\right)^2} + \frac{1}{\left(1+e^{\alpha\rho\frac{2\pi}{3}}\right)^2} + 1\right]\frac{e^{-\alpha\rho\frac{2\pi}{3}}-e^{\alpha\rho\frac{2\pi}{3}}}{\left(2+e^{-\alpha\rho\frac{2\pi}{3}}\right)\left(2+e^{\alpha\rho\frac{2\pi}{3}}\right)},
\end{aligned}
\]
and for $(i,j)=(1,2),~(1,3),~(2,3),~(2,1),~(3,1),~(3,2)$, it becomes
\[
\begin{aligned}
&\frac{1}{\alpha}\left[\frac{1}{\left(1+e^{-\alpha\rho\frac{2\pi}{3}}\right)\left(1+e^{\alpha\rho\frac{2\pi}{3}}\right)}\right]
\ln\frac{2+e^{\alpha\rho\frac{2\pi}{3}}}{2+e^{-\alpha\rho\frac{2\pi}{3}}}\\
&\hspace{5mm}-\frac{1}{\alpha}\left[\frac{e^{-\alpha\rho\frac{2\pi}{3}}+e^{\alpha\rho\frac{2\pi}{3}}+1}{\left(1+e^{-\alpha\rho\frac{2\pi}{3}}\right)\left(1+e^{\alpha\rho\frac{2\pi}{3}}\right)}\right]\frac{e^{-\alpha\rho\frac{2\pi}{3}}-e^{\alpha\rho\frac{2\pi}{3}}}{\left(2+e^{-\alpha\rho\frac{2\pi}{3}}\right)\left(2+e^{\alpha\rho\frac{2\pi}{3}}\right)}.
\end{aligned}
\]
As in the case where $R=2$, specific computations of the eigenvalues must be performed numerically.

The meaning of the following figures is the same as for the case where $R=2$. Fig.~\ref{fig:3spikes-eigenvalues} illustrates the maximum eigenvalues of the matrix $\Gamma$ under $\mu=0.6$. Fig.~\ref{subfig:3maxeigenvalue} plots the maximum eigenvalue as a function of $\tau\geq 0$ for $\sigma = 3.0$. Fig.~\ref{subfig:3heatmap} shows the the maximum eigenvalues plotted as a heatmap on the $(\tau, \sigma)$-plane. The thick black curve indicates the points where the maximum eigenvalues are zero. It can be seen that the parameter dependence of the eigenvalues is very similar to the previous case of $R=2$.\footnote{The assumption of no black holes is also given by \eqref{nbh:R23} as in the case of $R=2$.}

\begin{figure}[H]
 \begin{subfigure}{0.5\columnwidth}
  \centering
  \includegraphics[width=\columnwidth]{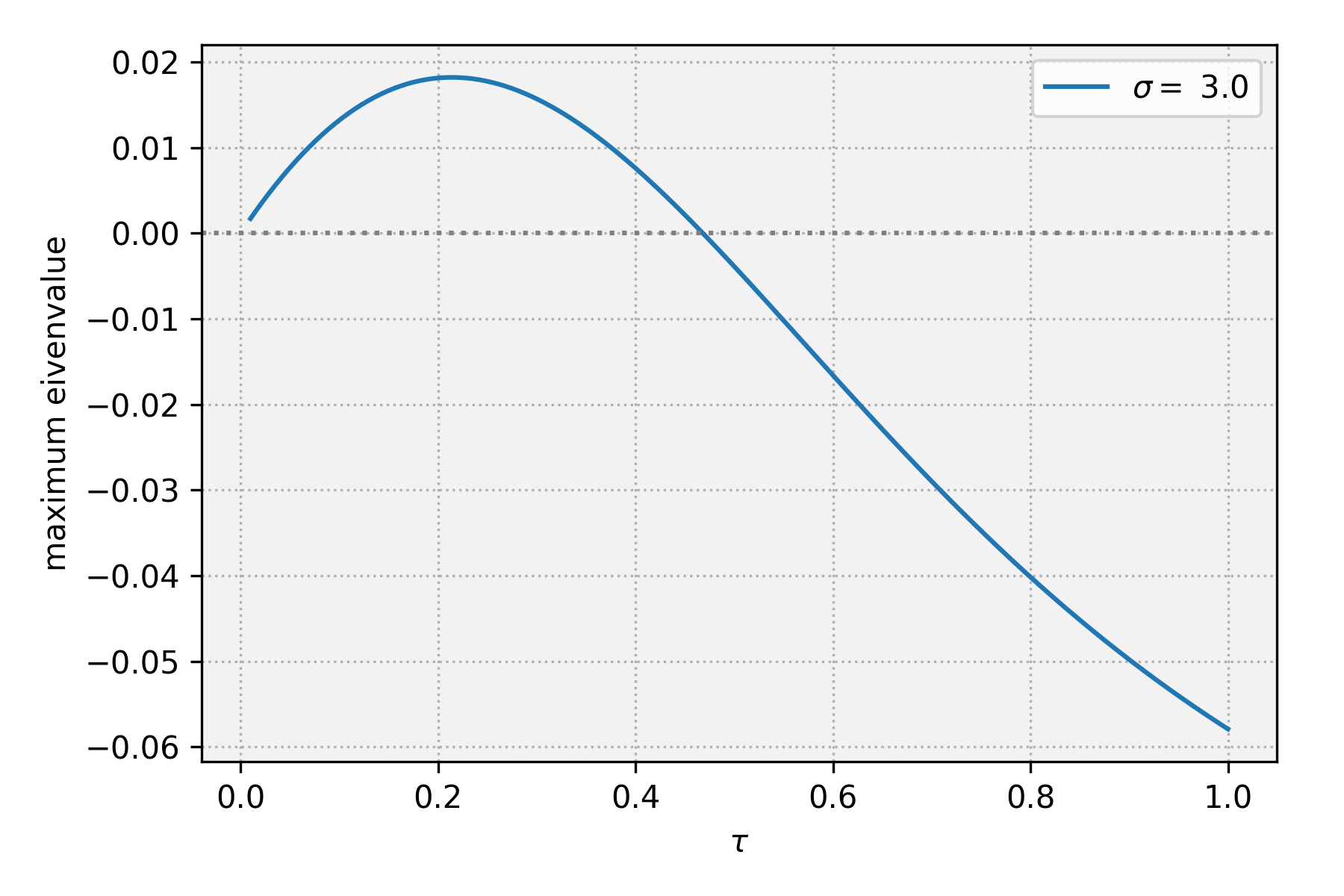}
  \caption{Maximum eigenvalue}\label{subfig:3maxeigenvalue}
 \end{subfigure}
 \begin{subfigure}{0.5\columnwidth}
  \centering
  \includegraphics[width=1.09\columnwidth]{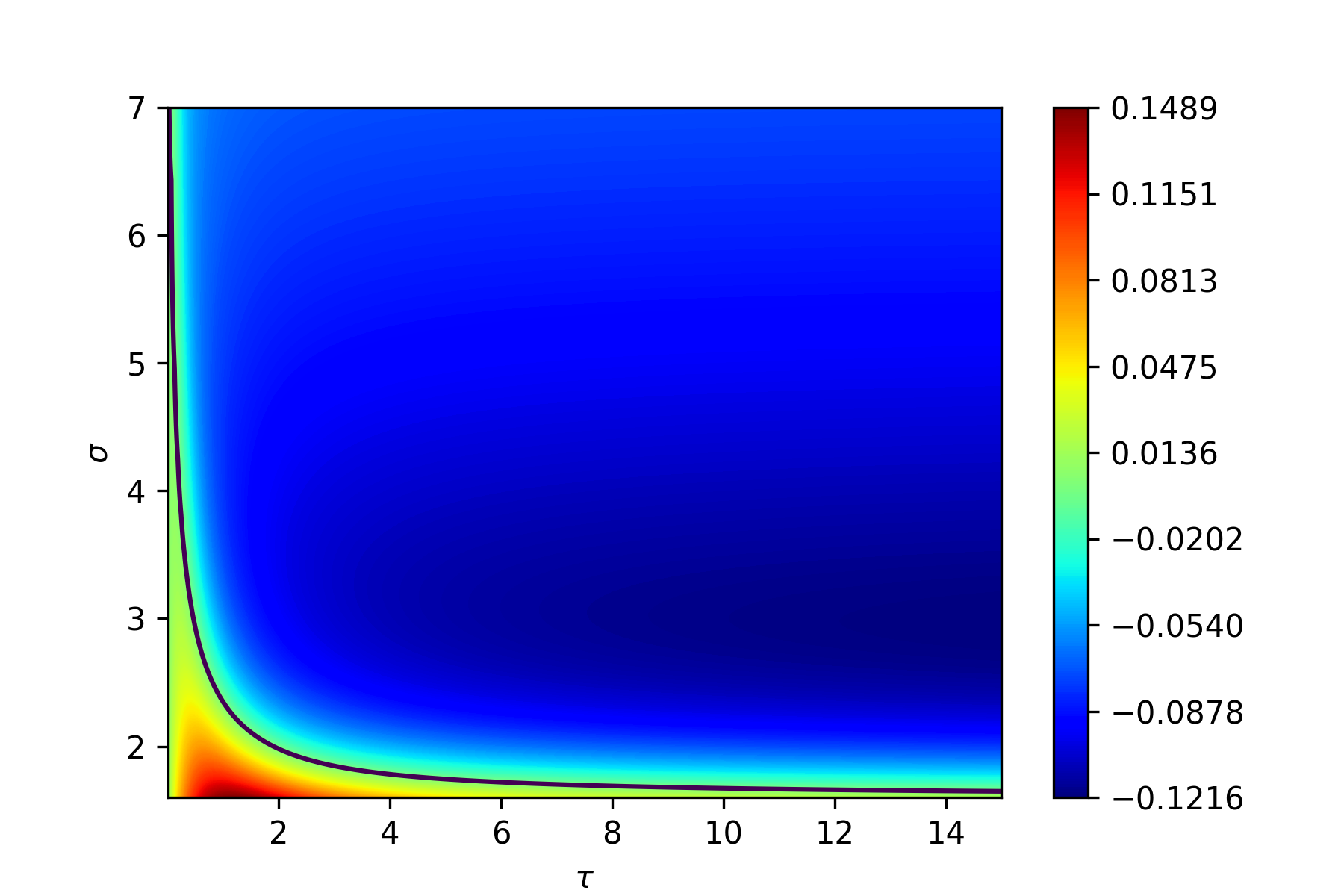}
  \caption{Heatmap of eigenvalues}\label{subfig:3heatmap}
 \end{subfigure}\\
 \caption{Eigenvalue plot of the matrix $\Gamma$ when $R=3$}
 \label{fig:3spikes-eigenvalues}
\end{figure}

\section{Conclusion and discussion}\label{sec:cd}
In this paper, we have considered the FE model in continuous space. We have formulated it as an initial value problem for an ordinary differential equation in a Banach space, and constructed a unique global solution. The stability of the homogeneous stationary solution has been investigated using Fourier analysis, and it has been found that the absolute value of the frequency of unstable modes decreases as transport costs decline. The asymptotic behavior of solutions that started near the homogeneous stationary solution has been simulated numerically, and it has been observed that numerical solutions asymptotically approach distributions with several spikes. The number of spikes decreases as transport costs decline. This has been further corroborated by examining the stability of the uniform spiky solutions.

There are two significant points to this study. First, by clarifying the unique existence and behavior of the global solution over time, this study has provided a mathematical foundation for existing research focusing mainly on stationary solutions. Second, this study has confirmed that the behavior of solutions to the FE model in a continuous space is qualitatively similar to that of many basic NEG models, such as the original CP model and the quasi-log-linear utility model.\footnote{Subsection \ref{subsec:comparison-models} provides a more detailed comparison with other models.} One might consider this obvious, but recalling that there is a model in which the behavior of solutions changes qualitatively in a multi-regional setting,\footnote{\citet{Ohtake2022agg} showed that in a model where a quasi-linear utility having a quadratic subutility, the homogeneous stationary solution is always unstable when the number of regions is a multiple of four.} it is by no means self-evident, and theoretical confirmation is essential.

Mathematically, it would be interesting to explore a more applicable sufficient condition for the unique existence of a global solution. In constructing a global solution, the fixed point theorem has been used; however, the sufficient condition assumed there is stricter than the parameter range in which the simulation actually works. Since the instantaneous equilibrium equation is a Fredholm-type integral equation, it may be possible to utilize its characteristics to make the sufficient condition more applicable.

The results of this study reinforce a practical implication in NEG. That is, as transport costs decline, agglomeration in a small number of regions tends to progress. \citet[p.~640]{BehLamOttTab2007} refer to this phenomenon as ``straw effect'' in the context of international trade. They states ``{\it reductions in transport frictions between large and small countries can decrease the attractiveness of the small country. We may call such a phenomenon a `straw effect', because economic activities migrate to large countries through new highways and infrastructure as juice in a glass is sucked up by a straw}''. \citet[p.~19]{Otta2008inf} also refers to this straw effect in the broad overview of infrastructure and economic geography. In order to address even more practical problems, it would be possible to extend the model to incorporate more realistic factors, such as non-uniform spaces, asymmetric or discontinuous transport costs, and so on.).

\section{Appendix}\label{sec:appendix}

\subsection{Proof for Theorem \ref{th:fp}}\label{pr:thfp}
We define an operator $E:C(\Omega)\to C(\Omega)$ by the right-hand side of \eqref{fp}
\begin{equation}\label{mapE}
\begin{aligned}
E[W](x)&:=\frac{\mu}{\sigma F}\int_\Omega(W(y)\left|\lambda(y)\right|+\phi(y))\\
&\hspace{20mm}\times G[\lambda](y)^{\sigma-1}T(x,y)^{1-\sigma}dy,\hspace{3mm}x\in\Omega.
\end{aligned}
\end{equation}
First, we must verify that $E$ is a mapping from the closed subset $C_{0+}(\Omega)$ to itself. It is immediate that for any $W\in C_{0+}(\Omega)$, the integrand of \eqref{mapE} is nonnegative almost everywhere $y\in\Omega$ and thus $E[W]\in C_{0+}(\Omega)$.

Second, we show that $E$ is a contraction on $C_{0+}(\Omega)$. Let $\lambda\in L^1(\Omega)$ satisfying \eqref{Lam1lambdaLam2} be fixed. For $W_1\in C_{0+}(\Omega)$ and $W_2\in C_{0+}(\Omega)$, by \eqref{Tbound}, \eqref{opG}, and \eqref{Lam1lambdaLam2}, we see that
\begin{align}
\left|E[W_1](x)-E[W_2](x)\right| &= \left|\frac{\mu}{\sigma}\int_\Omega\frac{\left(W_1(y)-W_2(y)\right)\left|\lambda(y)\right|T(x,y)^{1-\sigma}}{\int_\Omega\left|\lambda(z)\right|T(y,z)^{1-\sigma}dz}dy\right|\\
&\leq \frac{\mu}{\sigma}\int_\Omega\frac{\left|\lambda(y)\right|T(x,y)^{1-\sigma}}{\int_\Omega\left|\lambda(z)\right|T(y,z)^{1-\sigma}dz}dy\left\|W_1-W_2\right\|_{\infty}\label{contintegralterm}\\
&\leq \frac{\mu}{\sigma}\left(\frac{T_{\max}}{T_{\min}}\right)^{\sigma-1}\frac{\Lambda_2}{\Lambda_1}\left\|W_1-W_2\right\|_{\infty}.
\end{align}
This immediately yields 
\[
\left\|E[W_1]-E[W_2]\right\|_{\infty}\leq \frac{\mu}{\sigma}\left(\frac{T_{\max}}{T_{\min}}\right)^{\sigma-1}\frac{\Lambda_2}{\Lambda_1}\left\|W_1-W_2\right\|_{\infty}.
\]
Therefore, if \eqref{sufficientcondition} holds, then $E$ becomes a contraction. The Banach fixed-point theorem completes the proof.
\qed

\vspace{1mm}
\noindent
{\bf Note}: The condition \eqref{sufficientcondition} is generally restrictive. Clearly, it is easier to satisfy this condition by choosing a function $T$ for transport costs that does not vary greatly across space (such as using a polynomial rather than an exponential function). After all, this condition follows from the requirement that the inequality
\begin{equation}\label{musigmalmTlmTintegralsm1}
\frac{\mu}{\sigma}\int_\Omega\frac{\left|\lambda(y)\right|T(x,y)^{1-\sigma}}{\int_\Omega\left|\lambda(z)\right|T(y,z)^{1-\sigma}dz}dy < 1
\end{equation}
holds for any $\lambda$ as shown in the proof. For a specific $\lambda$, this holds under arbitrary transport costs. For example, if $\lambda$ is homogeneous $\lambda(x)=\ol{\lambda}~(\forall x)$, then the left hand side of \eqref{musigmalmTlmTintegralsm1} is
\begin{align}
\frac{\mu}{\sigma}\int_\Omega\frac{\ol{\lambda}T(x,y)^{1-\sigma}}{\int_\Omega\ol{\lambda}T(y,z)^{1-\sigma}dz}dy
&= \frac{\mu}{\sigma} \frac{\int_\Omega T(x,y)^{1-\sigma}dy}{\int_\Omega T(y, z)^{1-\sigma}dz}\\
&=\frac{\mu}{\sigma}<1.
\end{align}
Other examples include, for instance, if $\lambda$ is a single delta function $\lambda(x)=\delta(x-x^*)$, then the left hand side of \eqref{musigmalmTlmTintegralsm1} is
\begin{align}
\frac{\mu}{\sigma}\int_\Omega\frac{\delta(y-x_*)T(x,y)^{1-\sigma}}{\int_\Omega\delta(z-x_*)T(y,z)^{1-\sigma}dz}dy
&= \frac{\mu}{\sigma} T(x,x_*)^{1-\sigma}\\
&\leq\frac{\mu}{\sigma}<1.
\end{align}
Thus, as long as we consider solutions in the vicinity of these specific solutions, the condition \eqref{sufficientcondition} may not be particularly restrictive.

\subsection{Proof for Lemma \ref{lem:bounded}}\label{subsec:prooflembounded}
First, $\lambda\in Q_{\lambda_0}$ is shown to be bounded. For any $\lambda\in Q_{\lambda_0}$ 
\begin{equation}\label{ll0b}
\left\|\lambda-\lambda_0\right\|_{L^1}\leq b.
\end{equation}
holds. From the basic properties of normed spaces, 
\begin{equation}\label{absnn}
\left|\left\|\lambda\right\|_{L^1}-\left\|\lambda_0\right\|_{L^1}\right| \leq \left\|\lambda-\lambda_0\right\|_{L^1}
\end{equation}
holds.\footnote{See \citet[p.~769~(23c)]{ZeidlerFixedPoint}.} From \eqref{ll0b} and \eqref{absnn}, it is verified that 
\begin{equation}\label{absb}
\left|\left\|\lambda\right\|_{L^1}-\left\|\lambda_0\right\|_{L^1}\right| \leq b.
\end{equation}
Since $\left\|\lambda_0\right\|_{L^1}=\Lambda$ and $\Lambda-b>0$, the inequality \eqref{absb} immediately yields
\begin{equation}\label{lambdabound}
0<\Lambda-b \leq \left\|\lambda\right\|_{L^1} \leq \Lambda+b.
\end{equation}

Second, the range of the operator $G$ is shown to be bounded. By \eqref{Tbound}, \eqref{opG}, and \eqref{lambdabound}, we see that
\begin{align}\label{Glowerupperbounds}
0<F^{\frac{1}{\sigma-1}}T_{\min}\left(\Lambda+b\right)^{\frac{1}{1-\sigma}}\leq 
\left|G[\lambda](x)\right| 
\leq F^{\frac{1}{\sigma-1}}T_{\max}\left(\Lambda-b\right)^{\frac{1}{1-\sigma}}.
\end{align}
Thus, we obtain 
\begin{equation}\label{Gbound}
0<G_{L}\leq \left\|G[\lambda]\right\|_\infty\leq G_{U},
\end{equation}
where
\begin{align}
&G_{L} := F^{\frac{1}{\sigma-1}}T_{\min}\left(\Lambda+b\right)^{\frac{1}{1-\sigma}},\label{Glower}\\
&G_{U} := F^{\frac{1}{\sigma-1}}T_{\max}\left(\Lambda-b\right)^{\frac{1}{1-\sigma}}\label{Gupper}
\end{align}

Third, the range of the operator $w$ is shown to be bounded.  By \eqref{Tbound}, \eqref{totalimpop}, \eqref{fp}, \eqref{lambdabound}, \eqref{Gbound}, and \eqref{Gupper}, we see that
\[
\begin{aligned}
\left|w[\lambda](x)\right|
&= \left|\frac{\mu}{\sigma F}\int_\Omega\left(w[\lambda](y)\left|\lambda(y)\right|+\phi(y)\right)G[\lambda](y)^{\sigma-1}T(x,y)^{1-\sigma}dy\right| \\
&\leq \frac{\mu}{\sigma F}(\Lambda+b)G_{U}^{\sigma-1}T_{\min}^{1-\sigma}\left\|w[\lambda]\right\|_{\infty}
+ \frac{\mu}{\sigma F}\Phi G_{U}^{\sigma-1}T_{\min}^{1-\sigma}\\
&\leq \frac{\mu}{\sigma}\left(\frac{T_{\max}}{T_{\min}}\right)^{\sigma-1}\frac{\Lambda+b}{\Lambda-b}\left\|w[\lambda]\right\|_\infty
+ \frac{\mu}{\sigma}\left(\frac{T_{\max}}{T_{\min}}\right)^{\sigma-1}\frac{\Phi}{\Lambda-b}.
\end{aligned}
\]
This yields
\[
\left\{1-\frac{\mu}{\sigma}\left(\frac{T_{\max}}{T_{\min}}\right)^{\sigma-1}\frac{\Lambda+b}{\Lambda-b}\right\}\left\|w[\lambda]\right\|_\infty
\leq \frac{\mu}{\sigma}\left(\frac{T_{\max}}{T_{\min}}\right)^{\sigma-1}\frac{\Phi}{\Lambda-b}.
\]
Therefore, if \eqref{sufficientconditionLb} holds, then
\begin{equation}\label{wbound}
\left\|w[\lambda]\right\|_\infty
\leq  \frac{\frac{\mu}{\sigma}\left(\frac{T_{\max}}{T_{\min}}\right)^{\sigma-1}\frac{\Phi}{\Lambda-b}}{1-\frac{\mu}{\sigma}\left(\frac{T_{\max}}{T_{\min}}\right)^{\sigma-1}\frac{\Lambda+b}{\Lambda-b}} =: w_{U}.
\end{equation}

Fourth, the range of the operator $\omega$ is shown to be bounded. By \eqref{opomega}, \eqref{Gbound}, and \eqref{wbound}, we have
\begin{equation}\label{omegabound}
\left\|\omega[\lambda]\right\|_\infty
\leq w_{U}G_{L}^{-\mu} =: \omega_{U}.
\end{equation}

Finally, we can show that the range of the operator $\Psi$ is bounded. By \eqref{opPsi}, \eqref{lambdabound}, and \eqref{omegabound}, we have
\begin{align}
\left\|\Psi[\lambda]\right\|_{L^1}&=
\int_\Omega\left|\Psi[\lambda](x)\right|dx\\
&\leq v\int_\Omega\left|\omega[\lambda](x)\lambda(x)\right|dx\\
&\hspace{5mm}+\frac{v}{\Lambda}\left|\int_\Omega\omega[\lambda](y)\lambda(y)dy\right|\int_\Omega\left|\lambda(x)\right|dx\\
&\leq v\omega_{U}(\Lambda+b)\left(1+\frac{\Lambda+b}{\Lambda}\right)=:K.
\end{align}
Note that the constant $K$ does not depend on $\lambda_0$.
\qed

\subsection{Proof for Lemma \ref{lem:lip}}\label{subsec:prooflemLip}
Let us introduce some concepts from functional analysis for the proof. Let an operator $\cl{F}:\mathfrak{X}\to \mathfrak{Y}$ be Fr\'{e}chet differentiable on a nonempty convex open set $\mathcal{D}\subset\mathfrak{X}$, where $\mathfrak{X}$ and $\mathfrak{Y}$ are Banach spaces with their norms $\left\|\cdot\right\|_{\mathfrak X}$ and $\left\|\cdot\right\|_{\mathfrak Y}$, respectively.\footnote{In this proof, we specifically consider $\mathcal{D}=C_+(\Omega)$ defined by the set of positive functions of $C(\Omega)$. In fact, since the lower bound of $G$ is positive (see \eqref{Glowerupperbounds}), it can be shown that the lower bounds of $w$ and $\omega$ are also positive, i.e., $w,\omega\in C_+(\Omega)$.} The Fr\'{e}chet derivative $\cl{F}^\prime(x)$ of $\cl{F}$ is defined by a linear operator $\mathfrak X$ to $\mathfrak Y$ satisfying
\begin{equation}
\cl{F}(x+h)-\cl{F}(x) = \cl{F}^\prime(x)(h) + o(\left\|h\right\|_{\mathfrak{X}}),
\end{equation}
where $o$ is the little-o notation. The operator norm of a linear operator $A:\mathfrak{X}\to \mathfrak{Y}$ is defined by
\begin{equation}\label{}
\left\|A\right\|_{\rm op}:=\sup_{\left\|h\right\|_{\mathfrak{X}}=1}\left\|A(h)\right\|_{\mathfrak{Y}}.
\end{equation}
In the following, we often use the well-known fact that
\begin{equation}\label{LipFrechet}
\left\|\cl{F}(x_1)-\cl{F}(x_2)\right\|_{\mathfrak Y}\leq \sup_{x\in D}\left\|\cl{F}^\prime(x)\right\|_{\rm op}\left\|x_1-x_2\right\|_{\mathfrak X}
\end{equation}
for any $x_1$ and $x_2$ in $\mathfrak{X}$.\footnote{See \citet[4.1b, p.191]{ZeidlerFixedPoint}.}

Firstly, we show that the operator $G$ is Lipschitz continuous. Let us denote $G[\lambda_i]$ by $G_i$ for $i=1,2$. From \eqref{opG}, we see that
\begin{align}
\left|G_1(x)-G_2(x)\right|
&\leq F^{\frac{1}{\sigma-1}}\left|g_1^{\frac{1}{1-\sigma}} - g_2^{\frac{1}{1-\sigma}}\right|\\
&\leq F^{\frac{1}{\sigma-1}}\left\|g_1^{\frac{1}{1-\sigma}} - g_2^{\frac{1}{1-\sigma}} \right\|_{\infty},\label{G1G2F11msigmag1g2maxnorm}
\end{align}
where $g_i(x):=\int_\Omega|\lambda_i(y)|T(x,y)^{1-\sigma}dy$ for $i=1,2$. From \eqref{Tbound} and \eqref{lambdabound}, it is easy to see that
\begin{equation}\label{giest}
g_i \geq T_{\max}^{1-\sigma}(\Lambda-b)>0
\end{equation}
for $i=1,2$. The Fr\'{e}chet derivative of $\cl{F}(g)=g^{\frac{1}{1-\sigma}}$ is $\cl{F}^\prime(g)=\frac{1}{1-\sigma}g^{\frac{\sigma}{1-\sigma}}$. From \eqref{giest}, the operator norm satisfies
\begin{align}
\left\|\cl{F}^\prime(g)\right\|_{\rm op} &= \sup_{\left\|h\right\|_\infty=1}\left\|\frac{1}{1-\sigma}g^{\frac{\sigma}{1-\sigma}}h\right\|_{\infty}\\
&= \frac{1}{\sigma-1}\left\|g^{\frac{\sigma}{1-\sigma}}\right\|_\infty\\
&\leq \frac{1}{\sigma-1}T_{\max}^\sigma(\Lambda-b)^{\frac{\sigma}{1-\sigma}}\label{1sig1TmaxsigLambsig1sig}
\end{align}
By using \eqref{LipFrechet} and \eqref{1sig1TmaxsigLambsig1sig}, we obtain
\begin{equation}\label{g1g211sigmanormmuGL}
\left\|g_1^{\frac{1}{1-\sigma}} - g_2^{\frac{1}{1-\sigma}}\right\|_{\infty}
\leq \frac{1}{\sigma-1}T_{\max}^{\sigma}(\Lambda-b)^{\frac{\sigma}{1-\sigma}}\left\|g_1-g_2\right\|_\infty
\end{equation}
It follows that
\begin{align}
\left|g_1(x)-g_2(x)\right|
&= \left|\int_\Omega\left(|\lambda_1(y)|-|\lambda_2(y)|\right)T(x,y)^{1-\sigma}dy\right|\\
&\leq T_{\min}^{1-\sigma}\int_\Omega\left||\lambda_1(y)|-|\lambda_2(y)|\right|dy \\
&\leq T_{\min}^{1-\sigma}\left\|\lambda_1-\lambda_2\right\|_{L^1}\label{g1g2lips}
\end{align}
The last inequality is due to
\begin{equation}\label{abslm12}
\left||\lambda_1(y)|-|\lambda_2(y)|\right|\leq\left|\lambda_1(y)-\lambda_2(y)\right| \text{~for a.e.~}y\in\Omega.
\end{equation}
As a result of \eqref{G1G2F11msigmag1g2maxnorm}, \eqref{g1g211sigmanormmuGL}, and \eqref{g1g2lips}, we obtain
\begin{equation}\label{Glip}
\left\|G_1-G_2\right\|_\infty
\leq L_G\left\|\lambda_1-\lambda_2\right\|_{L^1}.
\end{equation}
where
\[
L_G := \frac{1}{\sigma-1}F^{\frac{1}{\sigma-1}}T_{\max}^{\sigma}T_{\min}^{1-\sigma}(\Lambda-b)^{\frac{\sigma}{1-\sigma}} > 0.
\]

Secondly, we show that the operator $w$ is Lipschitz continuous. Let us denote $w[\lambda_i]$ by $w_i$ for $i=1,2$. It follows from \eqref{fp} that
\begin{equation}\label{w1w2musigmaFIJ}
\begin{aligned}
\left|w_1-w_2\right| 
&\leq \frac{\mu}{\sigma F}\int_\Omega\left|w_1(y)|\lambda_1(y)|G_1(y)^{\sigma-1}-w_2(y)|\lambda_2(y)|G_2(y)^{\sigma-1}\right|T(x,y)^{1-\sigma}dy\\
&\hspace{7mm}+\frac{\mu}{\sigma F}\int_\Omega\phi(y)\left|G_1(y)^{\sigma-1}-G_2(y)^{\sigma-1}\right|T(x,y)^{1-\sigma}dy\\
&= \frac{\mu}{\sigma F}\left(I+J\right),
\end{aligned}
\end{equation}
where 
\[
\begin{aligned}
&I = \int_\Omega\left|w_1(y)|\lambda_1(y)|G_1(y)^{\sigma-1}-w_2(y)|\lambda_2(y)|G_2(y)^{\sigma-1}\right|T(x,y)^{1-\sigma}dy,\\
&J = \int_\Omega\phi(y)\left|G_1(y)^{\sigma-1}-G_2(y)^{\sigma-1}\right|T(x,y)^{1-\sigma}dy.
\end{aligned}
\]
It follows from \eqref{Tbound}, \eqref{lambdabound}, \eqref{Gbound}, \eqref{wbound}, and \eqref{abslm12} that
\begin{align}
&I \leq \int_\Omega\left|w_1(y)\right|\left|\lambda_1(y)\right|\left|G_1(y)^{\sigma-1}-G_2(y)^{\sigma-1}\right|T(x,y)^{1-\sigma}dy\\
&\hspace{10mm} + \int_\Omega |w_1(y)|\left||\lambda_1(y)|-|\lambda_2(y)|\right|G_2(y)^{\sigma-1}T(x,y)^{1-\sigma}dy\\
&\hspace{15mm}+\int_\Omega|\lambda_2(y)||w_1(y)-w_2(y)|G_2(y)^{\sigma-1}T(x,y)^{1-\sigma}dy\\
&\hspace{3mm}\leq w_{U}(\Lambda+b)T_{\min}^{1-\sigma}\left\|G_1^{\sigma-1}-G_2^{\sigma-1}\right\|_{\infty} \label{G1sigm1G2sigm1wULambTmin}\\
&\hspace{25mm}+ w_{U}G_{U}^{\sigma-1}T_{\min}^{1-\sigma}\left\|\lambda_1-\lambda_2\right\|_{L^1}\\
&\hspace{35mm}+ (\Lambda+b)G_{U}^{\sigma-1}T_{\min}^{1-\sigma}\left\|w_1-w_2\right\|_{\infty}.
\end{align}
By \eqref{Gbound} and \eqref{LipFrechet}, we have
\begin{equation}\label{G1G2sigmamin1eqG1G2}
\left\|G_1^{\sigma-1}-G_2^{\sigma-1}\right\|_{\infty}\leq (\sigma-1)C_G^{\sigma-2}\left\|G_1-G_2\right\|_{\infty},
\end{equation}
where 
\begin{equation}
C_G := \left\{
\begin{aligned}
&G_{L},\hspace{3mm}\text{when $\sigma-2\leq 0$},\\
&G_{U},\hspace{3mm}\text{when $\sigma-2 > 0$}.
\end{aligned}
\right.
\end{equation}
By applying \eqref{Glip} and \eqref{G1G2sigmamin1eqG1G2} to \eqref{G1sigm1G2sigm1wULambTmin}, we have
\begin{equation}\label{Ilip}
\begin{aligned}
I &\hspace{3mm}\leq w_{U}(\Lambda+b)(\sigma-1)C_G^{\sigma-2}L_GT_{\min}^{1-\sigma}\left\|\lambda_1-\lambda_2\right\|_{\infty} \\
&\hspace{10mm}+ w_{U}G_{U}^{\sigma-1}T_{\min}^{1-\sigma}\left\|\lambda_1-\lambda_2\right\|_{L^1}\\
&\hspace{15mm}+ (\Lambda+b)G_{U}^{\sigma-1}T_{\min}^{1-\sigma}\left\|w_1-w_2\right\|_{\infty}.
\end{aligned}
\end{equation}
In the same manner, we obtain 
\begin{equation}\label{Jlip}
J \leq (\sigma-1)C_G^{\sigma-2}L_GT_{\min}^{1-\sigma}\Phi\left\|\lambda_1-\lambda_2\right\|_{L^1}.
\end{equation}
By applying \eqref{Ilip} and \eqref{Jlip} to \eqref{w1w2musigmaFIJ}, we have
\begin{equation}
\begin{aligned}
\left\|w_1-w_2\right\|_{\infty}
\leq &\frac{\mu}{\sigma F}\left\{w_{U}(\Lambda+b)(\sigma-1)C_G^{\sigma-2}L_GT_{\min}^{1-\sigma}\right.\\
&\hspace{10mm}\left.+ w_{U}G_{U}^{\sigma-1}T_{\min}^{1-\sigma}\right.\\
&\hspace{10mm}\left. + (\sigma-1)C_G^{\sigma-2}L_GT_{\min}^{1-\sigma}\Phi\right\}\left\|\lambda_1-\lambda_2\right\|_{L^1}\\
&\hspace{10mm} + \frac{\mu}{\sigma F}(\Lambda+b)G_{U}^{\sigma-1}T_{\min}^{1-\sigma}\left\|w_1-w_2\right\|_{\infty}.\label{w1w2normpre}
\end{aligned}
\end{equation}
Therefore, if
\begin{equation}\label{1minusmusigmaFLbGUsigm1Tmin1msigmagtzero}
1-\frac{\mu}{\sigma F}(\Lambda+b)G_{U}^{\sigma-1}T_{\min}^{1-\sigma}
 > 0
\end{equation}
holds, then we obtain from \eqref{w1w2normpre} that
\begin{equation}\label{wlip}
\left\|w_1-w_2\right\|_{\infty}\leq L_w \left\|\lambda_1-\lambda_2\right\|_{L^1},
\end{equation}
where
\begin{align}
L_w &:=
\frac{
\frac{\mu}{\sigma F}\left\{(\sigma-1)C_G^{\sigma-2}L_GT_{\min}^{1-\sigma}\left(w_U(\Lambda+b)+\Phi\right) + w_UG_U^{\sigma-1}T_{\min}^{1-\sigma}\right\}
}
{
1-\frac{\mu}{\sigma F}(\Lambda+b)G_U^{\sigma-1}T_{\min}^{1-\sigma}
}\\
&>0.
\end{align}
The condition \eqref{1minusmusigmaFLbGUsigm1Tmin1msigmagtzero} is equivalent to \eqref{sufficientconditionLb} because of \eqref{Gupper}.

Thirdly, we show that the operator $\omega$ is Lipschitz continuous. Let us denote $\omega[\lambda_i]$ by $\omega_i$ for $i=1,2$. By \eqref{opomega}, \eqref{Gbound}, \eqref{wbound}, and \eqref{LipFrechet}, we see that
\begin{align}
\left|\omega_1(x)-\omega_2(x)\right|
&\leq \left\|w_1\right\|_{\infty}\left\|G_1^{-\mu}-G_2^{-\mu}\right\|_{\infty}+\left\|w_1-w_2\right\|_{\infty}\left\|G_2^{-\mu}\right\|_{\infty} \\
&\leq \mu w_{U} G_{L}^{-\mu-1}\left\|G_1-G_2\right\|_{\infty}+{G_{L}}^{-\mu}\left\|w_1-w_2\right\|_{\infty}
\end{align}
Then, by applying \eqref{Glip} and \eqref{wlip}, we obtain
\begin{equation}\label{omegalip}
\left\|\omega_1-\omega_2\right\|_{\infty}\leq L_\omega \left\|\lambda_1-\lambda_2\right\|_{L^1},
\end{equation}
where 
\[
L_\omega := \mu w_{U}G_{L}^{-\mu-1}L_G + {G_{L}}^{-\mu}L_w > 0.
\]

Finally, we show that the operator $\Psi$ is Lipschitz continuous. Let us denote $\Psi[\lambda_i]$ and $\tilde{\omega}[\lambda_i]$ by $\Psi_i$ and $\tilde{\omega}_i$ for $i=1,2$, respectively. From \eqref{opaverageomega} and \eqref{opPsi}, we have
\begin{equation}\label{psi1psi2ololintolol}
\begin{aligned}
\left|\Psi_1(x)-\Psi_2(x)\right|
&\leq v\left|\omega_1(x)\lambda_1(x)-\omega_2(x)\lambda_2(x)\right|\\
&\hspace{5mm}+v\left|\tilde{\omega}_2\lambda_2(x)-\tilde{\omega}_1\lambda_1(x)\right|.
\end{aligned}
\end{equation}
At the first term of the right-hand side of \eqref{psi1psi2ololintolol}, it is easy to see that
\begin{align}
\left|\omega_1(x)\lambda_1(x)-\omega_2(x)\lambda_2(x)\right|
&=\left|\omega_1(x)\lambda_1(x)-\omega_2(x)\lambda_1\right.\\
&\left.\hspace{7mm}+\omega_2(x)\lambda_1(x)-\omega_2(x)\lambda_2(x)\right| \\
&\leq \left|\omega_1(x)-\omega_2(x)\right|\left|\lambda_1(x)\right|\\
&\hspace{7mm}+ \left|\omega_2(x)\right| \left|\lambda_1(x)-\lambda_2(x)\right|.\label{om1om2lmplusom2lm1lm2}
\end{align}
Integrating the both sides of \eqref{om1om2lmplusom2lm1lm2} over $\Omega$ and applying \eqref{lambdabound}, \eqref{omegabound}, and \eqref{omegalip}, we obtain
\begin{equation}\label{firstterm}
\begin{aligned}
\left\|\omega_1\lambda_1-\omega_2\lambda_2\right\|_{L^1}
\leq \left\{(\Lambda+b)L_\omega + \omega_U\right\}\left\|\lambda_1-\lambda_2\right\|_{L^1}.
\end{aligned}
\end{equation}
At the second term of the right-hand side of \eqref{psi1psi2ololintolol}, it is easy to see that
\begin{align}
\left|\tilde{\omega}_2\lambda_2(x)-\tilde{\omega}_1\lambda_1(x)\right|
&=\frac{1}{\Lambda}\left|\int_\Omega\omega_2(y)\lambda_2(y)dy\lambda_2(x)
-\int_\Omega\omega_1(y)\lambda_1(y)dy\lambda_1(x)\right|\\
&\leq \frac{1}{\Lambda}\left|\int_\Omega\omega_2(y)\lambda_2(y)dy\right|\left|\lambda_1(x)-\lambda_2(x)\right| \\
&\hspace{10mm}+ \frac{\left|\lambda_1(x)\right|}{\Lambda}\left|\int_\Omega\omega_2(y)\lambda_2(y)dy-\int_\Omega\omega_1(y)\lambda_1(y)dy\right| \\
&\leq \frac{1}{\Lambda}\left|\int_\Omega\omega_2(y)\lambda_2(y)dy\right|\left|\lambda_1(x)-\lambda_2(x)\right| \\
&\hspace{10mm}+ \frac{\left|\lambda_1(x)\right|}{\Lambda}\left|\int_\Omega\omega_2(y)\left(\lambda_2(y)-\lambda_1(y)\right)dy\right| \\
&\hspace{10mm}+ \frac{\left|\lambda_1(x)\right|}{\Lambda}\left|\int_\Omega\lambda_1(y)\left(\omega_2(y)-\omega_1(y)\right)dy\right|
\label{omtil2lm2-ometil1lm1}
\end{align}
Integrating the both sides of \eqref{omtil2lm2-ometil1lm1} and applying \eqref{lambdabound}, \eqref{omegabound}, and \eqref{omegalip} yields
\begin{align}
\left\|\tilde{\omega}_2\lambda_2-\tilde{\omega}_1\lambda_1\right\|_{L^1}
&= \frac{1}{\Lambda}\left\|\int_\Omega\omega_2(y)\lambda_2(y)dy\lambda_2
-\int_\Omega\omega_1(y)\lambda_1(y)dy\lambda_1\right\|_{L^1}\\
&\leq \frac{(\Lambda+b)(2\omega_U+(\Lambda+b)L_\omega)}{\Lambda}\left\|\lambda_1-\lambda_2\right\|_{L^1}.\label{secondterm}
\end{align}
By applying \eqref{firstterm} and \eqref{secondterm} to \eqref{psi1psi2ololintolol}, we obtain
\begin{equation}\label{Psilip}
\begin{aligned}
\left\|\Psi_1-\Psi_2\right\|_{L^1}
\leq L\left\|\lambda_1-\lambda_2\right\|_{L^1},
\end{aligned}
\end{equation}
where
\[
L :=(\Lambda+b)L_\omega+\omega_U+\frac{(\Lambda+b)(2\omega_U+(\Lambda+b)L_\omega)}{\Lambda} > 0.
\]
Note that $L$ does not depend on $\lambda_0$.
\qed

\subsection{Dependence of $\tau_k^*$ on frequency}\label{subsec:deptau*k}
The following theorem is essentially the same as \citet[Theorem 6]{Ohtake2025agriculture}.
\begin{theo}\label{th:tau*k}
Assume that \eqref{nbh} holds. For any integer $k\neq 0$, the critical point satisfies that $\tau_{|k|}^*<\tau_{|k|+2}^*$ and $\tau_{-|k|}^*<\tau_{-|k|-2}^*$.
\end{theo}
\begin{proof}
Suppose now that $Z_{|k|}=Z^*$ for a given $k$. When $|k|$ changes to $|k|+2$, it follows from $(-1)^{|k|}=(-1)^{|k|+2}$ in \eqref{Zk} that \footnote{This was pointed out by an anonymous reviewer of \citet{Ohtake2025agriculture} in the proof of its Theorem 6.}
\begin{equation}\label{}
Z_{|k|+2}<Z_{|k|}.
\end{equation}
Now to make $Z_{|k|+2}=Z^*$ hold again at $|k|+2$, $Z_{|k|}$ must be increased, i.e., $\tau$ must be increased. Thus, $\tau_{|k|}^*<\tau_{|k|+2}^*$. Similarly, $\tau_{-|k|}^*<\tau_{-|k|-2}^*$.
\end{proof}

\subsection{Comparison to other models}\label{subsec:comparison-models}
Let us refer to the core-periphery model proposed by \citet{Krug91} as the CP model, and the model based on the quasi-linear-log utility function proposed by \citet{Pfl} as the QLLU model. The CP model, the QLLU model, and the FE model are commonly used NEG models assuming the iceberg transport costs. In this section, we compare the FE model with these two models in terms of the stability of the homogeneous stationary solution, numerical asymptotic behavior, and the stability of non-homogeneous stationary solutions. Much of the content here consists of summaries of existing research (\citet{OhtakeYagi_Asym}, \citet{OhtakeYagi_Stab}, and \citet{Ohtake2023cont}).\footnote{Some of these comparisons are also presented in \citet[Section 6]{Ohtake2025pattern}. In some cases, previous results are recalculated to ensure consistency with the framework of this paper. The stability analysis of spiky stationary solutions for the QLLU model is conducted for the first time.} In the following, the notation is the same as in the FE model \eqref{1} on a circle $S$.

\subsubsection{Models}
\paragraph{CP model:}
The CP model under the racetrack setteing is the following system of integral and differential equations.

\begin{equation}\label{CP}
\left\{
\begin{aligned}
&Y(t,x) = \mu w(t,x)\lambda(t,x) + (1-\mu)\phi(x),\\
&G(t,x) = \left[\int_S \lambda(t,y)w(t,y)^{1-\sigma}T(x,y)^{1-\sigma}dy\right]^{\frac{1}{1-\sigma}},\\ 
&w(t,x) = \left[\int_S Y(t,y)G(t,y)^{\sigma-1}T(x,y)^{1-\sigma}dy\right]^{\frac{1}{\sigma}},\\
&\omega(t,x) = w(t,x)G(t,x)^{-\mu},\\
&\frac{\partial}{\partial t}\lambda(t,x) = v\left[\omega(t,x) - \int_{S} \omega(t,y)\lambda(t,y)dy\right]\lambda(t,x).
\end{aligned}\right.
\end{equation}
Note that $\int_S\lambda(t,x)dx\equiv 1~(\forall t\geq 0)$ and $\int_S\phi(x)dx\equiv 1$, respectively.\footnote{This is because, under the unit system adopted here, $\lambda$ and $\phi$ are not measures of population density but rather functions representing each region's share in the total population. However, this would not essential to the nature of the model.}

\paragraph{QLLU model:}
The QLLU model under the racetrack setteing is the following system of integral and differential equations. 
\begin{equation}\label{QLLU}
\left\{
\begin{aligned}
&G(t, x) = \left[\frac{1}{F}\int_{S} \lambda(t, y)T(x, y)^{1-\sigma}dy\right]^{\frac{1}{1-\sigma}},\\
&w(t, x)=\frac{\mu}{\sigma F}\int_{S} \left(\lambda(t,y)+\phi(y)\right)G(t, y)^{\sigma-1}T(x, y)^{1-\sigma}dy,\\ 
&\omega(t, x) = w(t, x)-\mu\ln G(t, x),\\
&\frac{\partial \lambda(t,x)}{\partial t} = v\left[\omega(t,x) - \frac{1}{\Lambda}\int_{S} \omega(t,y)\lambda(t,y)dy\right]\lambda(t,x).
\end{aligned}\right.
\end{equation}

\subsubsection{Stability of homogeneous stationary solutions}
\paragraph{CP model:}
Assume that the population density of immobile workers is homogeneous:
\begin{equation}\label{CPphibar}
\phi\equiv \ol{\phi} = \frac{1}{2\pi\rho}.
\end{equation}
Then, the CP model \eqref{CP} has the homogeneous stationary solution 
\begin{align}
&\lambda \equiv \bar{\lambda} = \frac{1}{2\pi\rho},\label{CPlambdabar}\\
&w \equiv \bar{w} = 1,\label{CPwbar}\\
&G \equiv \bar{G} = \left[\frac{2\bar{\lambda}\left(1-e^{-\alpha\rho\pi}\right)}{\alpha}\right]^{\frac{1}{1-\sigma}},\label{CPGbar}\\
&\omega \equiv \bar{\omega} = \bar{w}\bar{G}^{-\mu}.\label{CPomegabar}
\end{align}
As in Subsection \ref{subsec:contstabilityanalysis}, by linearizing the equations of \eqref{CP} around the homogeneous states \eqref{CPphibar}, \eqref{CPlambdabar}, \eqref{CPwbar}, \eqref{CPGbar}, and \eqref{CPomegabar}, we can obtain the eigenvalue $\Gamma_k$ ($=\Gamma_{-k}$) of the $k$-th mode as follows.
\begin{equation}
\Gamma_k = v\bar{G}^{-\mu}
\left[
(1-\mu Z_k)\frac{\frac{\mu}{\sigma}Z_k-\frac{1}{\sigma}Z_k^2}{1-\frac{\mu}{\sigma}Z_k-\frac{\sigma-1}{\sigma}Z_k^2}
+\frac{\mu Z_k}{\sigma-1}
\right].
\end{equation}
For this eigenvalue, critical curves can be plotted (Fig.~\ref{fig:CPccurves}), similar to the case of the FE model (Fig.~\ref{fig:ccurves}). For each value of frequency $k$, mode $k$ is unstable in the lower-left area of this curve and stable in the upper-right area. As with the FE model, we can see that the critical curve shifts outward as the absolute value of the frequency increases.
\begin{figure}[H]
\centering
\includegraphics[width=0.6\columnwidth]{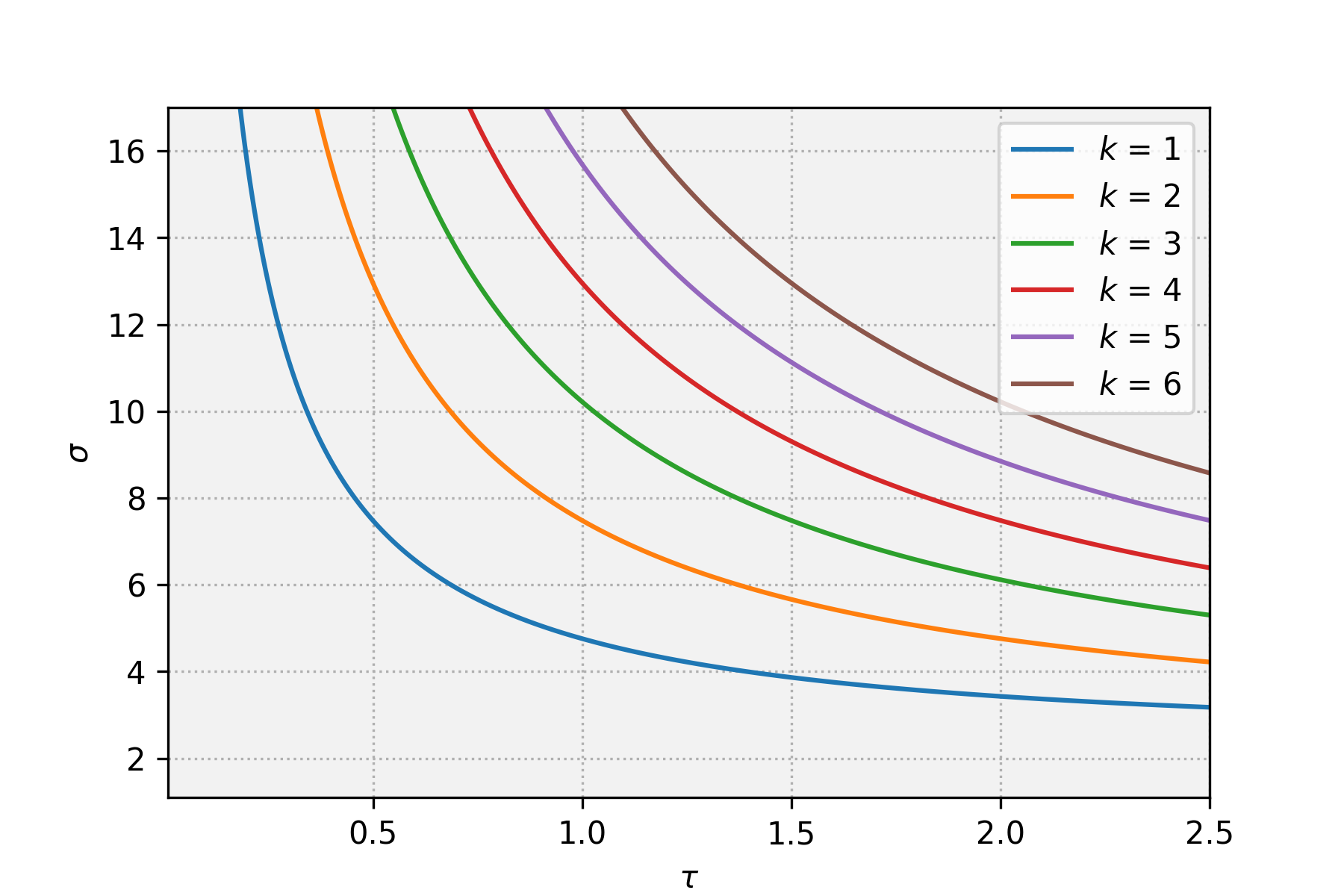}
\caption{Critical curves of the CP model for various frequencies}
\label{fig:CPccurves}
\end{figure}

\paragraph{QLLU model:}
Assume that the population density of immobile workers is homogeneous:
\begin{equation}\label{QLLUphibar}
\phi\equiv \ol{\phi} = \frac{\Phi}{2\pi\rho}.
\end{equation}
Then, the QLLU model \eqref{QLLU} has the homogeneous stationary solution 
\begin{align}
&\lambda \equiv \bar{\lambda} = \frac{\Lambda}{2\pi\rho},\label{QLLUlambdabar}\\
&w \equiv \bar{w} = \frac{\mu\left(\bar{\phi}+\bar{\lambda}\right)}{\sigma\bar{\lambda}},\label{QLLUwbar}\\
&G \equiv \bar{G} = \left[\frac{2\bar{\lambda}\left(1-e^{-\alpha\rho\pi}\right)}{F\alpha}\right]^{\frac{1}{1-\sigma}},\label{QLLUGbar}\\
&\omega \equiv \bar{\omega} = \bar{w}-\mu\ln\bar{G}.,\label{QLLUomegabar}
\end{align}
As in Subsection \ref{subsec:contstabilityanalysis}, by linearizing the equations of \eqref{QLLU} around the homogeneous states \eqref{QLLUphibar}, \eqref{QLLUlambdabar}, \eqref{QLLUwbar},\eqref{QLLUGbar}, and \eqref{QLLUomegabar}, we can obtain the eigenvalue $\Gamma_k$ ($=\Gamma_{-k}$) of the $k$-th mode as follows.
\begin{equation}
\Gamma_k = v\frac{\mu}{\sigma}\left[-\frac{\bar{\phi}+\bar{\lambda}}{\bar{\lambda}}Z_k^2 + \frac{2\sigma-1}{\sigma-1}Z_k\right].
\end{equation}
For this eigenvalue, critical curves can be plotted (Fig.~\ref{fig:QLLUccurves}), similar to the case of the FE model (Fig.~\ref{fig:ccurves}). For each value of frequency $k$, mode $k$ is unstable in the lower-left area of this curve and stable in the upper-right area. As with the FE model, we can see that the critical curve shifts outward as the absolute value of the frequency increases.
\begin{figure}[H]
\centering
\includegraphics[width=0.6\columnwidth]{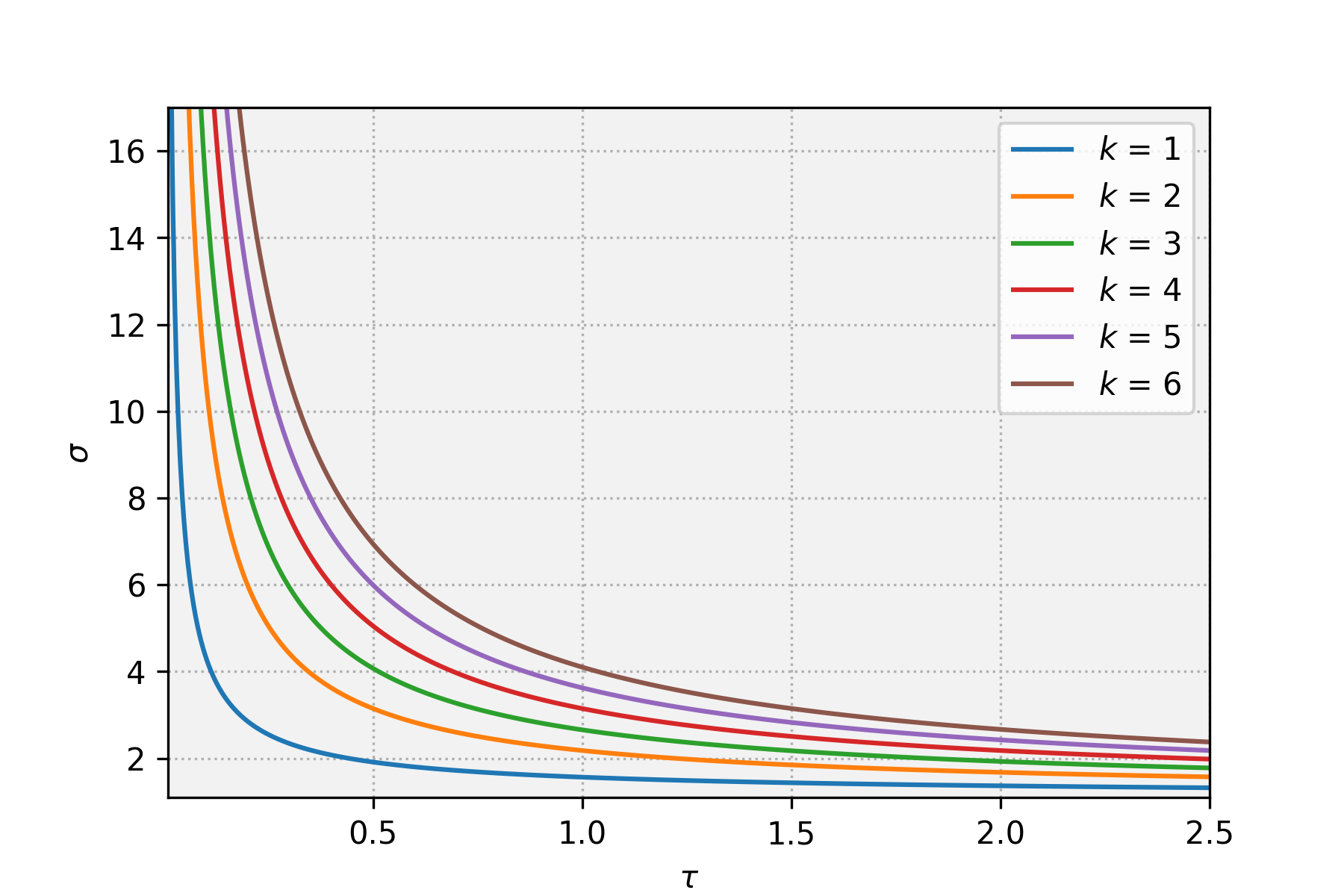}
\caption{Critical curves of the QLLU model for various frequencies}
\label{fig:QLLUccurves}
\end{figure}

\subsubsection{Numerical simulation for time evolution}
\paragraph{CP model:}
Under the same settings as in Section \ref{sec:numerical}, the results of numerical simulations of the time evolution of the solution to the CP model are shown in Figs.~\ref{fig:CPt0p5}–\ref{fig:CPt2p1} below. These results represent the largest number of spikes observed in five simulations. In this simulation, $\mu$ is fixed to $0.6$ and $\sigma$ is fixed to $5.0$, while $\tau$ is varied through the values $0.5$, $0.7$, $1.1$, $1.7$, $1.8$, and $2.1$. Similar to the case of the FE model, we can see that as transport costs rise, the largest number of spikes in the stationary state increases.

\begin{figure}[H]
 \begin{subfigure}{0.5\columnwidth}
  \centering
  \includegraphics[width=\columnwidth]{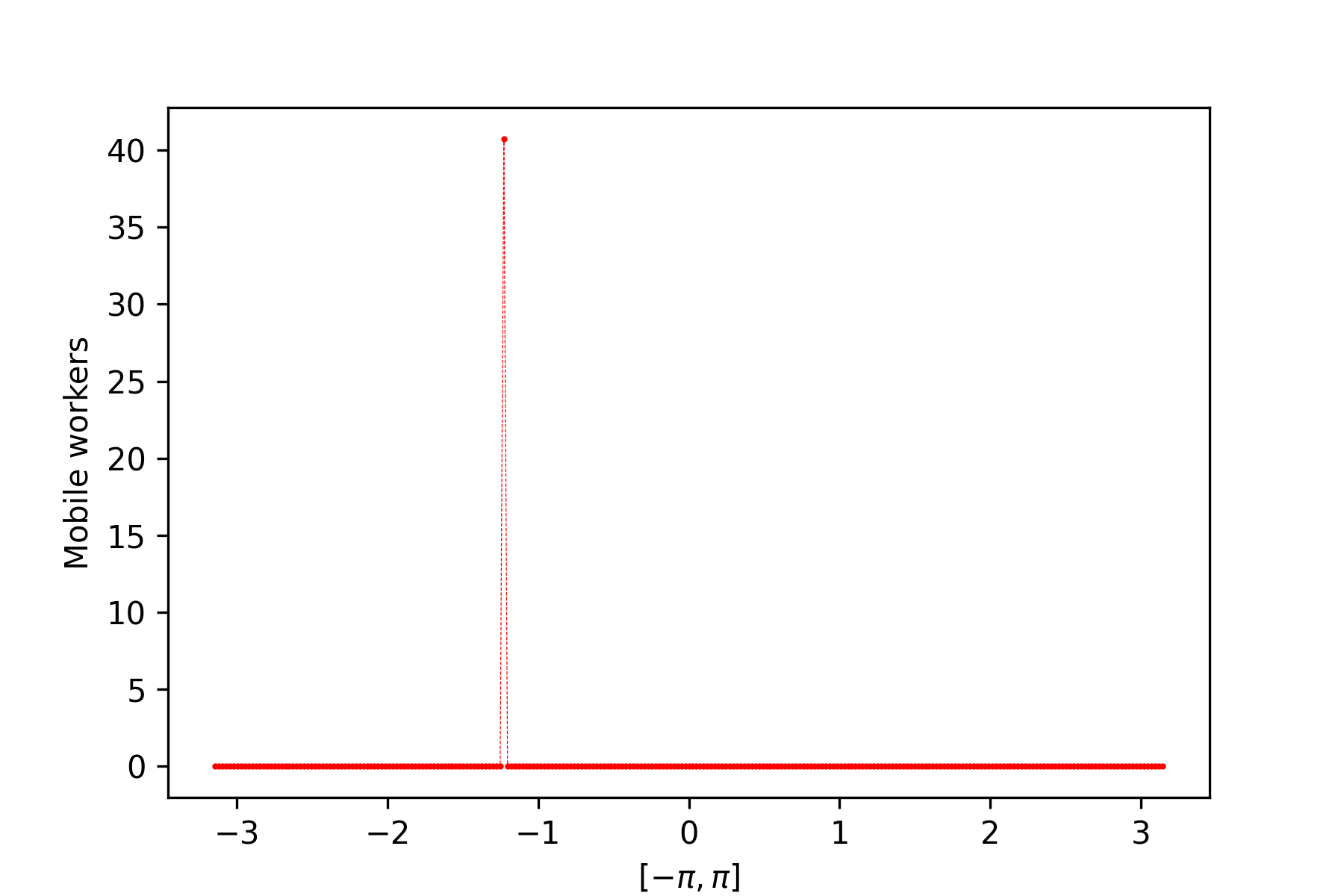}
  \caption{Mobile population}
 \end{subfigure}
 \begin{subfigure}{0.5\columnwidth}
  \centering
  \includegraphics[width=\columnwidth]{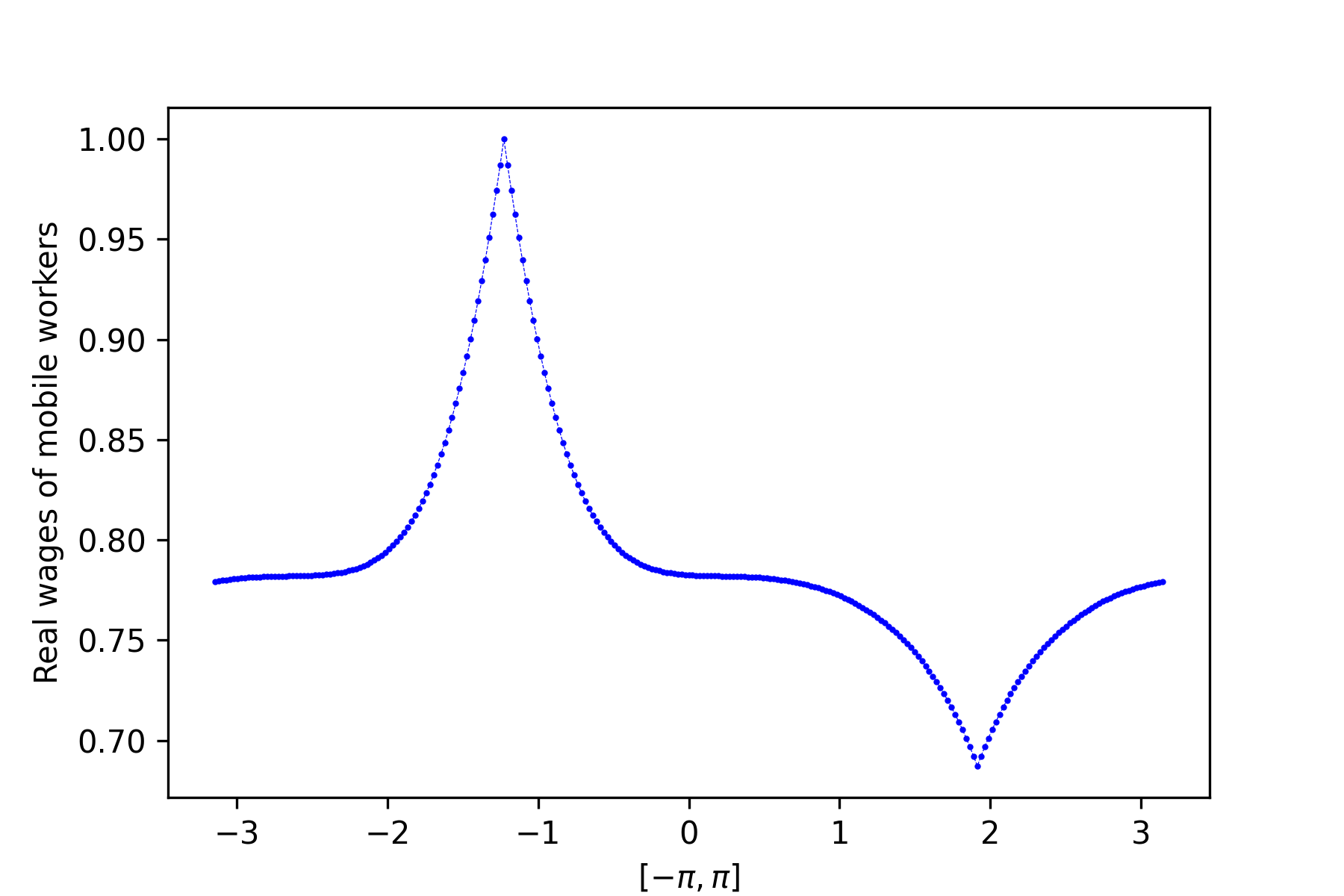}
  \caption{Real wage}
 \end{subfigure}\\
 \caption{Stationary solution for $(\sigma,\tau)=(5.0, 0.5)$}
 \label{fig:CPt0p5}
\end{figure}

\begin{figure}[H]
 \begin{subfigure}{0.5\columnwidth}
  \centering
  \includegraphics[width=\columnwidth]{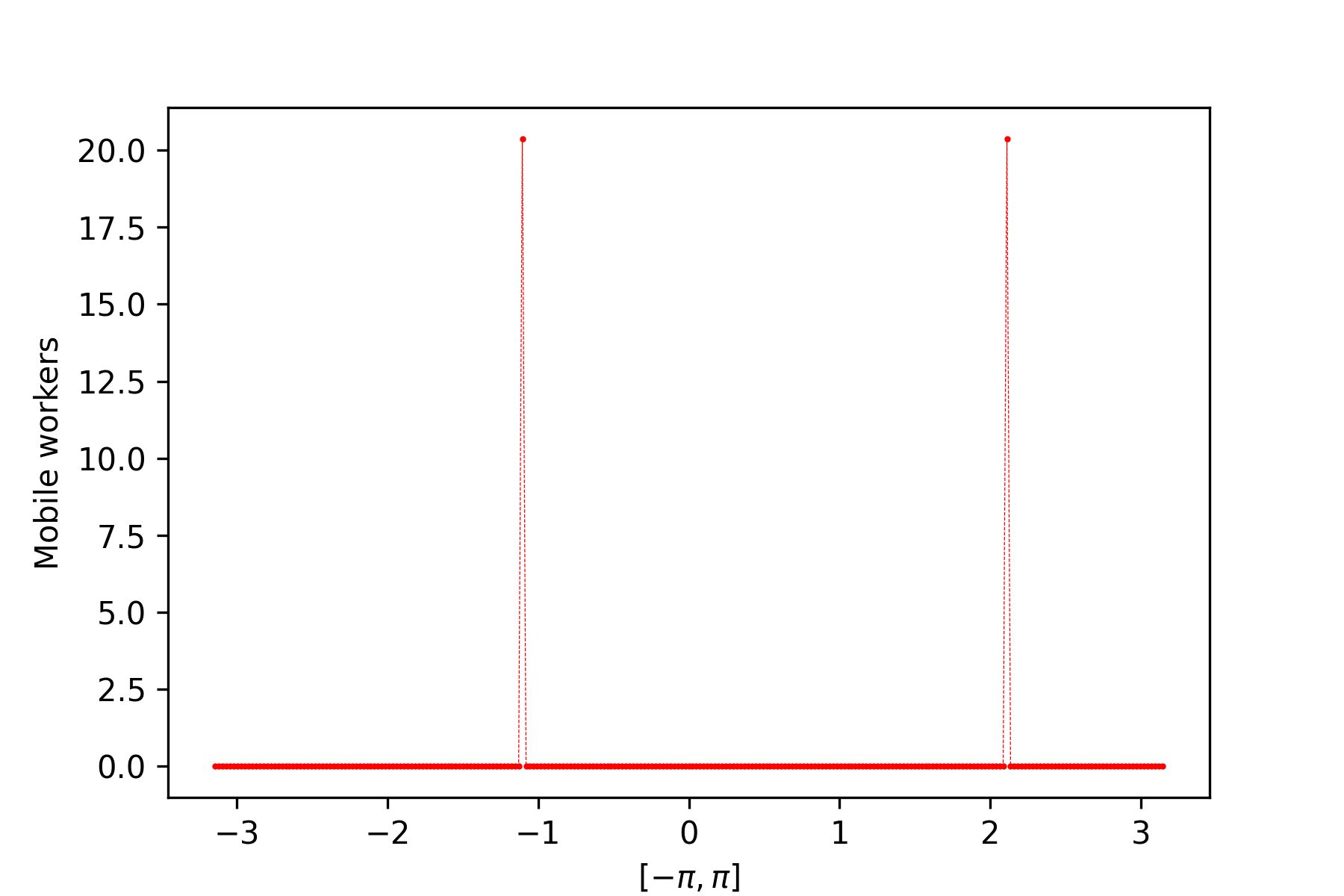}
  \caption{Mobile population}
 \end{subfigure}
 \begin{subfigure}{0.5\columnwidth}
  \centering
  \includegraphics[width=\columnwidth]{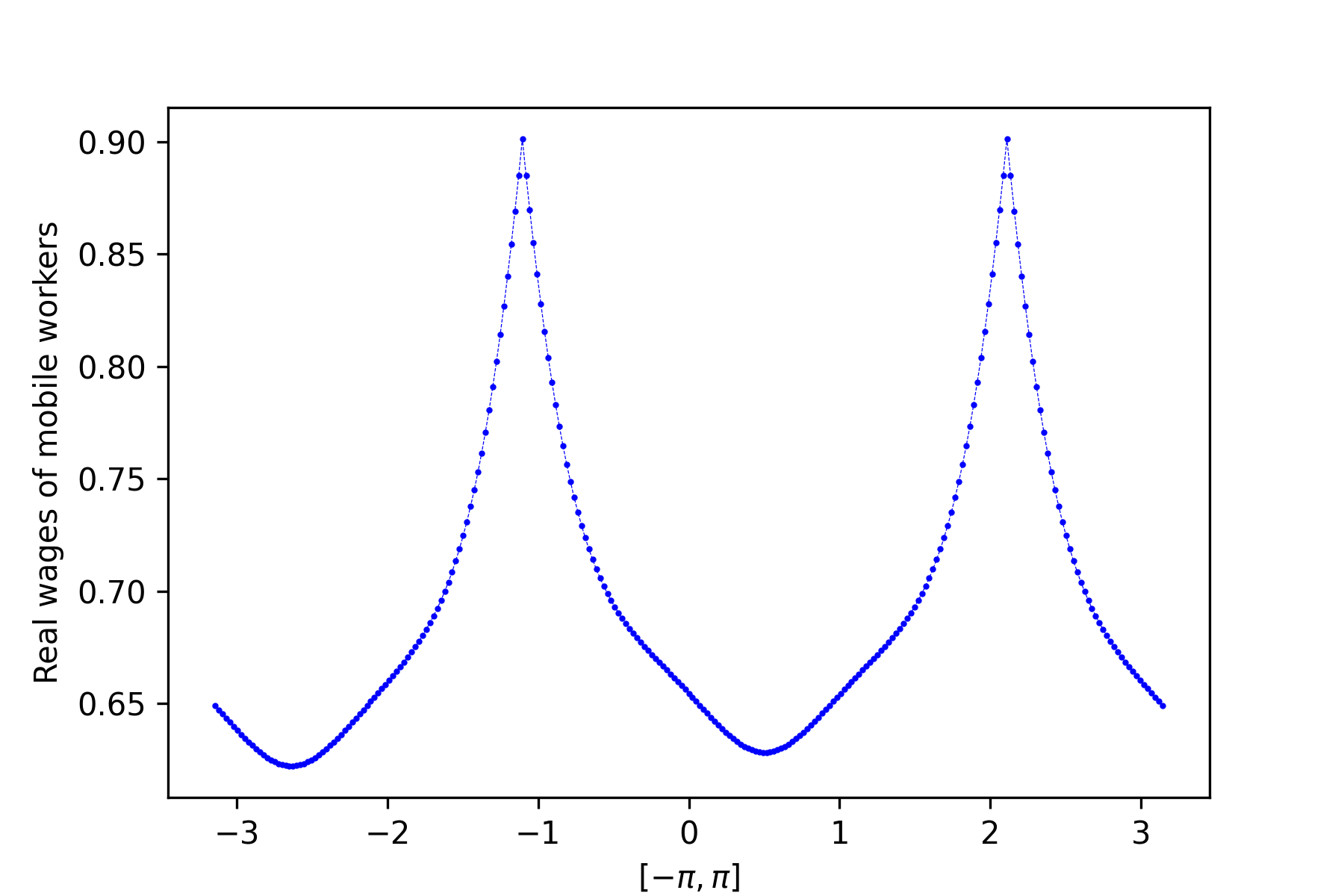}
  \caption{Real wage}
 \end{subfigure}\\
 \caption{Stationary solution for $(\sigma,\tau)=(5.0, 0.7)$}
 \label{fig:t0p7}
\end{figure}

\begin{figure}[H]
 \begin{subfigure}{0.5\columnwidth}
  \centering
  \includegraphics[width=\columnwidth]{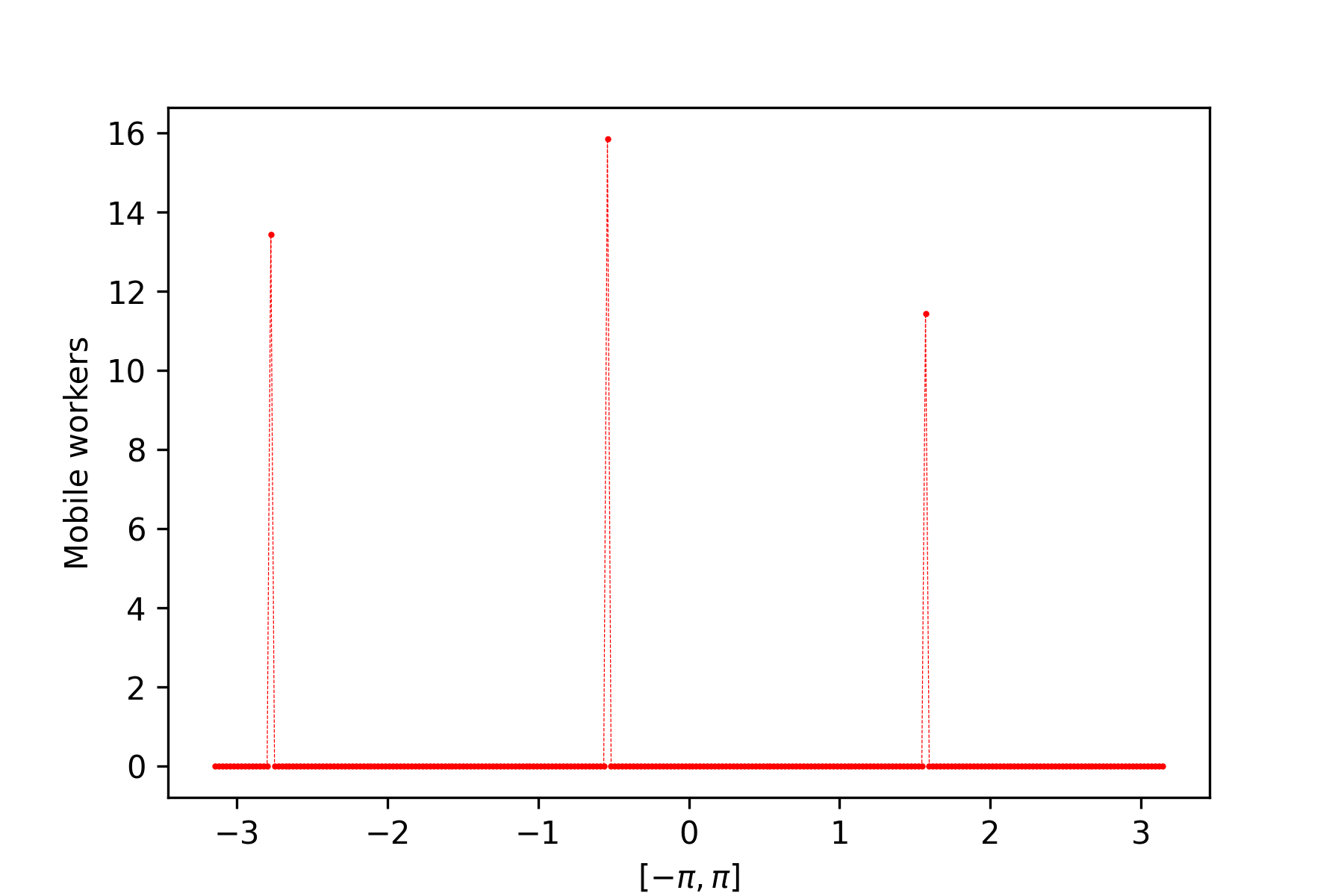}
  \caption{Mobile population}
 \end{subfigure}
 \begin{subfigure}{0.5\columnwidth}
  \centering
  \includegraphics[width=\columnwidth]{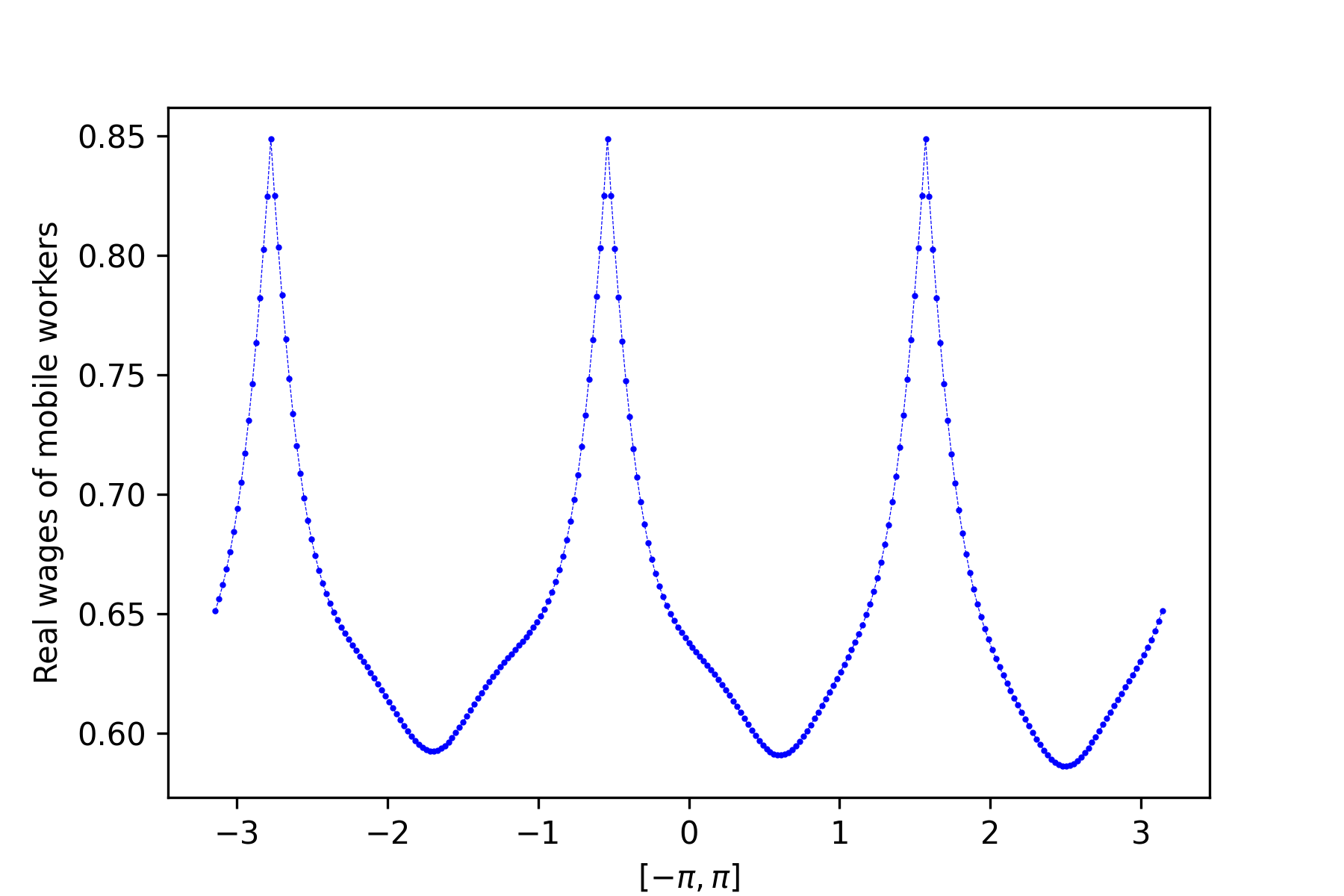}
  \caption{Real wage}
 \end{subfigure}\\
 \caption{Stationary solution for $(\sigma,\tau)=(5.0, 1.1)$}
 \label{fig:t1p1}
\end{figure}

\begin{figure}[H]
 \begin{subfigure}{0.5\columnwidth}
  \centering
  \includegraphics[width=\columnwidth]{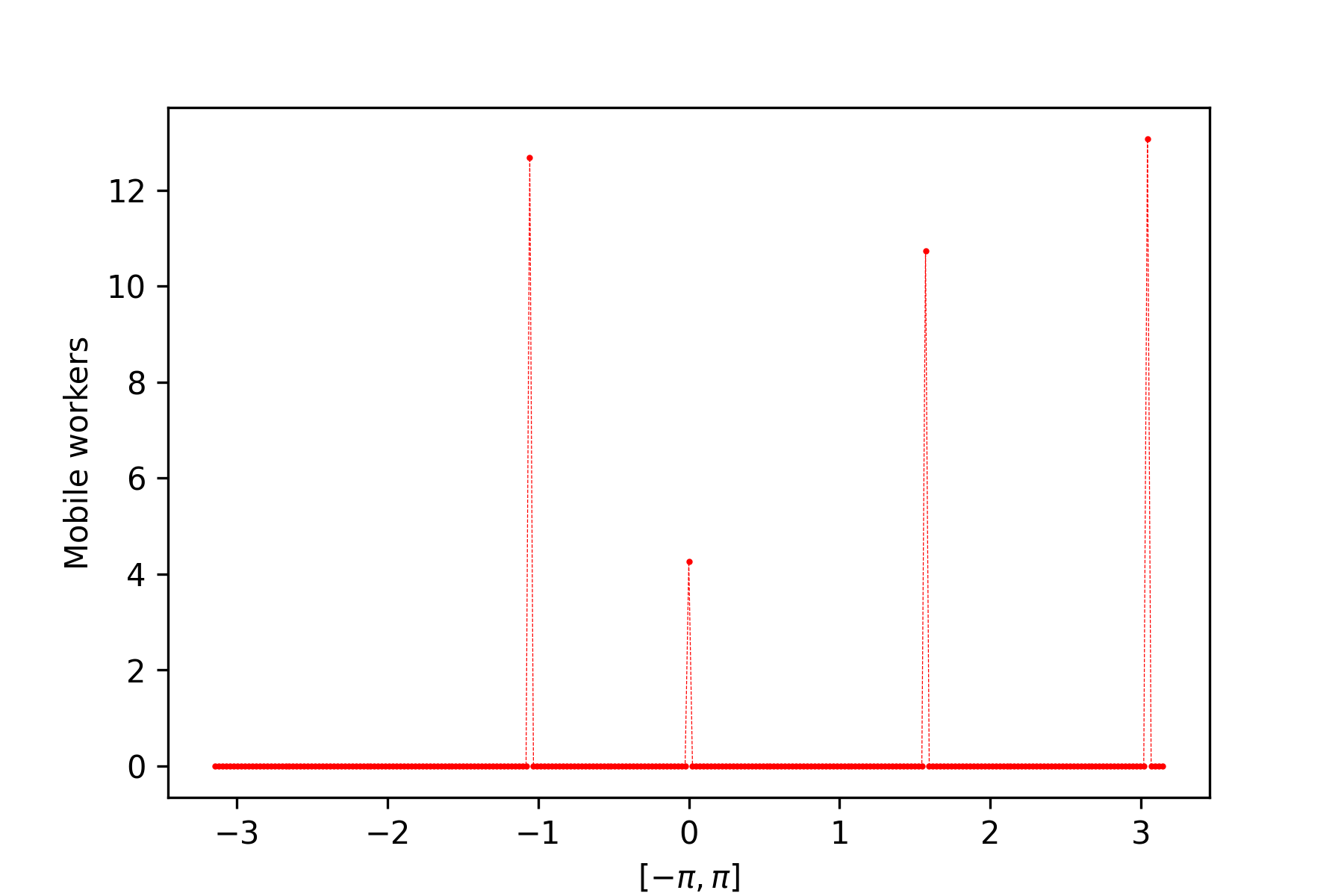}
  \caption{Mobile population}
 \end{subfigure}
 \begin{subfigure}{0.5\columnwidth}
  \centering
  \includegraphics[width=\columnwidth]{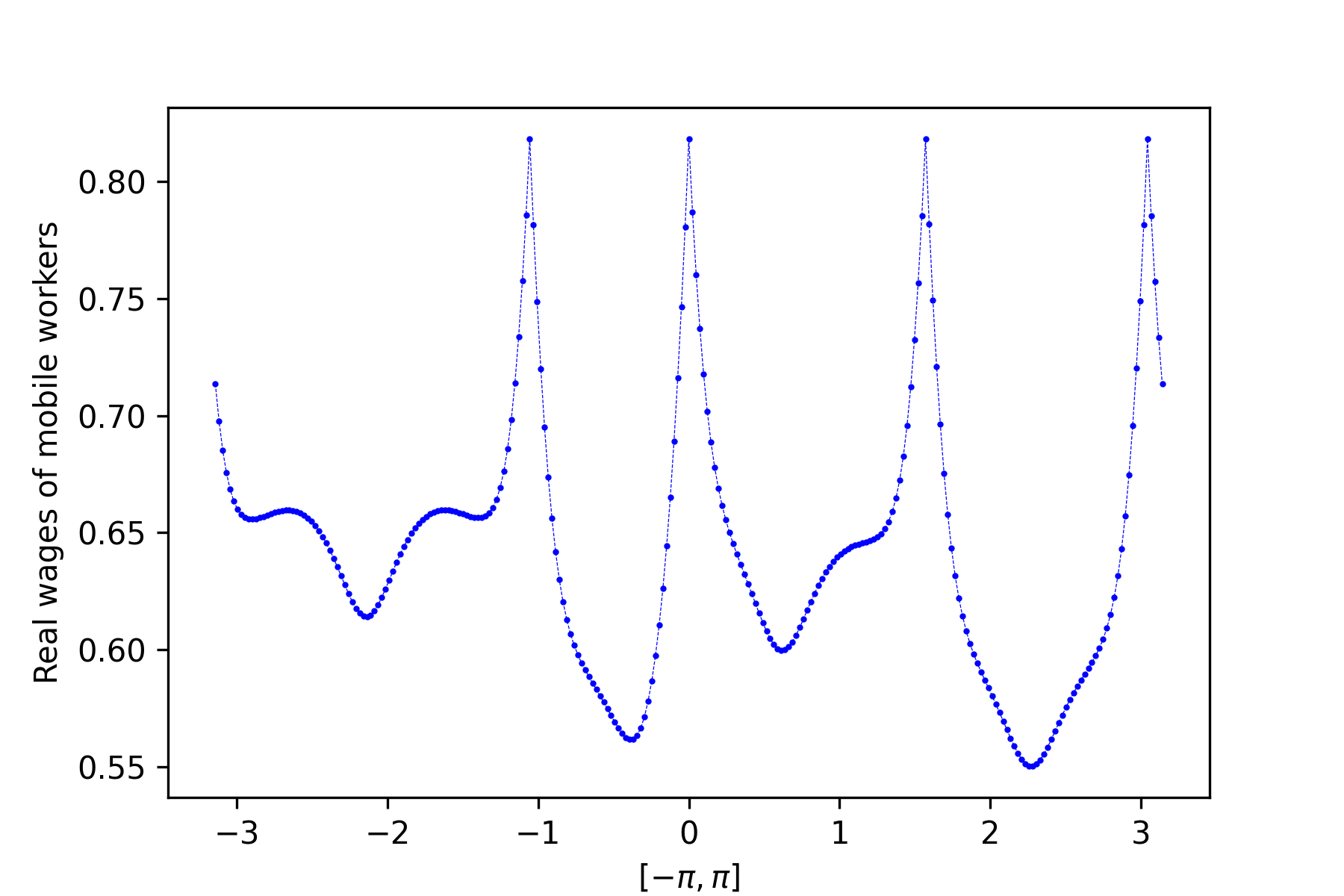}
  \caption{Real wage}
 \end{subfigure}\\
 \caption{Stationary solution for $(\sigma,\tau)=(5.0, 1.7)$}
 \label{fig:t1p7}
\end{figure}

\begin{figure}[H]
 \begin{subfigure}{0.5\columnwidth}
  \centering
  \includegraphics[width=\columnwidth]{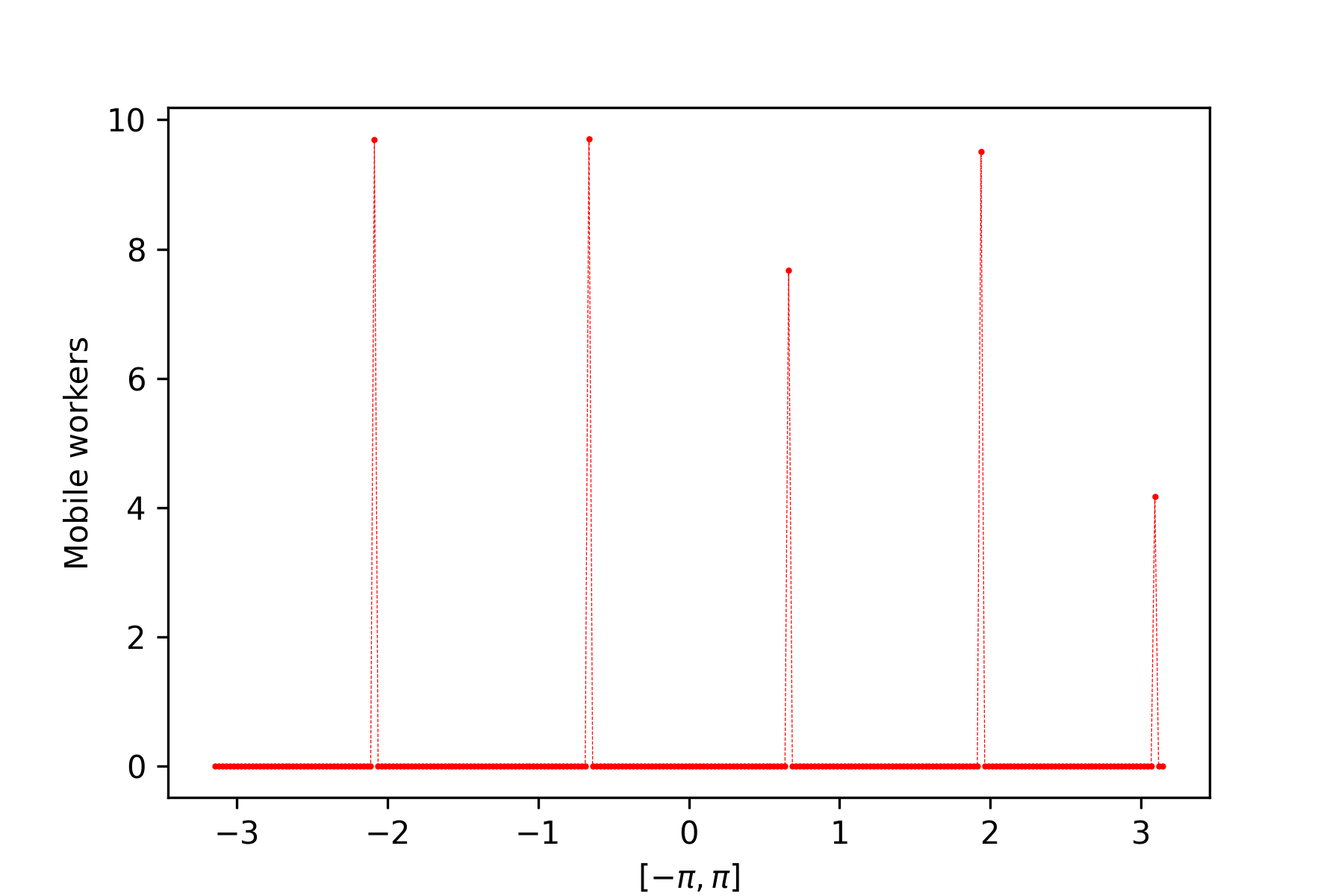}
  \caption{Mobile population}
 \end{subfigure}
 \begin{subfigure}{0.5\columnwidth}
  \centering
  \includegraphics[width=\columnwidth]{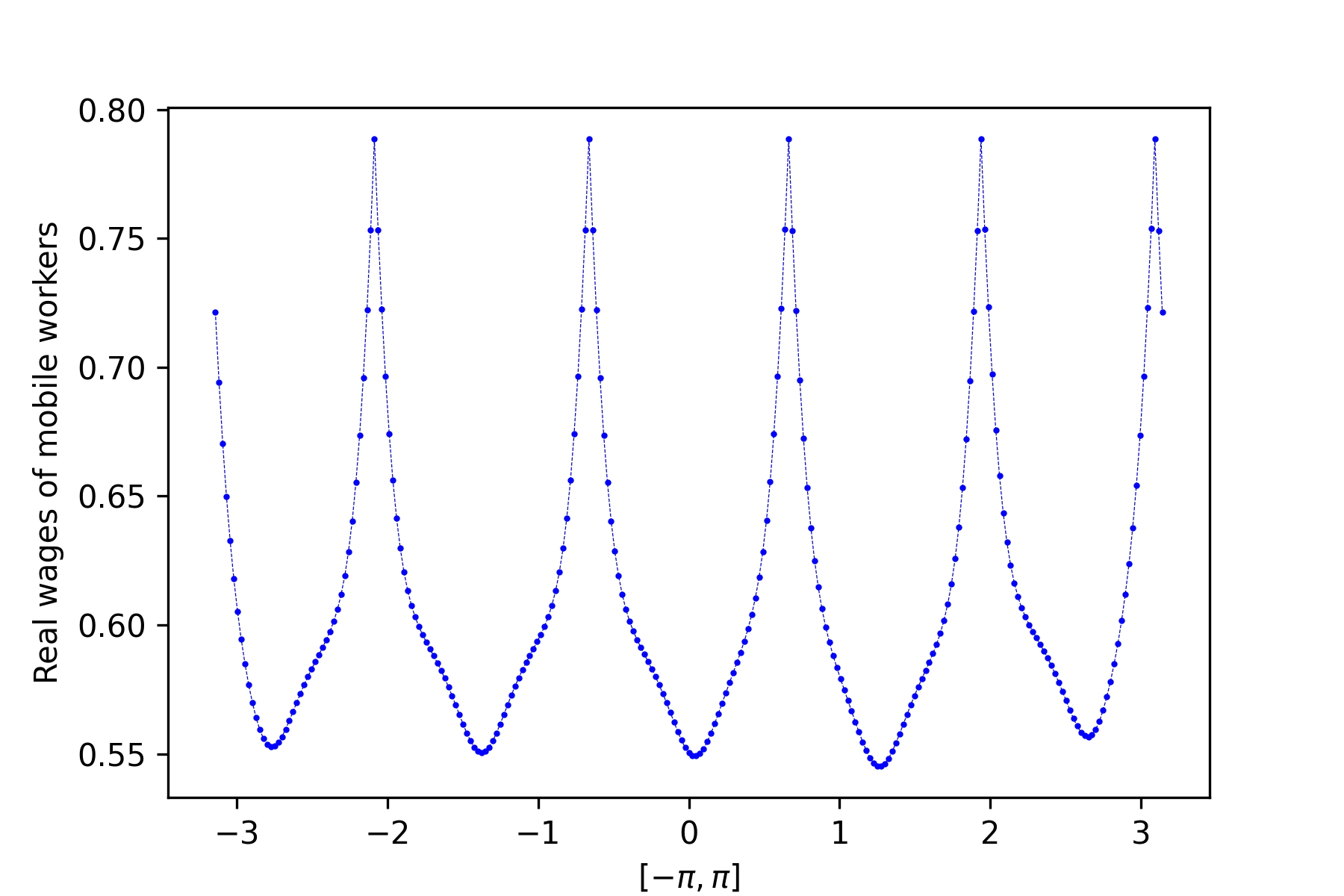}
  \caption{Real wage}
 \end{subfigure}\\
 \caption{Stationary solution for $(\sigma,\tau)=(5.0, 1.8)$}
 \label{fig:t1p8}
\end{figure}

\begin{figure}[H]
 \begin{subfigure}{0.5\columnwidth}
  \centering
  \includegraphics[width=\columnwidth]{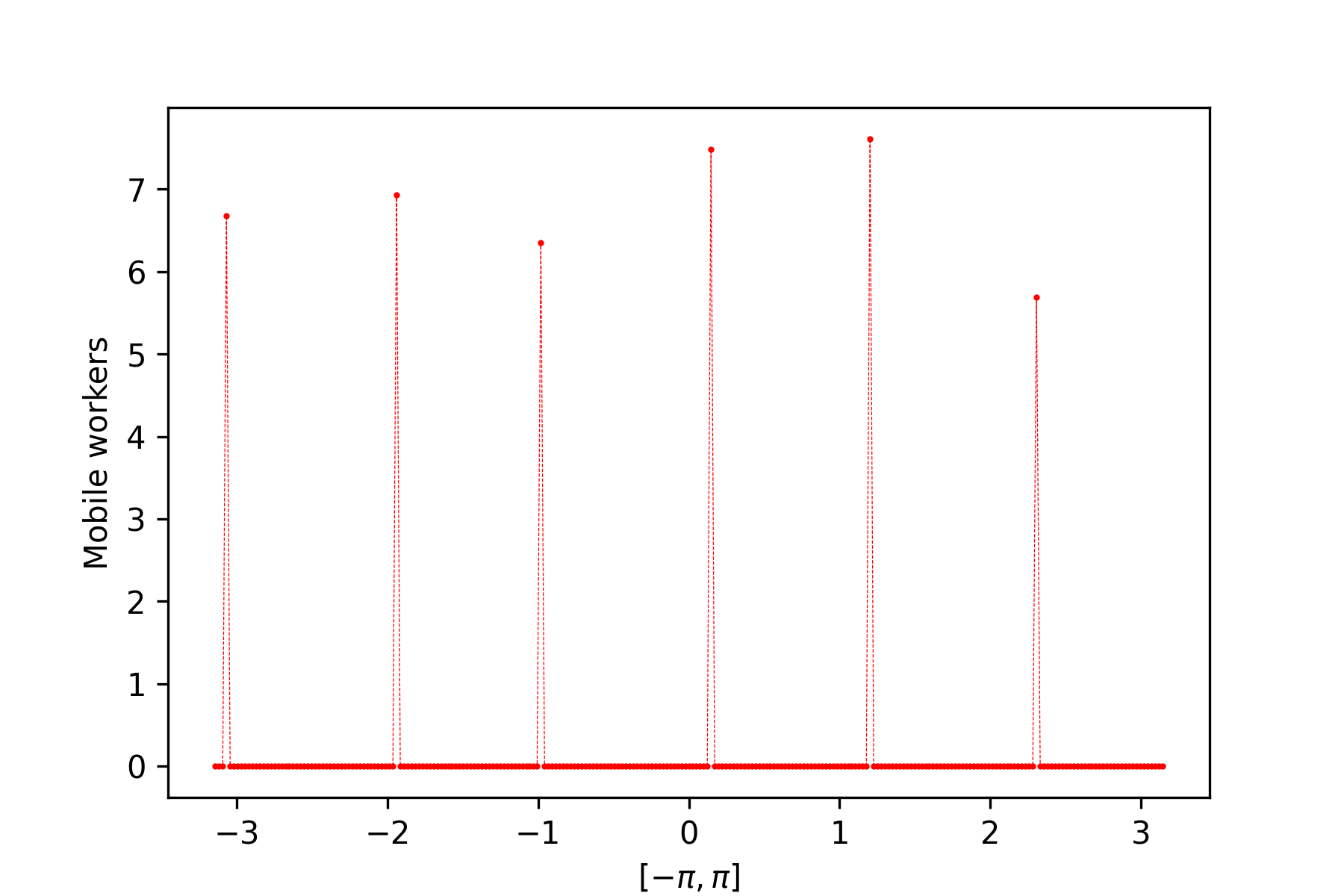}
  \caption{Mobile population}
 \end{subfigure}
 \begin{subfigure}{0.5\columnwidth}
  \centering
  \includegraphics[width=\columnwidth]{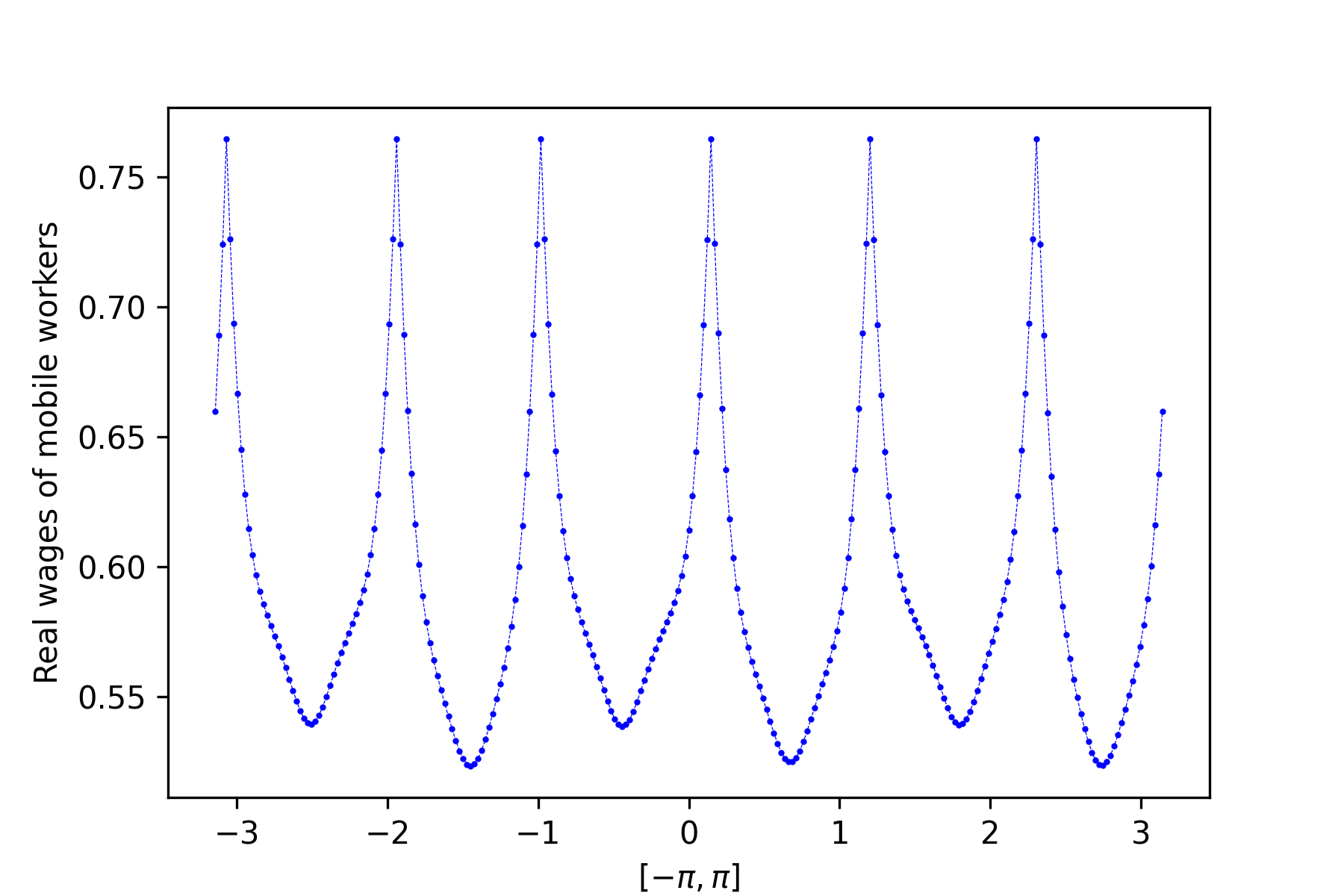}
  \caption{Real wage}
 \end{subfigure}\\
 \caption{Stationary solution for $(\sigma,\tau)=(5.0, 2.1)$}
 \label{fig:CPt2p1}
\end{figure}

\paragraph{QLLU model:}
Under the same settings as in Section \ref{sec:numerical}, the results of numerical simulations of the time evolution of the solution to the QLLU model are shown in Figs.~\ref{fig:QLLUt0p05}–\ref{fig:QLLUt0p42} below. These results represent the largest number of spikes observed in five simulations. In this simulation, $\mu$ is fixed to $0.6$ and $\sigma$ is fixed to $5.0$, while $\tau$ is varied through the values $0.05$, $0.1$, $0.2$, $0.3$, $0.35$, and $0.42$. Similar to the case of the FE model, we can see that as transport costs rise, the largest number of spikes in the stationary state increases.

\begin{figure}[H]
 \begin{subfigure}{0.5\columnwidth}
  \centering
  \includegraphics[width=\columnwidth]{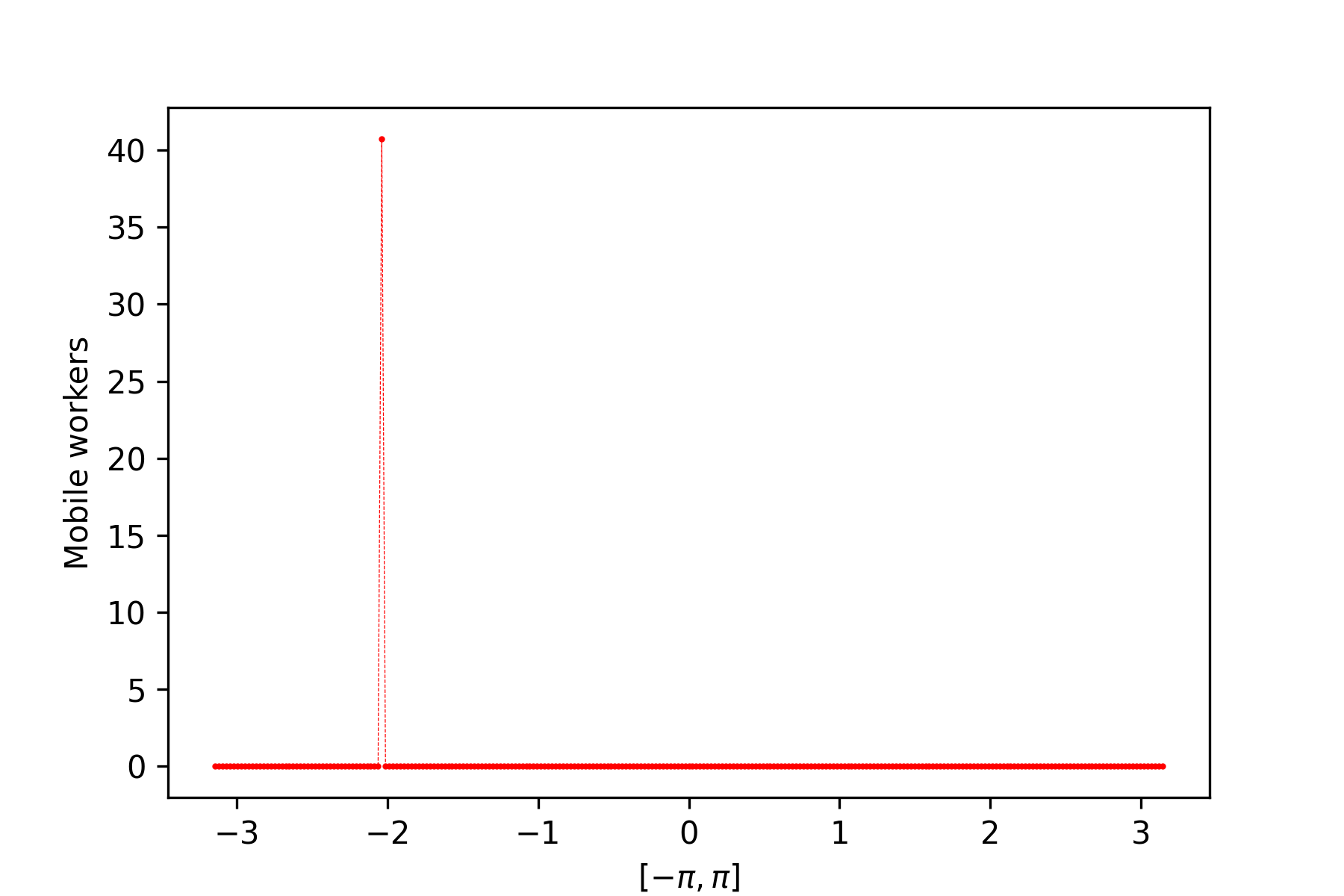}
  \caption{Mobile population}
 \end{subfigure}
 \begin{subfigure}{0.5\columnwidth}
  \centering
  \includegraphics[width=\columnwidth]{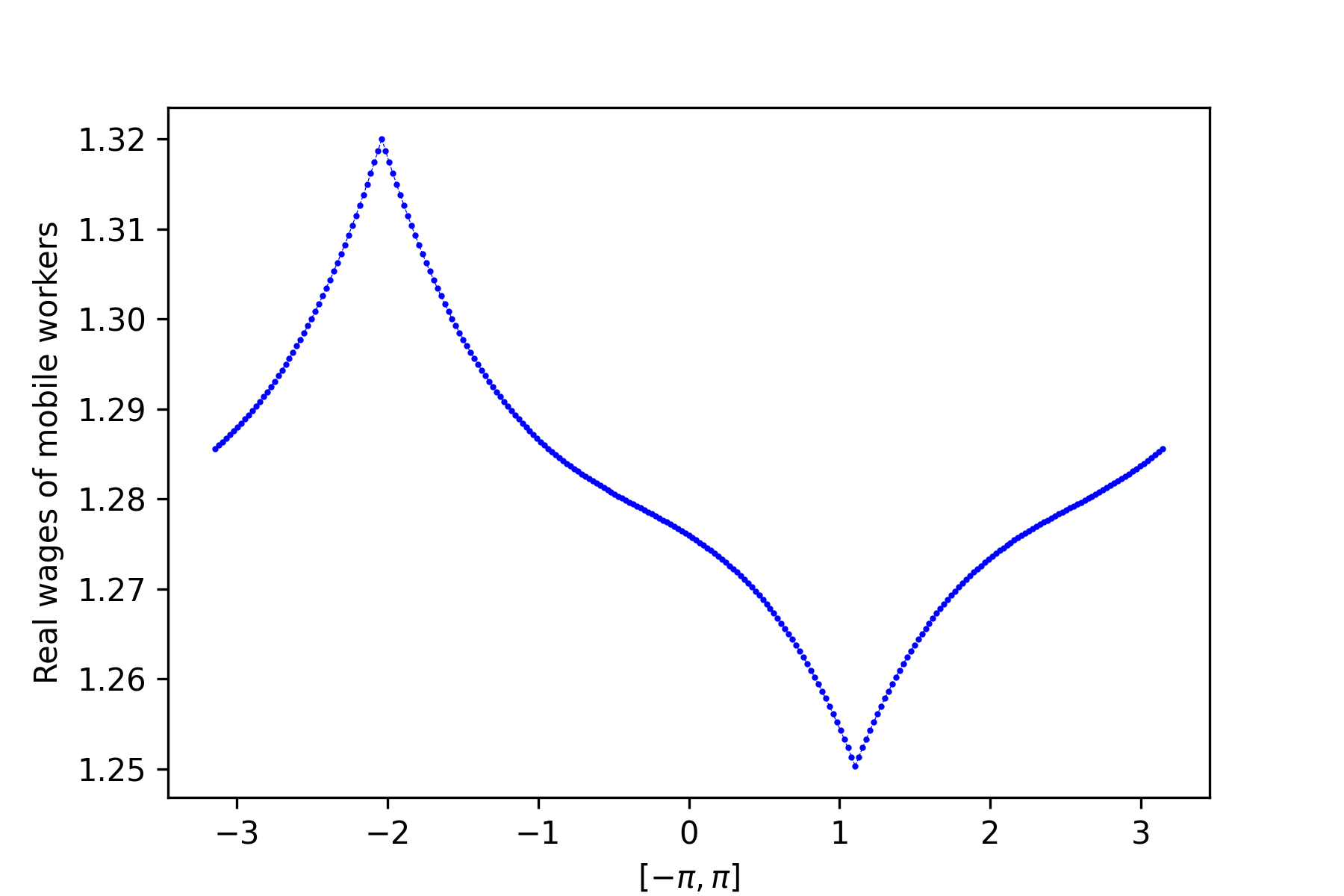}
  \caption{Real wage}
 \end{subfigure}\\
 \caption{Stationary solution for $(\sigma,\tau)=(5.0, 0.05)$}
 \label{fig:QLLUt0p05}
\end{figure}

\begin{figure}[H]
 \begin{subfigure}{0.5\columnwidth}
  \centering
  \includegraphics[width=\columnwidth]{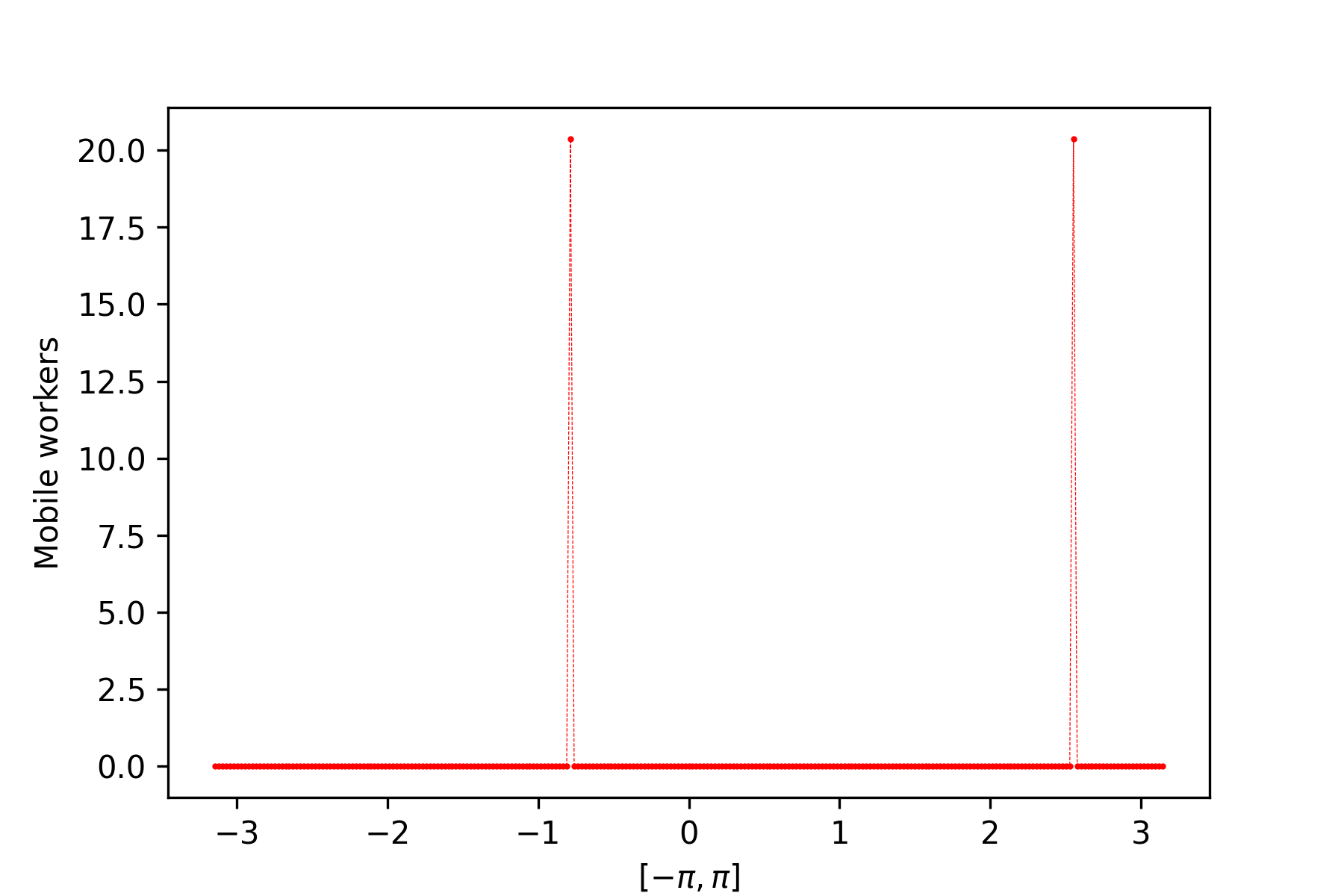}
  \caption{Mobile population}
 \end{subfigure}
 \begin{subfigure}{0.5\columnwidth}
  \centering
  \includegraphics[width=\columnwidth]{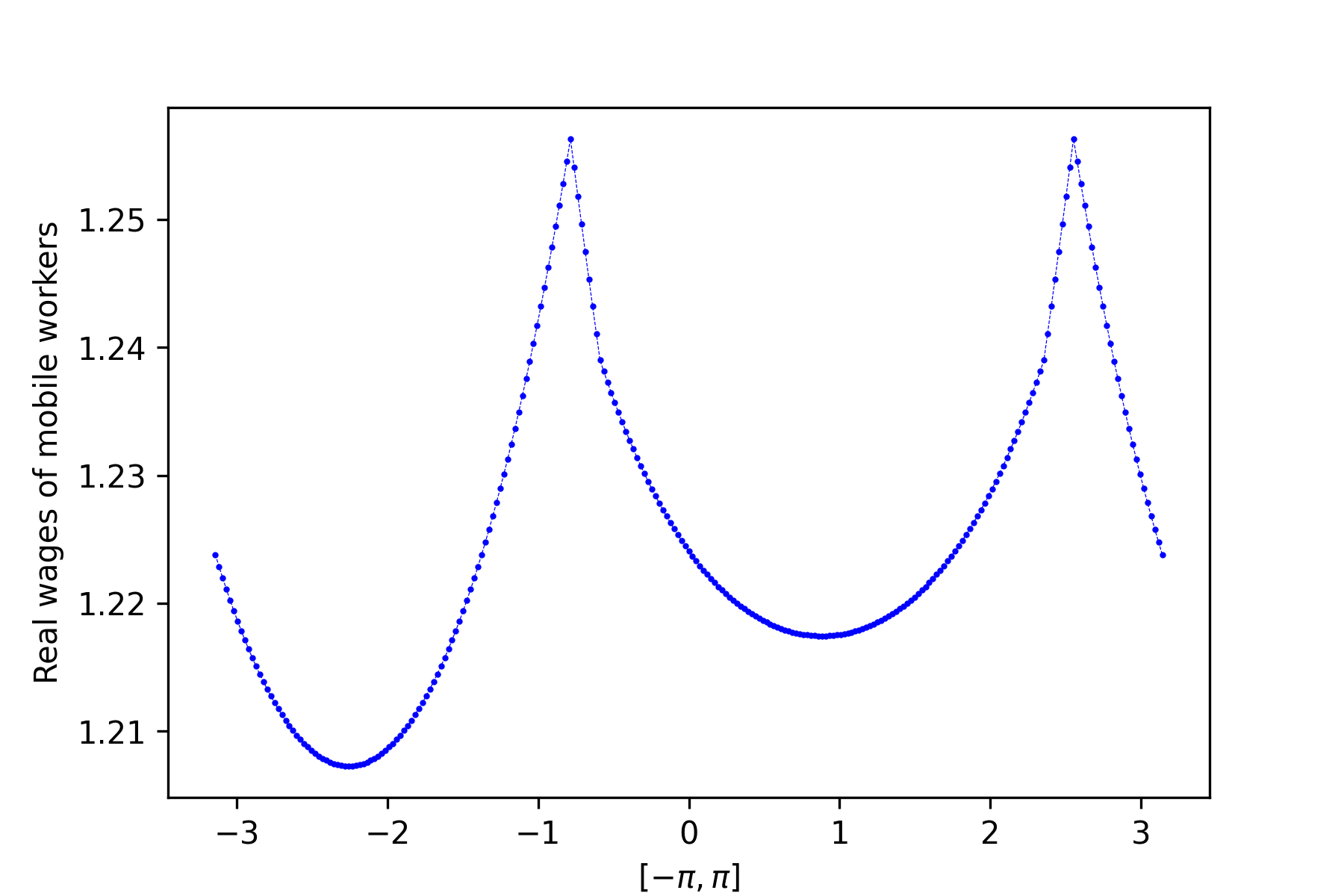}
  \caption{Real wage}
 \end{subfigure}\\
 \caption{Stationary solution for $(\sigma,\tau)=(5.0, 0.1)$}
 \label{fig:QLLUt0p1}
\end{figure}

\begin{figure}[H]
 \begin{subfigure}{0.5\columnwidth}
  \centering
  \includegraphics[width=\columnwidth]{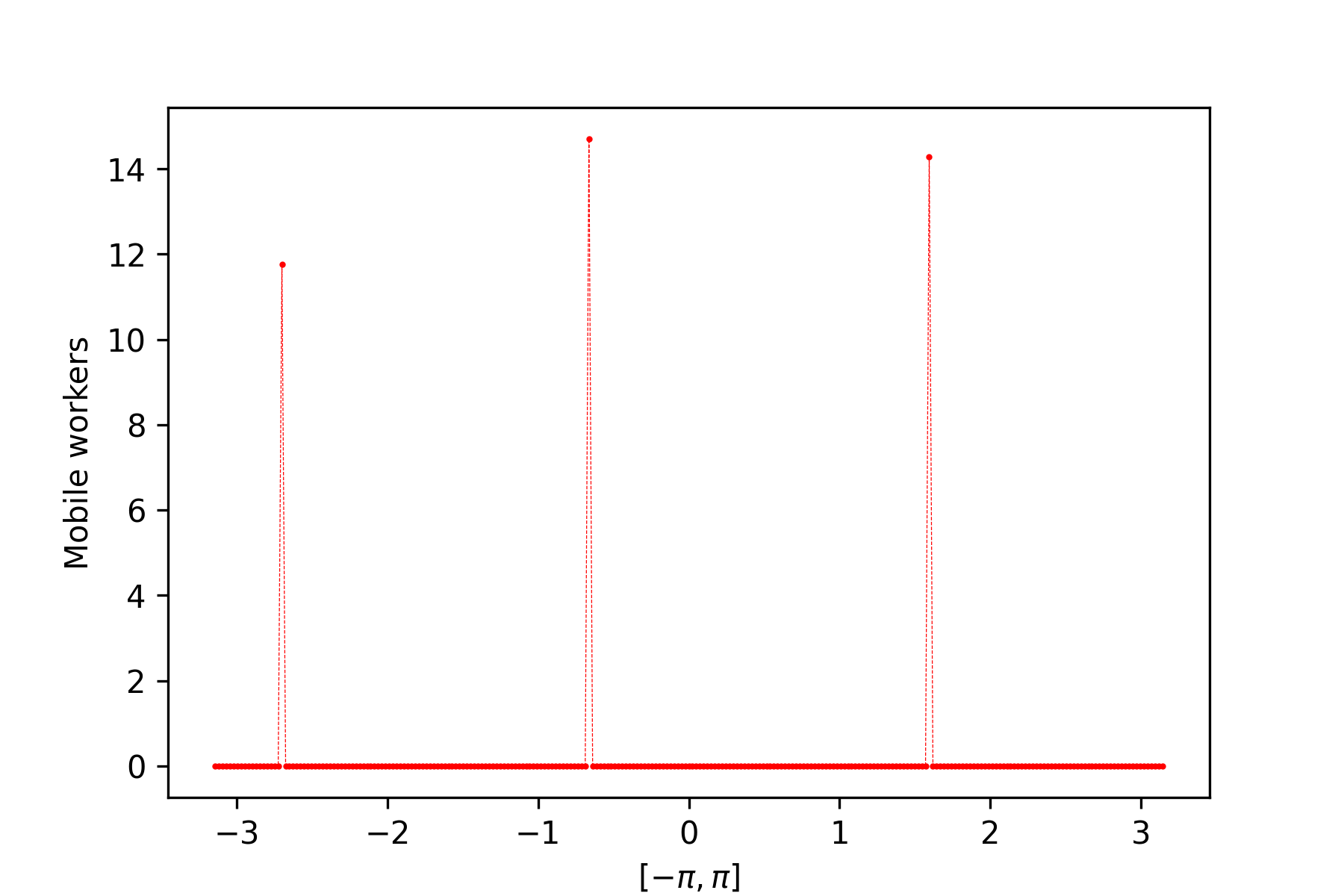}
  \caption{Mobile population}
 \end{subfigure}
 \begin{subfigure}{0.5\columnwidth}
  \centering
  \includegraphics[width=\columnwidth]{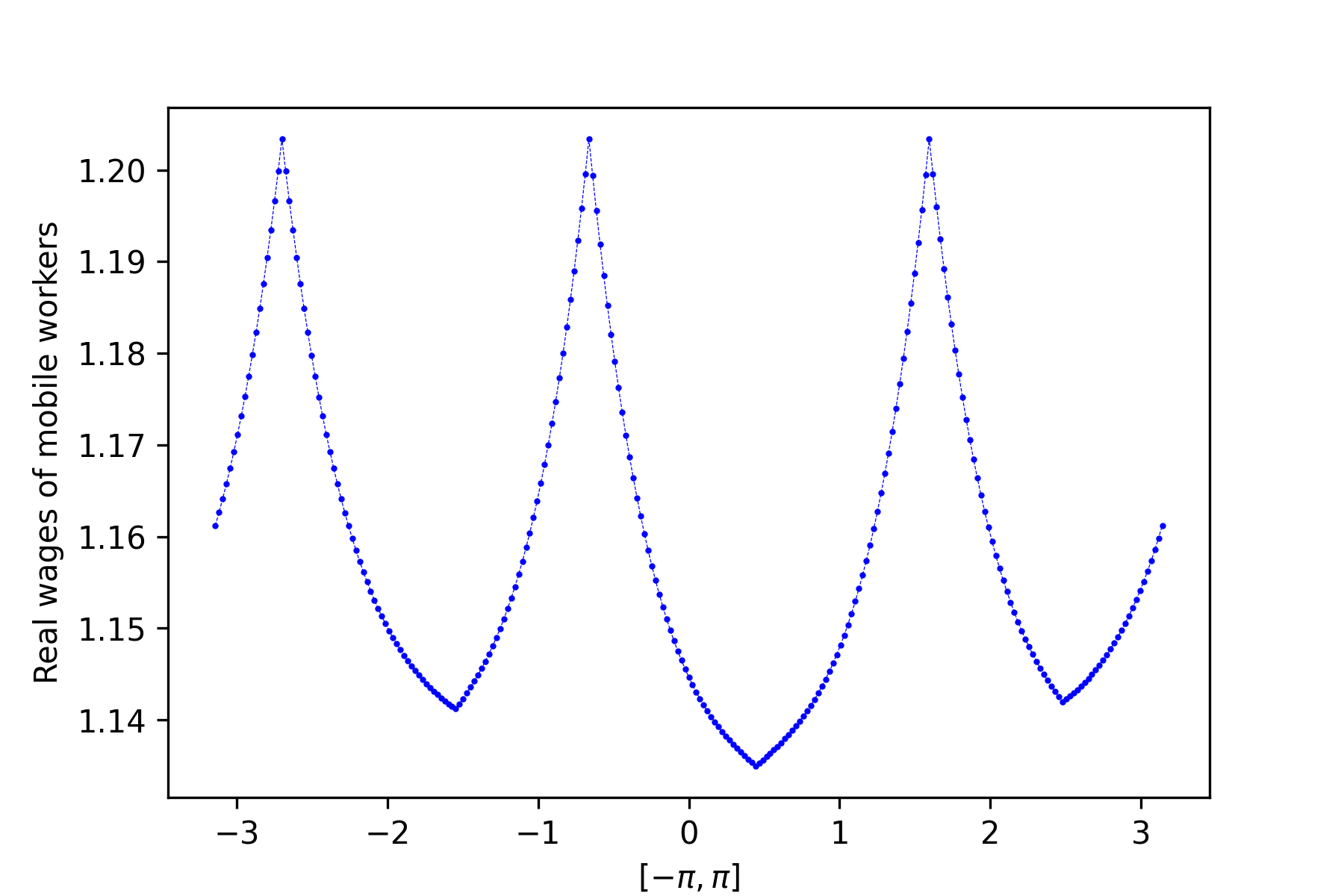}
  \caption{Real wage}
 \end{subfigure}\\
 \caption{Stationary solution for $(\sigma,\tau)=(5.0, 0.2)$}
 \label{fig:QLLUt0p2}
\end{figure}

\begin{figure}[H]
 \begin{subfigure}{0.5\columnwidth}
  \centering
  \includegraphics[width=\columnwidth]{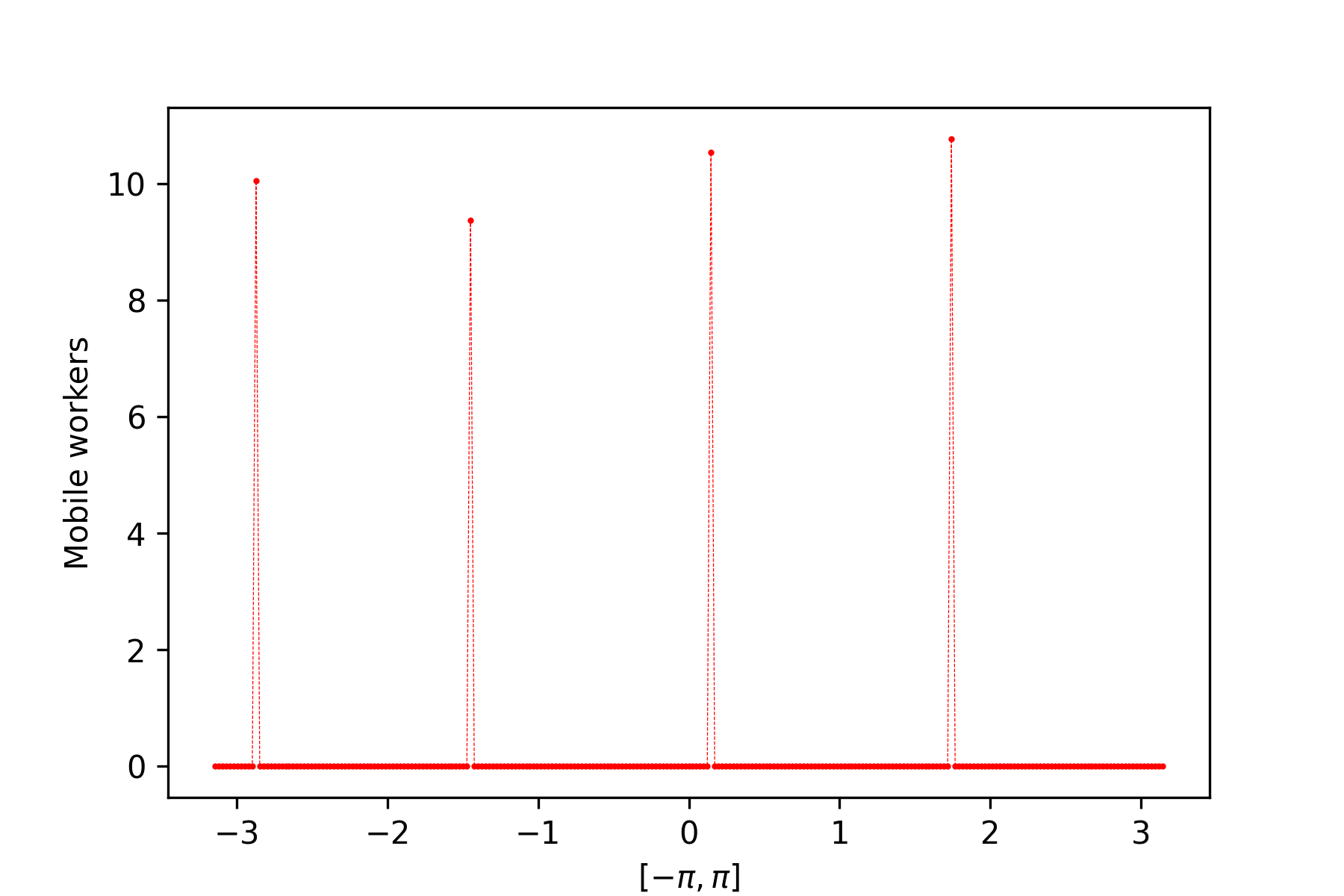}
  \caption{Mobile population}
 \end{subfigure}
 \begin{subfigure}{0.5\columnwidth}
  \centering
  \includegraphics[width=\columnwidth]{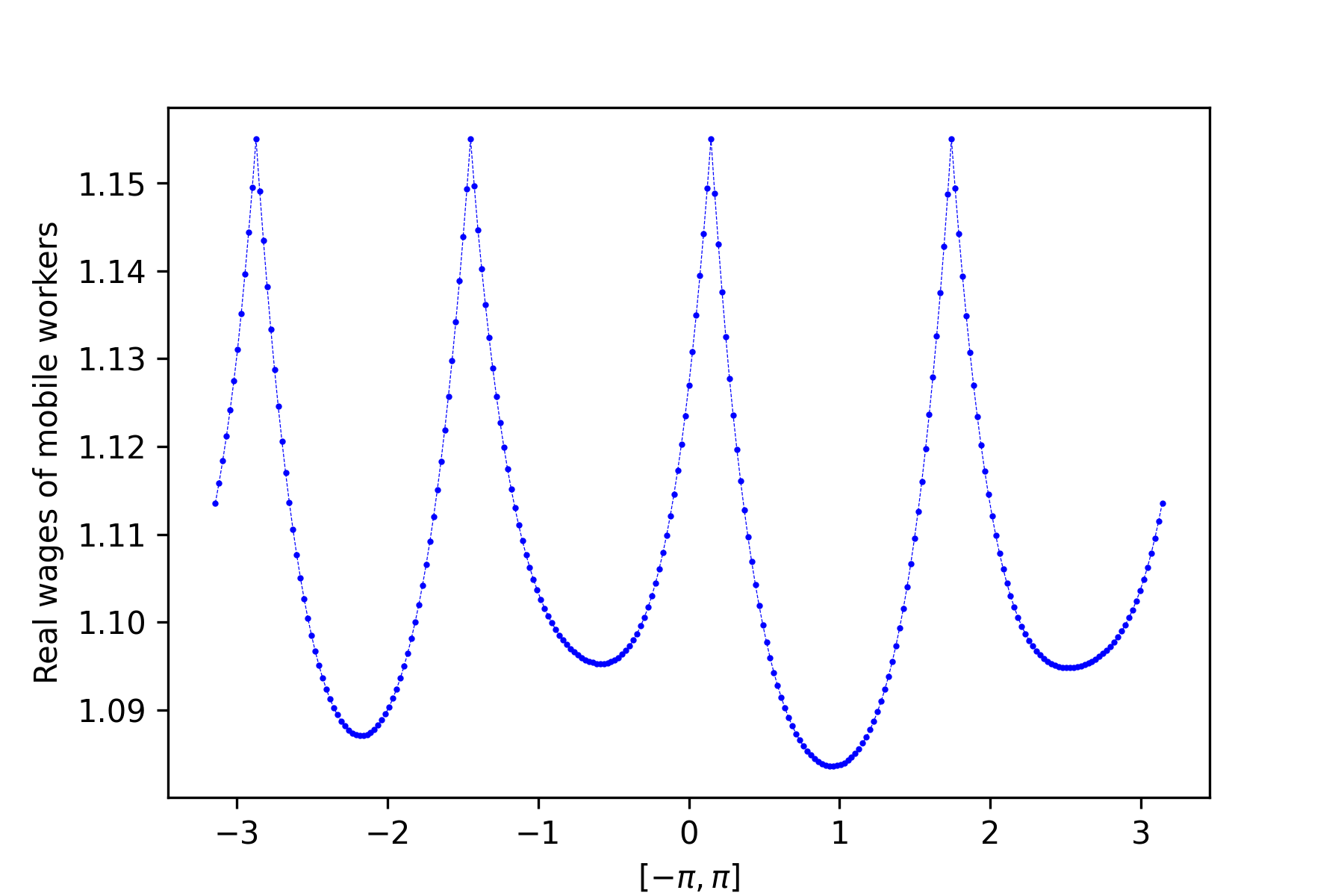}
  \caption{Real wage}
 \end{subfigure}\\
 \caption{Stationary solution for $(\sigma,\tau)=(5.0, 0.3)$}
 \label{fig:QLLUt0p3}
\end{figure}

\begin{figure}[H]
 \begin{subfigure}{0.5\columnwidth}
  \centering
  \includegraphics[width=\columnwidth]{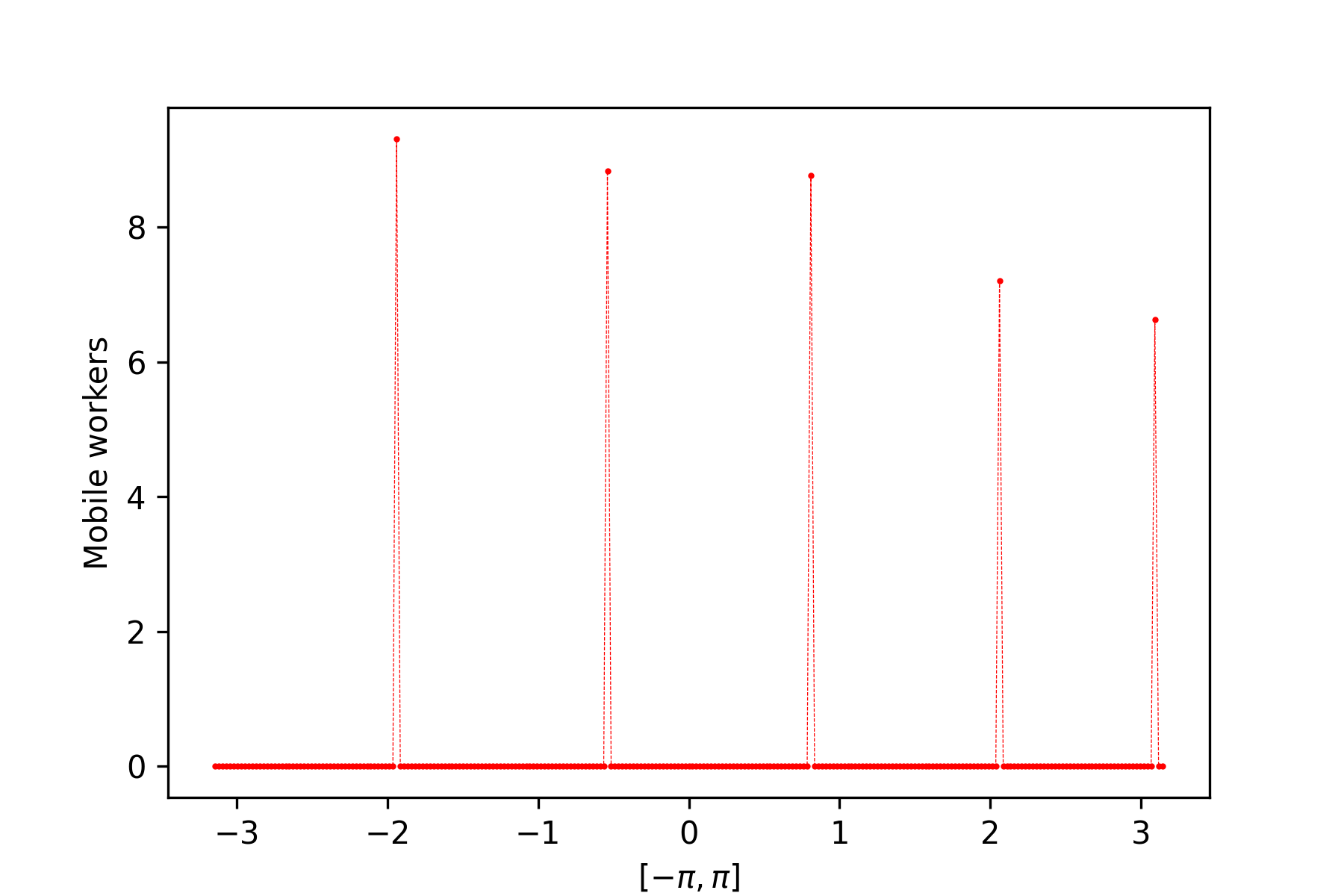}
  \caption{Mobile population}
 \end{subfigure}
 \begin{subfigure}{0.5\columnwidth}
  \centering
  \includegraphics[width=\columnwidth]{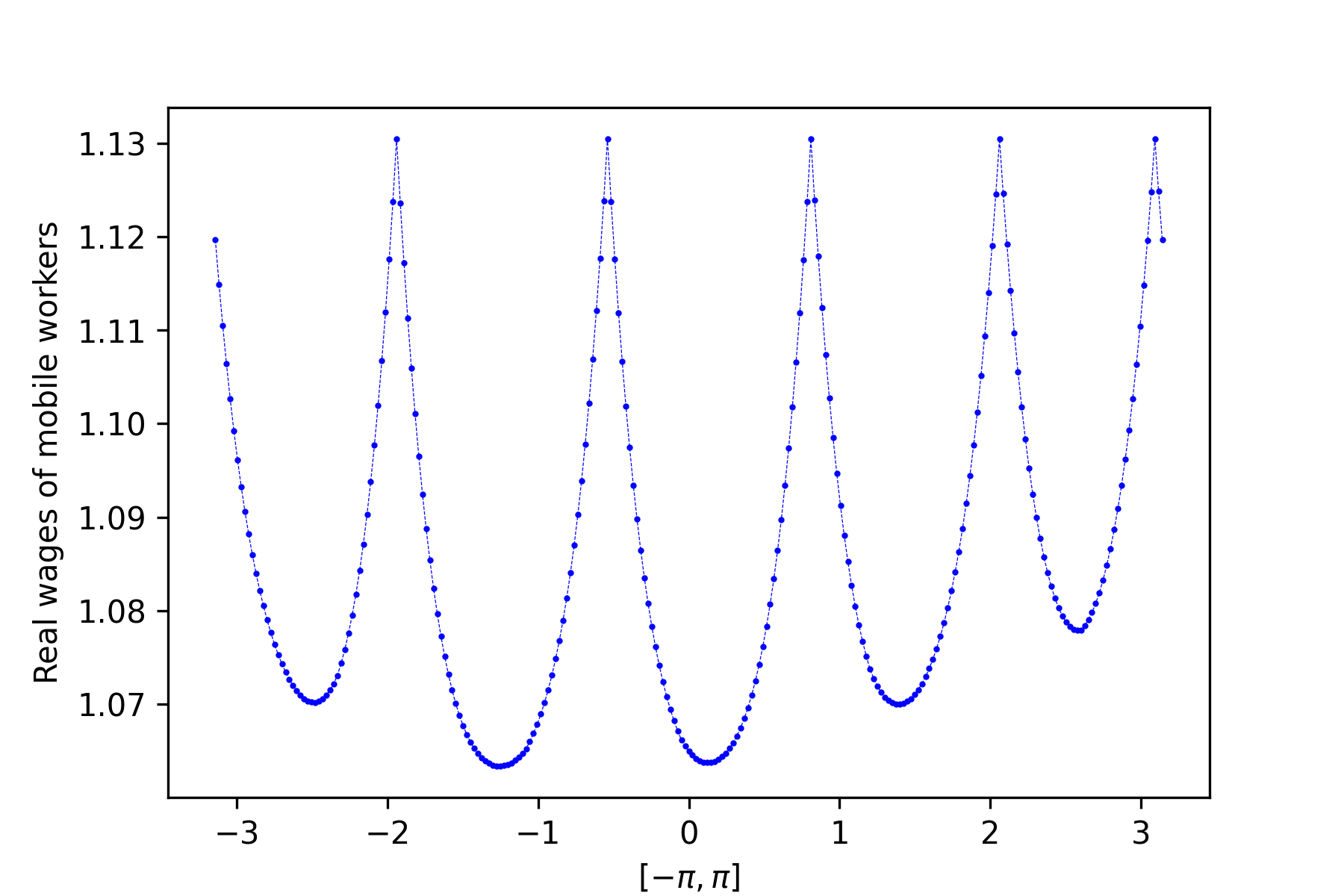}
  \caption{Real wage}
 \end{subfigure}\\
 \caption{Stationary solution for $(\sigma,\tau)=(5.0, 0.35)$}
 \label{fig:QLLUt0p35}
\end{figure}

\begin{figure}[H]
 \begin{subfigure}{0.5\columnwidth}
  \centering
  \includegraphics[width=\columnwidth]{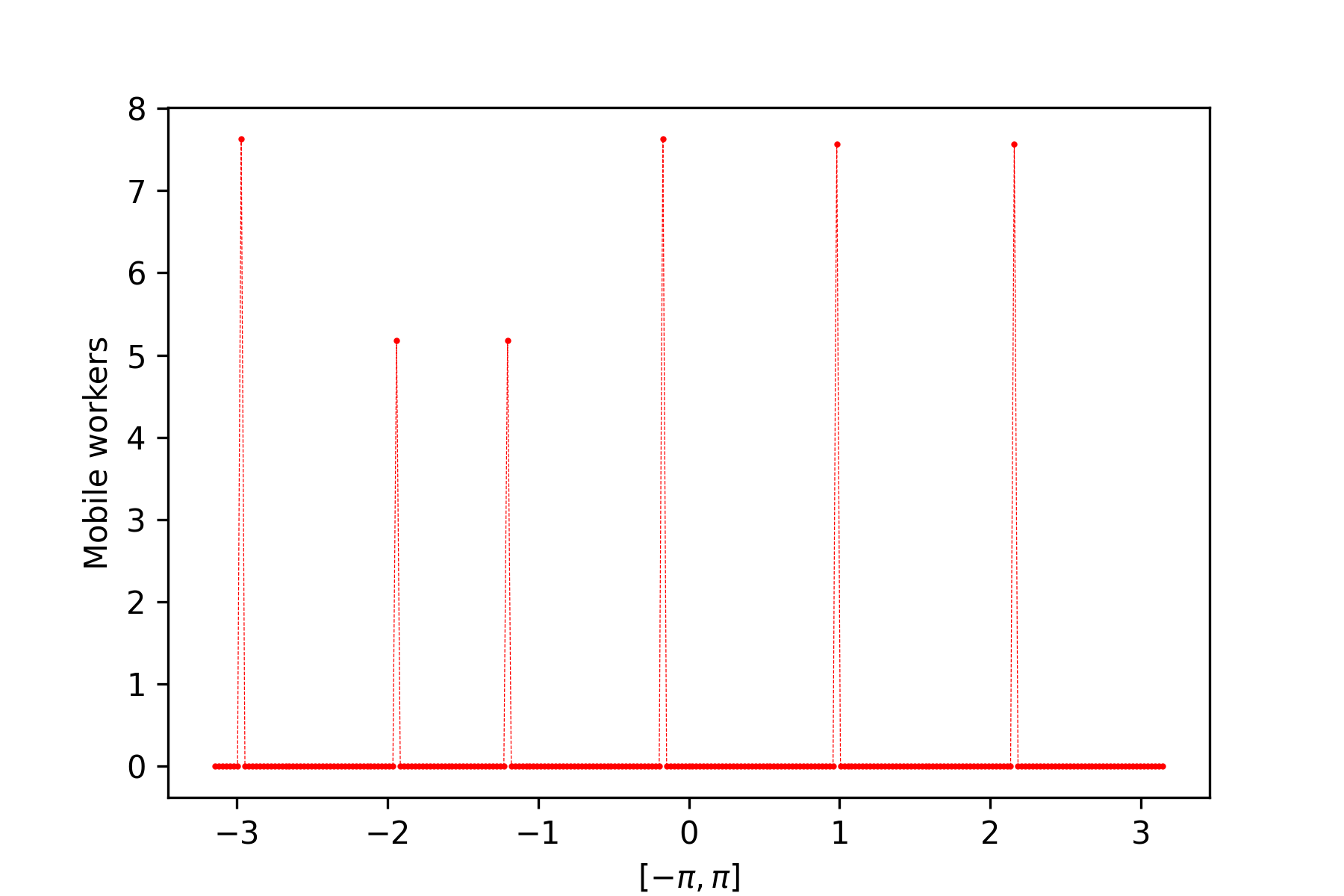}
  \caption{Mobile population}
 \end{subfigure}
 \begin{subfigure}{0.5\columnwidth}
  \centering
  \includegraphics[width=\columnwidth]{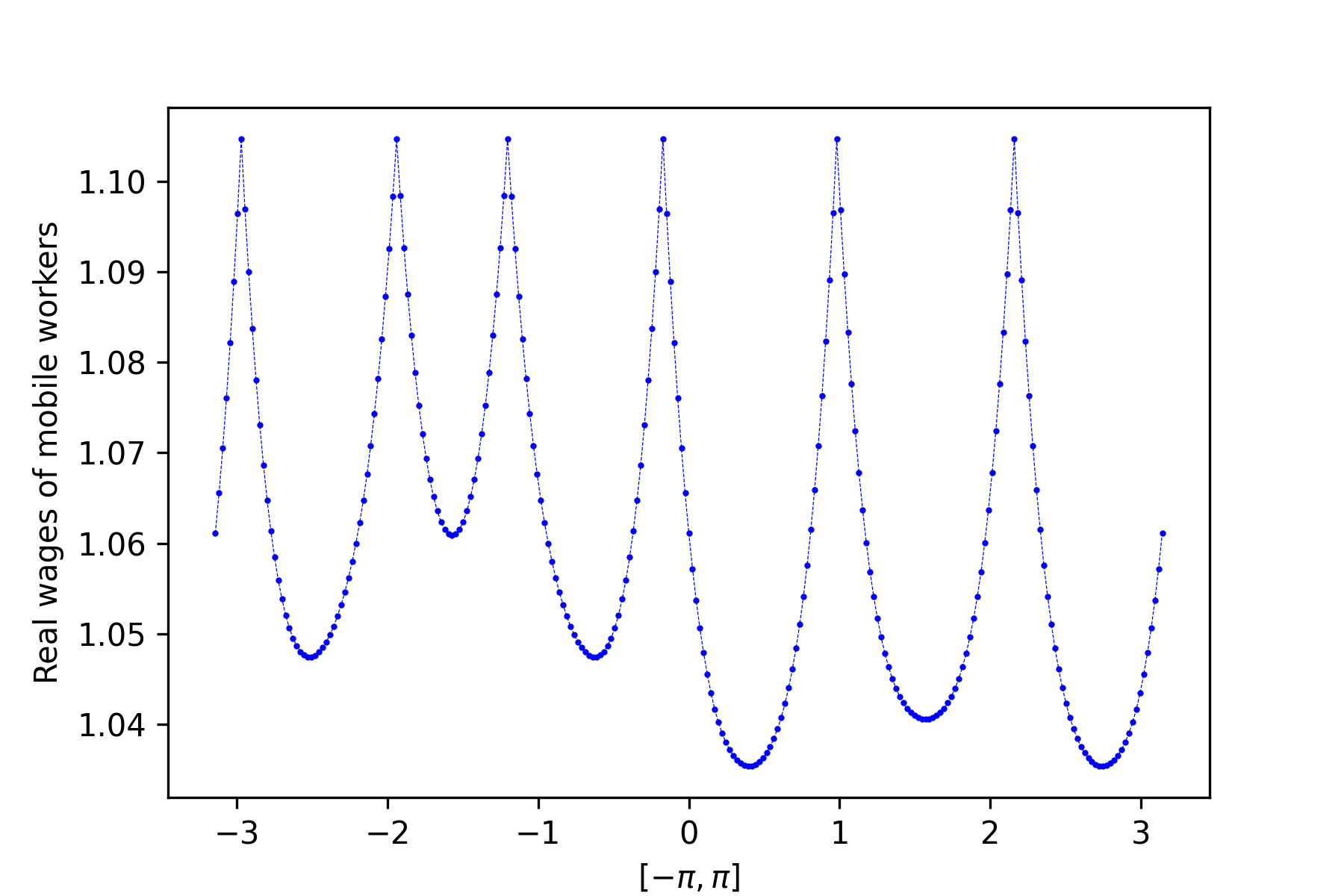}
  \caption{Real wage}
 \end{subfigure}\\
 \caption{Stationary solution for $(\sigma,\tau)=(5.0, 0.42)$}
 \label{fig:QLLUt0p42}
\end{figure}

\subsubsection{Stability of spiky stationary solutions}\label{subsec:comstspsol}

\paragraph{CP model:}
The homogeneous immobile population distribution \eqref{CPphibar} is still assumed. In the case of the uniform spiky solution
\begin{equation}
\bar{\lambda}(x)=\bar{\lambda}\sum_{i=1}^R\delta(x-x_i),\quad\text{where}\quad\bar{\lambda} = \frac{1}{R},\label{CPspikelambdabar}
\end{equation}
the price index and the nominal wage at the equidistantly located centers $i=1,2,\cdots,R$ are give by
\begin{align}
&\bar{G} = \left[\bar{\lambda}\bar{w}^{1-\sigma}\sum_{j=1}^R e^{-\alpha d(x_i, x_j)}\right]^{\frac{1}{1-\sigma}},\label{CPspikeGbar}\\
&\bar{w} =  \frac{\bar{\phi}}{\bar{\lambda}}\int_S\frac{e^{-\alpha d(x_i, y)}}{\sum_{j=1}^R e^{-\alpha d(x_j, y)}}dy,\label{CPspikewbar}
\end{align}
respectively. Linearizing the CP model around the stationary states \eqref{CPphibar}, \eqref{CPspikelambdabar}, \eqref{CPspikeGbar}, and \eqref{CPspikewbar}, we obtain the $R\times R$ matrix
\begin{equation}\label{CPdefinitionGamma}
\Gamma:= v\bar{\lambda}\bar{G}^{-\mu}\left[(I-C)(I-A)^{-1}B-D\right],
\end{equation}
which corresponds to \eqref{matGammaspike} for the FE model. In the right-hand side of \eqref{CPdefinitionGamma}, the matrices $A$, $B$, $C$, and $D$ are defined as
\begin{align}
&A_{ij} := 
\frac{\mu\bar{\lambda}\bar{G}^{\sigma-1}\bar{w}^{1-\sigma}}{\sigma}e^{-\alpha|x_i-x_j|}\\
&\hspace{10mm}+\frac{\mu(\sigma-1)\bar{\lambda}^2\bar{w}^{2-2\sigma}\bar{G}^{2\sigma-2}}{\sigma}\sum_ke^{-\alpha|x_k-x_j|}e^{-\alpha|x_k-x_i|}\\
&\hspace{10mm}+\frac{(1-\mu)(\sigma-1)\bar{\phi}\bar{\lambda}^{-1}\bar{w}^{-1}}{\sigma}\int_S\frac{e^{-\alpha|y-x_j|}e^{-\alpha|y-x_i|}}{\left[\sum_ke^{-\alpha|y-x_k|}\right]^2}dy,\\
&B_{ij} := \frac{\mu\bar{w}^{2-\sigma}\bar{G}^{\sigma-1}}{\sigma}e^{-\alpha|x_i-x_j|}\\
&\hspace{10mm}-\frac{\mu\bar{\lambda}\bar{w}^{3-2\sigma}\bar{G}^{2\sigma-2}}{\sigma}\sum_ke^{-\alpha|x_k-x_j|}e^{-\alpha|x_k-x_i|}\\
&\hspace{10mm}-\frac{(1-\mu)\bar{\phi}\bar{\lambda}^{-2}}{\sigma}\int_S\frac{e^{-\alpha|y-x_j|}e^{-\alpha|y-x_i|}}{\left[\sum_ke^{-\alpha|y-x_k|}\right]^2}dy,\\
&C_{ij} := \mu\bar{\lambda}\bar{w}^{1-\sigma}\bar{G}^{\sigma-1}e^{-\alpha|x_i-x_j|},\\
&D_{ij} := -\frac{\mu}{\sigma-1}\bar{w}^{2-\sigma}\bar{G}^{\sigma-1}e^{-\alpha|x_i-x_j|},
\end{align}
respectively. Fig.~\ref{fig:CP-spiky-heatmaps} shows heatmaps of the maximum eigenvalues of $\Gamma$ plotted in the $(\tau, \sigma)$-plane, obtained by specifically calculating \eqref{CPdefinitionGamma} for $R=2$ and $R=3$, respectively. The heatmaps are similar to those for the FE model (Sections \ref{subsubsec:heatmapR2} and \ref{subsubsec:heatmapR3}), and it can be confirmed that the uniform spiky solution is unstable in the lower-left area of the downward-sloping critical curve, while it is stable in the upper-right area.

\begin{figure}[H]
 \begin{subfigure}{0.5\columnwidth}
  \centering
  \includegraphics[width=\columnwidth]{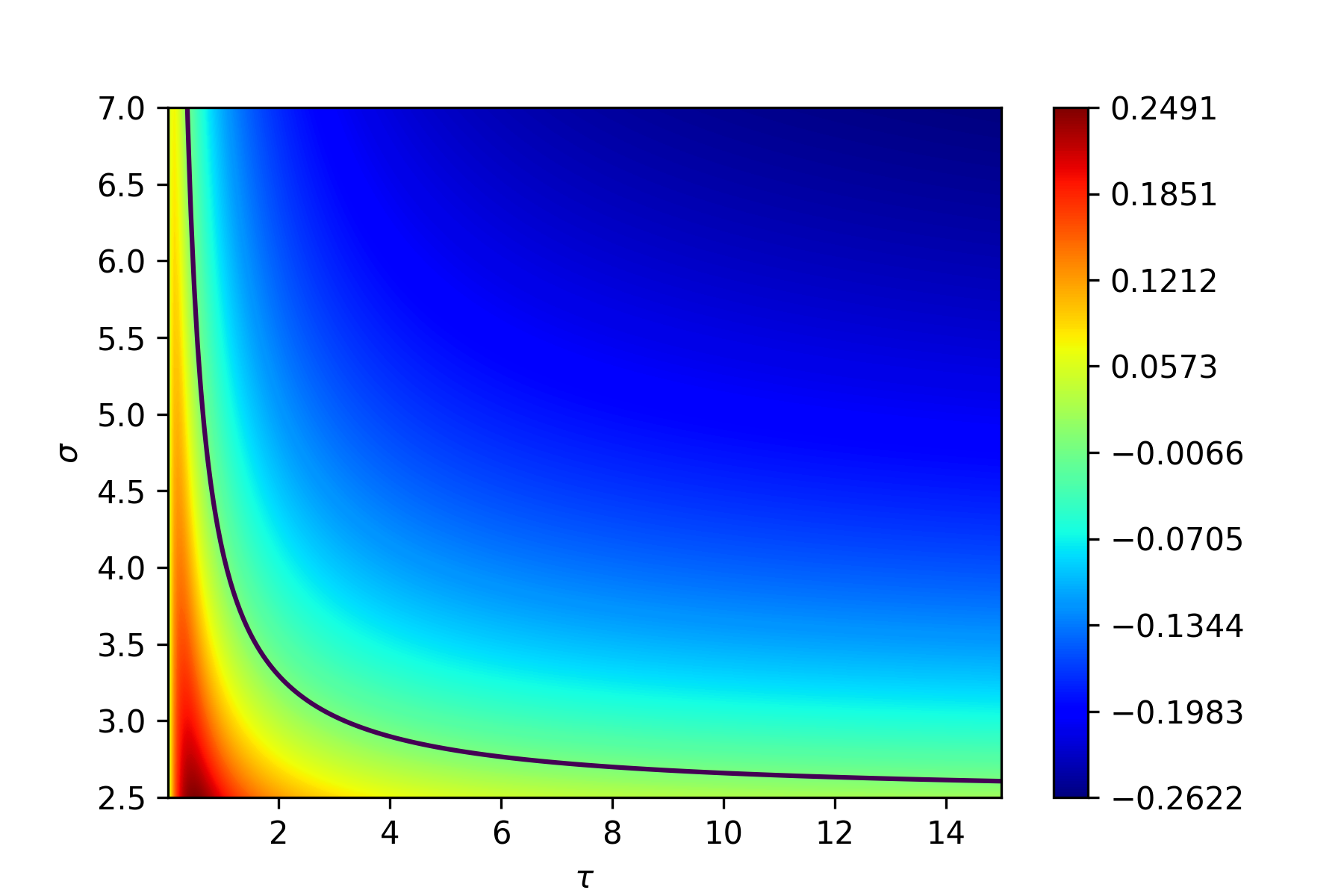}
  \caption{$R=2$}\label{subfig:CP-spiky-heatmapR2}
 \end{subfigure}
 \begin{subfigure}{0.5\columnwidth}
  \centering
  \includegraphics[width=\columnwidth]{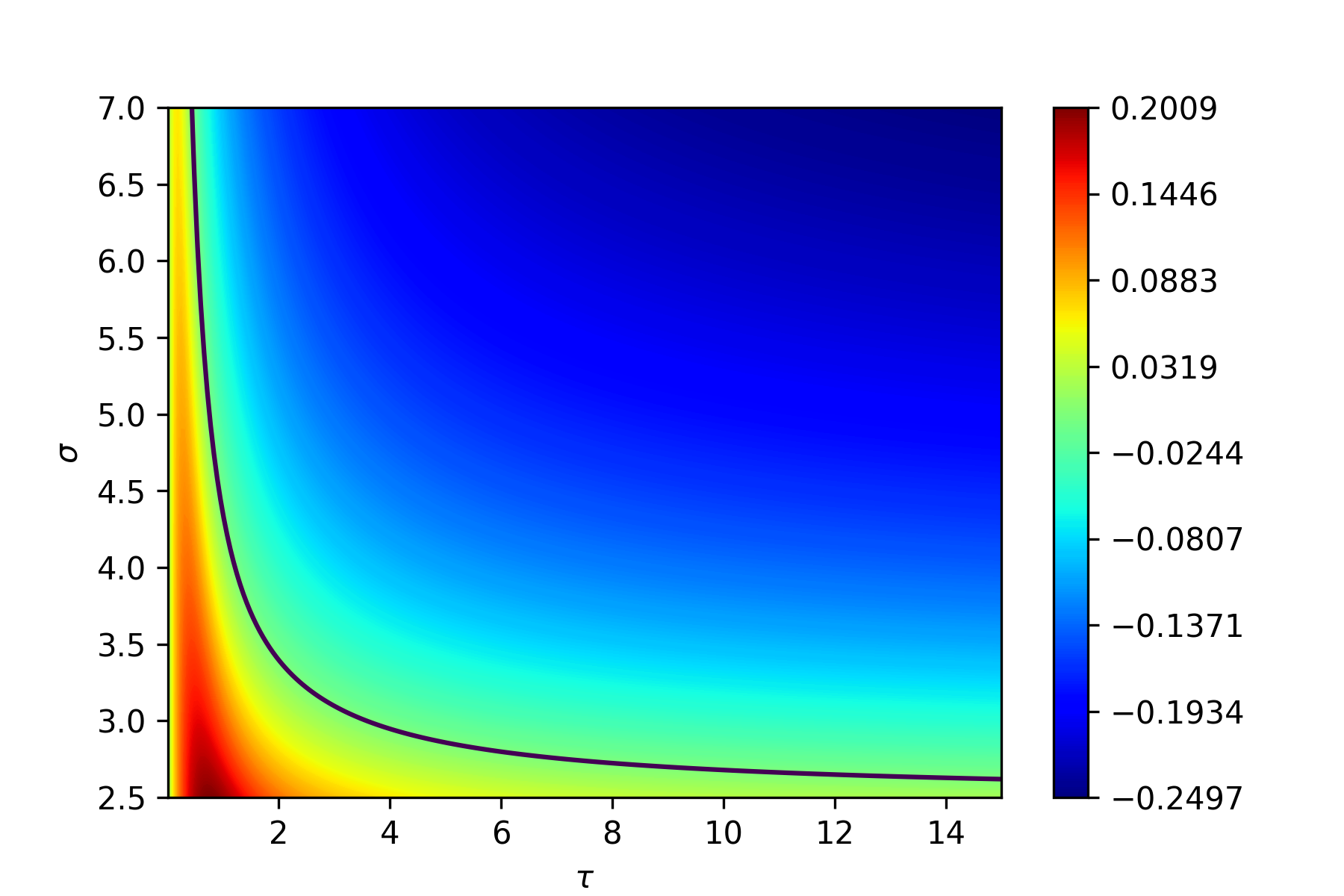}
  \caption{$R=3$}\label{subfig:CP-spiky-heatmapR3}
 \end{subfigure}\\
 \caption{Maximal eigenvalue of the matrix $\Gamma$ for the CP model}
 \label{fig:CP-spiky-heatmaps}
\end{figure}

\paragraph{QLLU model:}
The homogeneous immobile population distribution \eqref{QLLUphibar} is still assumed. In the case of the uniform spiky solution
\begin{equation}\label{QLLUspikelambdabar}
\bar{\lambda}(x)=\bar{\lambda}\sum_{i=1}^R\delta(x-x_i),\quad\text{where}\quad\bar{\lambda} = \frac{\Lambda}{R},
\end{equation}
the price index and the nominal wage at the equidistantly located centers $i=1,2,\cdots,R$ are give by
\begin{align}
&\bar{G} = \left[\frac{\bar{\lambda}}{F}\sum_{j=1}^R e^{-\alpha d(x_i, x_j)}\right]^{\frac{1}{1-\sigma}},\label{QLLUspikeGbar}\\
&\bar{w} =   \frac{\mu}{\sigma}\left[1+\frac{\bar{\phi}}{\bar{\lambda}}\int_S\frac{e^{-\alpha d(x_i, y)}}{\sum_{j=1}^R e^{-\alpha d(y, x_j)}}dy\right],\label{QLLUspikewbar}
\end{align}
respectively. Linearizing the QLLU model around the stationary states \eqref{QLLUphibar}, \eqref{QLLUspikelambdabar}, \eqref{QLLUspikeGbar}, and \eqref{QLLUspikewbar}, we obtain the $R\times R$ matrix
\begin{equation}\label{QLLUdefinitionGamma}
\Gamma:= v\bar{\lambda}\left(A+B\right),
\end{equation}
which corresponds to \eqref{matGammaspike} for the FE model. In the right-hand side of \eqref{QLLUdefinitionGamma}, the matrices $A$ and $B$ are defined as
\begin{align}
&A_{ij} := \frac{\mu}{\sigma F}\bar{G}^{\sigma-1}e^{-\alpha|x_i-x_j|}\\
&\hspace{15mm}- \frac{\mu}{\sigma F^2}\bar{\lambda}\bar{G}^{2\sigma-2}\sum_k e^{-\alpha|x_i-x_k|}e^{-\alpha|x_k-x_j|}\\
&\hspace{20mm}- \frac{\mu\bar{\phi}}{\sigma \bar{\lambda}^2}\int_S\frac{e^{-\alpha|y-x_i|}e^{{-\alpha|y-x_j|}}}{\left[\sum_ke^{-\alpha|y-x_k|}\right]^2}dy,\\
&B_{ij} := \frac{\mu}{(\sigma-1)F}\bar{G}^{\sigma-1}e^{-\alpha|x_i-x_j|},
\end{align}
respectively. Fig.~\ref{fig:QLLU-spiky-heatmaps} shows heatmaps of the maximum eigenvalues of $\Gamma$ plotted in the $(\tau, \sigma)$-plane, obtained by specifically calculating \eqref{QLLUdefinitionGamma} for $R=2$ and $R=3$, respectively. The heatmaps are similar to those for the FE model (Sections \ref{subsubsec:heatmapR2} and \ref{subsubsec:heatmapR3}), and it can be confirmed that the uniform spiky solution is unstable in the lower-left area of the downward-sloping critical curve, while it is stable in the upper-right area.

\begin{figure}[H]
 \begin{subfigure}{0.5\columnwidth}
  \centering
  \includegraphics[width=\columnwidth]{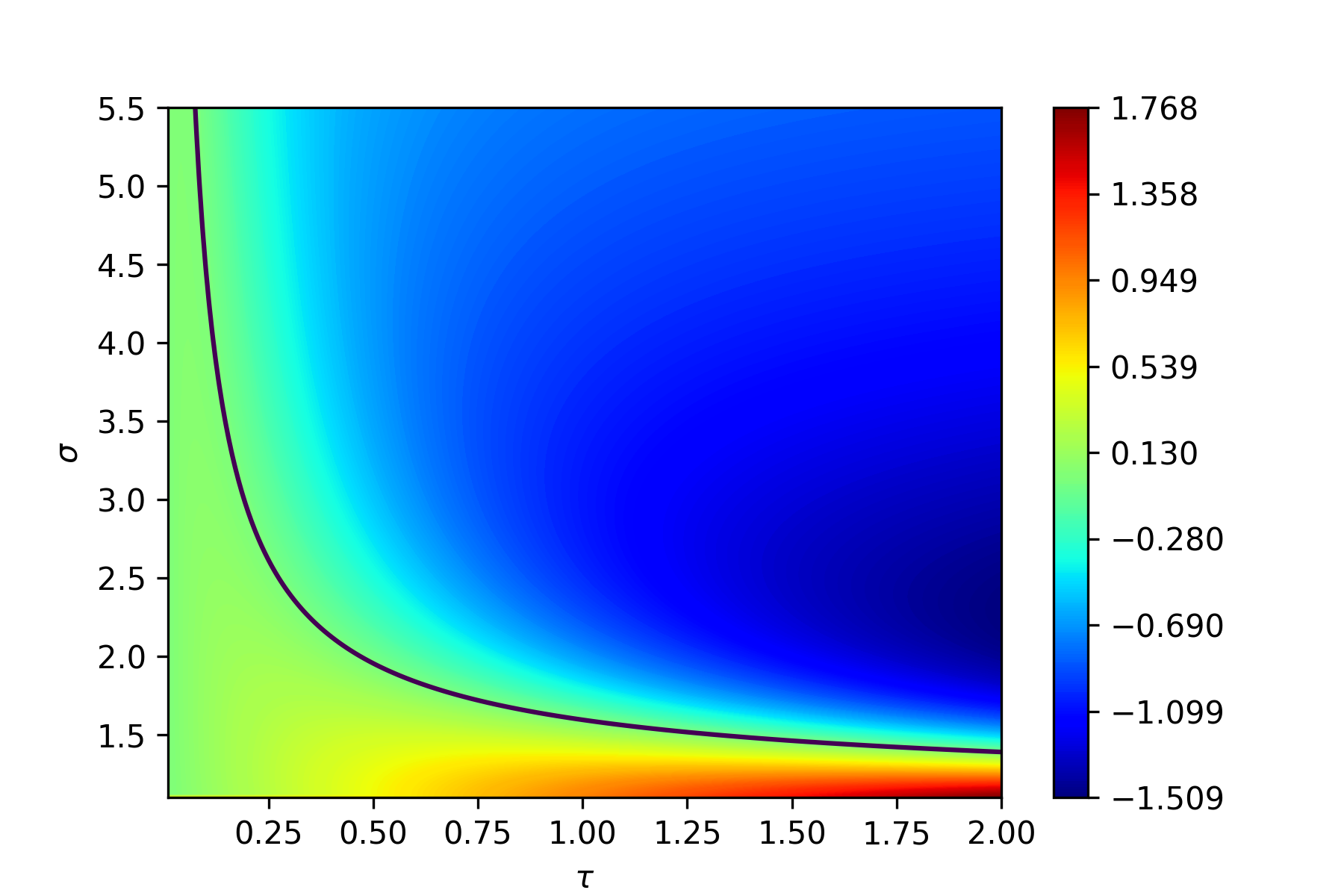}
  \caption{$R=2$}\label{subfig:QLLU-spiky-heatmapR2}
 \end{subfigure}
 \begin{subfigure}{0.5\columnwidth}
  \centering
  \includegraphics[width=\columnwidth]{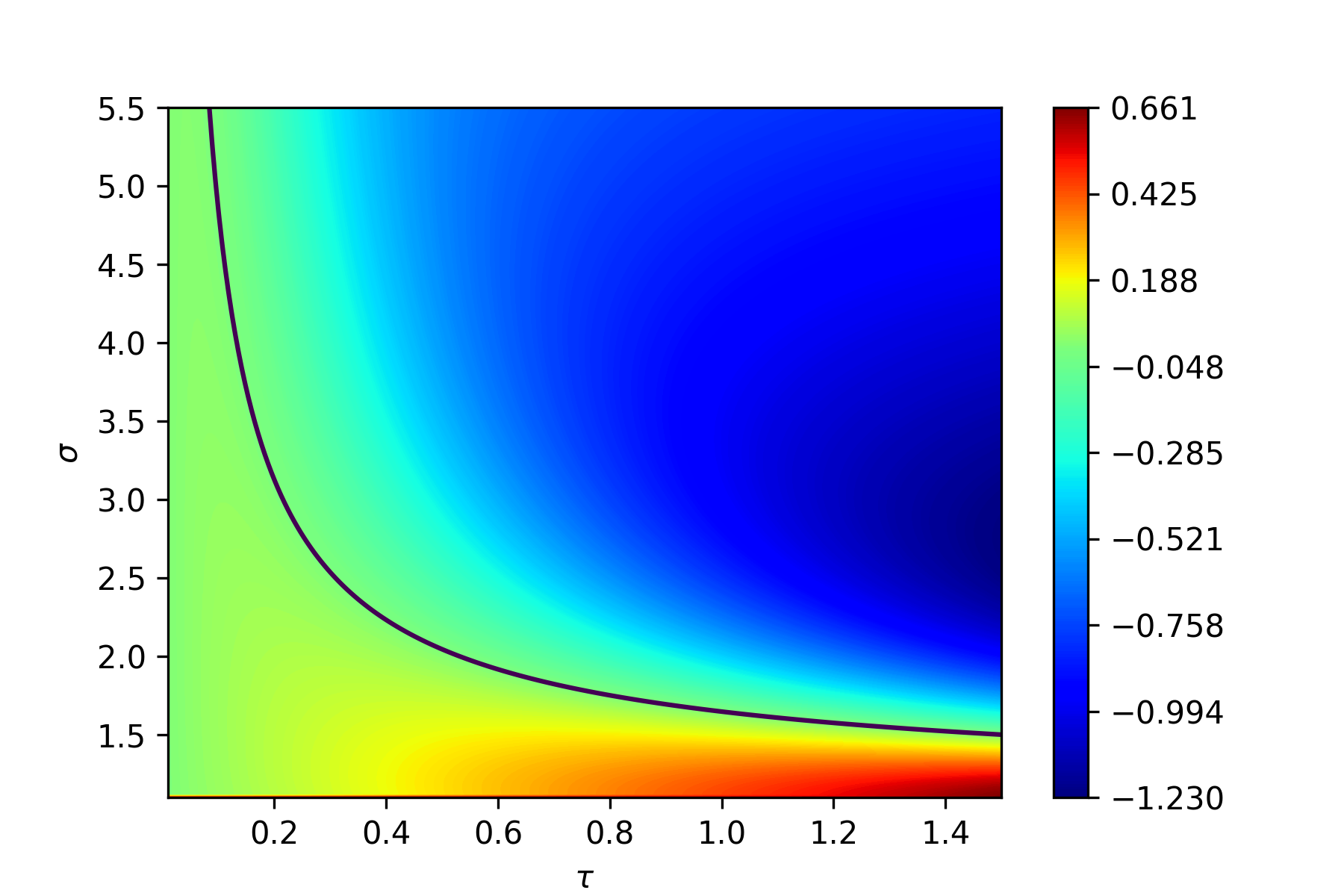}
  \caption{$R=3$}\label{subfig:QLLU-spiky-heatmapR3}
 \end{subfigure}\\
 \caption{Maximal eigenvalue of the matrix $\Gamma$ for the QLLU model}
 \label{fig:QLLU-spiky-heatmaps}
\end{figure}

\vspace{10mm}
\noindent
{\large {\bf Statements and Declarations}}

\vspace{3mm}
\noindent
{\small {\bf Competing Interests}:\\
The author has no relevant financial or non-financial interests to disclose.}

\vspace{3mm}
\noindent
{\small{\bf Ethical Approval}: Not applicable.}

\vspace{3mm}
\noindent
{\small{\bf Authors' contributions}: Single author (Kensuke Ohtake).}

\vspace{3mm}
\noindent
{\small{\bf Funding}:\\
The author declares that no funds, grants, or other support were received during the preparation of this manuscript.}

\vspace{3mm}
\noindent
{\small {\bf Availability of data and materials}: All the source codes for the numerical computations in this paper are available on \url{https://github.com/k-ohtake/contspace-femodel}.}

\bibliographystyle{econ-aea}

\end{document}